\documentclass[11pt]{article}
\usepackage{amsmath,amsfonts,amssymb,amsthm}
\usepackage{fullpage}
\usepackage{algorithm}
\usepackage{algorithmic}
\usepackage{epstopdf}
\usepackage{graphicx}
\usepackage{hyperref}
\usepackage{cleveref}
\usepackage{cite}
\usepackage{color}
\usepackage[sans]{dsfont}
\usepackage{mathtools}
\usepackage{xspace}

\newtheorem{problem}{Problem}
\newtheorem{theorem}{Theorem}[section]
\newtheorem{corollary}[theorem]{Corollary}

\newtheorem{heuristic algorithm}{Heuristic Algorithm}
\newtheorem{lemma}[theorem]{Lemma}
\newtheorem{definition}[theorem]{Definition}
\newtheorem{claim}[theorem]{Claim}

\newtheorem{observation}[theorem]{Observation}

\newtheorem{conjecture}{Conjecture}

\newtheorem{approximation algorithm}{Approximation Algorithm}
\newcommand{\poly}{\textnormal{poly}}

\newcommand{\ceil}[1]{\ensuremath{\left\lceil#1\right\rceil}}
\newcommand{\floor}[1]{\ensuremath{\left\lfloor#1\right\rfloor}}

\newcommand{\cset}{\mathcal{C}}
\newcommand{\fset}{\mathcal{F}}
\newcommand{\tset}{\mathcal{T}}
\newcommand{\set}[1]{\{#1\}}

\newcommand{\aset}{{\mathcal{A}}}
\newcommand{\bset}{{\mathcal{B}}}
\newcommand{\iset}{{\mathcal{I}}}
\newcommand{\rset}{{\mathcal{R}}}

\newcommand{\sset}{{\mathcal{S}}}
\newcommand{\hset}{{\mathcal{H}}}

\newcommand{\event}{{\cal{E}}}
\newcommand{\alg}{{\cal{A}}}
\newcommand{\vol}{\operatorname{vol}}
\newcommand{\yi}{{\sc Yes-Instance}\xspace}
\renewcommand{\ni}{{\sc No-Instance}\xspace}
\renewcommand{\phi}{\varphi}
\newcommand{\optcro}[1]{\mathsf{OPT}_{\mathsf{cr}}(#1)}

\newenvironment{proofof}[1]{\noindent{\bf Proof of #1.}}%
{\hfill\stopproof}
\renewenvironment{proof}{\par \smallskip{\bf Proof:}}{\hfill\stopproof}

\newcommand{\expect}[2][]{\text{\bf E}_{#1}\left [#2\right]}
\newcommand{\prob}[2][]{\text{\bf Pr}_{#1}\left [#2\right]}
\newcommand{\eps}{\varepsilon}
\newcommand{\val}{\textsf{val}}

\newenvironment{properties}[2][0]
{
	\begin{enumerate} \setcounter{enumi}{#1}}{\end{enumerate}}
	
\newcounter{note}[section]

\newcommand{\mynote}[2][red]{\textcolor{#1}{\sc\bf{[#2]}}}
\newcommand{\tnote}[2][blue]{\textcolor{#1}{\sc\bf{[#2]}}}

\def\stopproof{\square}
\def\square{\vbox{\hrule height.2pt\hbox{\vrule width.2pt height5pt \kern5pt
\vrule width.2pt} \hrule height.2pt}}

\makeatletter
\newtheorem*{rep@theorem}{\rep@title}
\newcommand{\newreptheorem}[2]{%
	\newenvironment{rep#1}[1]{%
		\def\rep@title{#2 \ref{##1}}%
		\begin{rep@theorem}}%
		{\end{rep@theorem}}}
\makeatother

\newreptheorem{theorem}{Theorem}

\newcommand{\opt}{\text{OPT}}
\newcommand{\dist}{\operatorname{dist}}

\newcommand{\DkS}{\textnormal{\textsf{Densest $k$-Subgraph}}\xspace}
\newcommand{\DBS}{\textnormal{\textsf{Bipartite Densest $(k_1,k_2)$-Subgraph}}\xspace}
\newcommand{\BDkS}{\textnormal{\sf{Bipartite Densest $(k_1,k_2)$-Subgraph}}\xspace}
\newcommand{\MCN}{\textnormal{\sf{Minimum Crossing Number}}\xspace}
\newcommand{\DkC}{\textnormal{\textsf{Dense $k$-Coloring}}\xspace}
\newcommand{\CSP}{\textnormal{\textsf{CSP}}\xspace}
\newcommand{\GCSP}{\textnormal{\textsf{Gap-CSP}}\xspace}

\newcommand{\WGP}{\textnormal{\textsf{(r,h)-Graph Partitioning}}\xspace}
\newcommand{\WGPwB}{\textnormal{\textsf{(r,h)-Graph Partitioning with Bundles}}\xspace}
\newcommand{\BCS}{\textnormal{\textsf{Maximum Bounded-Crossing Subgraph}}\xspace}
\newcommand{\LDCSP}{\textnormal{\textsf{LD-2CSP}}\xspace}
\newcommand{\MBCS}{\textnormal{\textsf{Maximum Bounded-Crossing Subgraph}}\xspace}
\newcommand{\size}{\operatorname{size}}
\newcommand{\reals}{\mathbb{R}}
\newcommand{\crn}{\mathsf{cr}}

\newcommand{\CrN}{\mathsf{CrN}}
\newcommand{\optcn}{\mathrm{OPT_{MBCS}}}
\newcommand{\optwgp}{\mathrm{OPT_{GP}}}
\newcommand{\pcn}{\mathrm{MBCS}}
\newcommand{\optdks}{\mathrm{OPT_{DkS}}}
\newcommand{\optbdks}{\mathrm{OPT_{BDkS}}}
\newcommand{\pdks}{\mathrm{DkS}}
\newcommand{\pbdks}{\mathrm{BDkS}}
\newcommand{\pwgp}{\mathrm{GP}}
\newcommand{\pcl}{\mathrm{DkC}}
\newcommand{\pdkc}{\mathrm{DkC}}
\newcommand{\optcl}{\mathrm{OPT_{DkC}}}

\newcommand{\DS}{\textsf{BDkS}\xspace}
\newcommand{\BCSP}{\textsf{Bipartite 2-CSP}\xspace}
\newcommand{\algdbs}{\alg_{\sf BDKS}}
\newcommand{\algsep}{\alg_{\sf sep}}

\newcommand{\conj}{{\sf{Low-Degree CSP Conjecture}}\xspace}

\newcommand{\tH}{\tilde H}

\renewcommand{\P}{\mbox{\sf P}}
\newcommand{\NP}{\mbox{\sf NP}}

\newcommand{\DTIME}{\mbox{\sf DTIME}}
\newcommand{\BPTIME}{\mbox{\sf BPTIME}}

\newcommand{\NDP}{{\sf NDP}\xspace}

\begin{document}

\begin{titlepage}
\title{A New Conjecture on Hardness of Low-Degree 2-CSP’s with Implications to Hardness of Densest $k$-Subgraph and Other Problems}

\author{Julia Chuzhoy\thanks{Toyota Technological Institute at Chicago. Email: {\tt cjulia@ttic.edu}. Supported in part by NSF grant CCF-2006464.}  
\and Mina Dalirrooyfard\thanks{Massachusetts Institute of Technology. Email: {\tt minad@mit.edu}. Part of the work was done while the author was a summer intern at TTIC.}
\and Vadim Grinberg\thanks{Weizmann Institute of Science. Email: {\tt vadim.grinberg@weizmann.ac.il}.}
\and Zihan Tan\thanks{DIMACS, Rutgers University. Email: {\tt zihantan1993@gmail.com}.  Supported by a grant to DIMACS from the Simons
	Foundation (820931). Work done while the author was a graduate student at University of Chicago. } 
}
\maketitle

\thispagestyle{empty}

\begin{abstract}

We propose a new conjecture on hardness of low-degree $2$-CSP's, and show that new hardness of approximation results for Densest $k$-Subgraph and several other problems, including a graph partitioning problem, and a variation of the Graph Crossing Number problem,  follow from this conjecture. The conjecture can be viewed as occupying a middle ground between the $d$-to-$1$ conjecture, and hardness results for $2$-CSP's that can be obtained via standard techniques, such as Parallel Repetition combined with standard $2$-prover protocols for the 3SAT problem. We hope that this work will motivate further exploration of hardness of $2$-CSP's in the regimes arising from the conjecture. We believe that a positive resolution of the conjecture will provide a good starting point for further hardness of approximation proofs.

Another contribution of our work is proving that the problems that we consider are roughly equivalent from the approximation perspective. Some of these problems arose in previous work, from which it appeared that they may be related to each other. We formalize this relationship in this work.
\end{abstract}

\end{titlepage}

\pagenumbering{gobble}
\tableofcontents
\newpage
\pagenumbering{arabic}

\section{Introduction}

In this paper we consider several graph optimization problems, the most prominent and extensively studied of which is \DkS. 
One of the main motivations of this work is to advance our understanding of the approximability of these problems. Towards this goal, we propose a new conjecture on the hardness of a class of 2-\CSP problems, that we call \conj, and we show that new hardness of approximation results for all these problems follow from this conjecture. We believe that the conjecture is interesting in its own right, as it can be seen as occupying a middle ground between the $d$-to-$1$ conjecture, and the type of hardness of approximation results that one can obtain for 2-\CSP problems via standard methods (such as using constant-factor hardness of approximation results for 3-SAT, combined with standard 2-prover protocols and Parallel Repetition). While our conditional hardness of approximation proofs are combinatorial and algorithmic in nature, we hope that this work will inspire complexity theorists to study the conjecture, and also lead to other hardness of approximation proofs that combine both combinatorial and algebraic techniques. 

We prove a new conditional hardness of approximation result for \DkS based on \conj.
In addition to the \DkS problem, we study three other problems. The first problem, called \WGP, recently arose in the hardness of approximation proof of the Node-Disjoint Paths problem of  \cite{NDP-grid-hardness}, who mention that the problem appears similar to \DkS, but could not formalize this intuition. We also study a new problem that we call \DkC, that can be viewed as a natural middle ground between \DkS and \WGP. The fourth problem that we study is a variation of the notoriously difficult \MCN problem, that we call \BCS. This problem also arose implicitly in \cite{NDP-grid-hardness}. We show that all four problems are roughly equivalent from the approximation perspective, in the regime where the approximation factors are somewhat large (but some of our reductions require quasi-polynomial time). We then derive conditional hardness of approximation results for all these problems based on these reductions and the conditional hardness of \DkS. 

The main contribution of this paper is thus twofold: first, we propose a new conjecture on hardness of \CSP's and show that a number of interesting hardness of approximation results follow from it. Second, we establish a close connection between the four problems that we study. The remainder of the Introduction is organized as follows.
We start by providing a brief overview of the four problems that we study in this paper. We then state the \conj and put it into context with existing results and well-known conjectures. Finally, we provide a more detailed overview of our results and techniques.

\paragraph{Densest $k$-Subgraph.}
In the \DkS problem, given an $n$-vertex graph $G$ and an integer $k>1$, the goal is to compute a subset $S$ of $k$ vertices of $G$, while maximizing the number of edges in $G[S]$. 
 \DkS  is one of the most basic graph optimization problems that has been studied extensively (see e.g. \cite{kortsarz1993choosing,feige1997densest,feige2001dense,feige2001approximation,feige2002relations,khot2006ruling,goldstein2009dense,dks10,dks_average_hardness,bhaskara2012polynomial,barman2015approximating,braverman2017eth,Manurangsi16,chlamtac2018densest,manurangsi2018inapproximability,lin2018parameterized,sotirov2020solving,chang2020hardness,hanaka2022computing}). At the same time it seems notoriously difficult, and despite this extensive work, our understanding of its approximability is still incomplete.
The best current approximation algorithm for \DkS, due to \cite{dks10}, achieves, for every $\eps>0$, an $O(n^{1/4+\eps})$-approximation, in  time $n^{O(1/\eps)}$. 
Even though the problem appears to be very hard, its hardness of approximation proof has been elusive. 
For example, no constant-factor hardness of approximation proofs for \DkS are currently known under the standard $\P\neq \NP$ assumption, or even the stronger assumption that $\NP\not\subseteq \BPTIME(n^{\poly\log n})$. In a breakthrough result, Khot \cite{khot2006ruling} proved a factor-$c$ hardness of approximation for \DkS, for some small constant $c$, assuming that $\NP\not\subseteq \cap_{\eps>0}\BPTIME(2^{n^{\eps}})$. 
Several other papers proved constant and super-constant hardness of approximation results for \DkS under \emph{average-case} complexity assumptions: namely that no efficient algorithm can refute random $3$-SAT or random $k$-AND formulas \cite{feige2002relations, dks_average_hardness}. Additionally,
a factor $2^{\Omega(\log^{2/3}n)}$-hardness of approximation was shown under
assumptions on solving Planted Clique \cite{dks_average_hardness}. 
In a recent breakthrough, Manurangsi~\cite{Manurangsi16} proved that, under the Exponential Time Hypothesis (ETH), the \DkS problem is hard to approximate to within factor $n^{1/(\log\log n)^c}$, for some constant $c$. Proving a super-constant hardness of \DkS under weaker complexity assumptions remains a tantalizing open question that we attempt to address in this paper. Unfortunately, it seems unlikely that the techniques of \cite{Manurangsi16} can yield such a result. In this paper we show that, assuming the \conj that we introduce,  \DkS  is \NP-hard to approximate to within factor $2^{(\log n)^{\eps}}$, for some constant $\eps>0$.

\paragraph{The $(r,h)$-Graph Partitioning Problem.}
A recent paper \cite{NDP-grid-hardness} on the hardness of approximation of the Node-Disjoint Paths  (\NDP) problem formulated and studied a new graph partitioning problem, called \WGP. The input to the problem is a graph $G$, and two integers, $r$ and $h$. The goal is to compute $r$ vertex-disjoint subgraphs $H_1,\ldots,H_r$ of $G$, such that for each $1\leq i\leq r$, $|E(H_i)|\leq h$, while maximizing  $\sum_{i=1}^r|E(H_i)|$. A convenient intuitive way of thinking about this problem is that we are interested in obtaining a balanced partition of the graph $G$ into $r$ vertex-disjoint subgraphs, so that the subgraphs contain sufficiently many edges. Unlike standard graph partitioning problems, that typically aim to minimize the number of edges connecting the different subgraphs in the solution, our goal is to maximize the total number of edges that are contained in the subgraphs. In order to avoid trivial solutions, in which one of the subgraphs contains almost the entire graph $G$, and the remaining subgraphs are almost empty, we place an upper bound $h$ on the number of edges that each subgraph may contribute towards the solution. Note that the subgraphs $H_i$ of $G$ in the solution need not be vertex-induced subgraphs. 

The work of \cite{NDP-grid-hardness} attempted to use \WGP as a proxy problem for proving hardness of approximation of \NDP. Their results imply that \NDP is at least as hard to approximate as \WGP, to within polylogarithmic factors. In order to prove hardness of \NDP, it would then be sufficient to show that \WGP is hard to approximate. 
Unfortunately, \cite{NDP-grid-hardness} were unable to do so. Instead, they considered a generalization of \WGP, called \WGPwB. They showed that \NDP is at least as hard as \WGPwB, and then proved hardness of this new problem. In the \WGPwB problem, the input is the same as in  \WGP, but now graph $G$ must be bipartite, and, for every vertex $v$, we are given a partition $\bset(v)$ of the set of edges incident to $v$ into subsets that are called \emph{bundles}. We require that, in a solution $(H_1,\ldots,H_r)$ to the problem, for every vertex $v\in V(G)$, and every bundle $\beta\in \bset(v)$, at most one edge of $\beta$ contributes to the solution; in other words, at most one edge of $\beta$ may lie in $\bigcup_iE(H_i)$. This is a somewhat artificial problem, but  this definition allows one to bypass some of the barriers that arise when trying to prove the hardness of \WGP from existing hardness results for \CSP's.

It was noted in \cite{NDP-grid-hardness} that the \WGP problem resembles the \DkS problem for two reasons. First, in \DkS, the goal is to compute a dense subgraph of a given graph, with a prescribed number of vertices. One can think of \WGP as the problem of computing many vertex-disjoint dense subgraphs of a given graph. Second, natural hardness of approximation proofs for both problems seem to run into the same barriers.
It is therefore natural to ask: (i) Can we prove that the \WGP problem itself is hard to approximate? In particular, can the techniques  of \cite{NDP-grid-hardness} be exploited in order to obtain such a proof? and (ii) Can we formalize this intuitive connection between \WGP and \DkS? In this paper we make progress on both these questions. Our conditional hardness result for \DkS indeed builds on the ideas from \cite{NDP-grid-hardness} for proving hardness of \WGPwB. We also provide ``almost'' approximation-preserving reductions between \WGP to \DkS:  we show that, if there is an efficient factor $\alpha(n)$-approximation algorithm for \DkS, then there is a randomized efficient factor $O(\alpha(n^2)\cdot\poly\log n)$-approximation algorithm to \WGP. We also provide a reduction in the opposite direction: we prove that, if there is an efficient $\alpha(n)$-approximation algorithm for \WGP, then there is a randomized algorithm for \DkS, that achieves approximation factor $O\left ((\alpha(n^{O(\log n)}))^3\cdot \log^2 n\right )$, in time $n^{O(\log n)}$. Therefore, we prove that \DkS and \WGP are roughly equivalent from the approximation perspective (at least for large approximation factors and quasi-polynomial running times).
Combined with our conditional hardness of approximation for \DkS, our results show that, assuming the \conj, for some constant $0<\eps\leq 1/2$, there is no efficient $2^{(\log n)^{\eps}}$-approximation algorithm for \WGP, unless $\NP\subseteq \BPTIME(n^{O(\log n)})$.

\paragraph{Maximum Bounded-Crossing Subgraph.}
The third problem that we study is a variation of the classical \MCN problem. In the \MCN problem, given an input $n$-vertex graph $G$, the goal is to compute a drawing of $G$ in the plane while minimizing the number of crossings in the drawing. We define the notions of graph drawing and crossings formally in the Preliminaries, but these notions are quite intuitive and the specifics of the definition are not important in this high-level overview.

The \MCN problem was initially introduced by Tur\'an \cite{turan_first} in 1944, and has been extensively studied since then  (see, e.g., \cite{chuzhoy2011algorithm, chuzhoy2011graph, chimani2011tighter, chekuri2013approximation, KawarabayashiSidi17, kawarabayashi2019polylogarithmic,chuzhoy2020towards}, and also \cite{richter_survey, pach_survey, matousek_book, schaefer2012graph} for excellent surveys). 
But despite all this work, most aspects of the problem are still poorly understood. A long line of work \cite{leighton1999multicommodity, even2002improved,  chuzhoy2011graph,chuzhoy2011algorithm,KawarabayashiSidi17, kawarabayashi2019polylogarithmic,improved-improved-gmt-arxiv,chuzhoy2022subpolynomial} has recently led to the first sub-polynomial approximation algorithm for the problem in \emph{low degree graphs}. Specifically, \cite{chuzhoy2022subpolynomial} obtain a factor $O\left(2^{O((\log n)^{7/8}\log\log n)}\cdot\Delta^{O(1)}\right )$-approximation algorithm for \MCN, where $\Delta$ is the maximum vertex degree. To the best of our knowledge, no non-trivial approximation algorithms are known for the problem when vertex degrees in the input graph $G$ can be arbitrary. However, on the negative side, only APX-hardness is known for the problem \cite{cabello2013hardness,ambuhl2007inapproximability}. As the current understanding of the \MCN problem from the approximation perspective is extremely poor, it is natural to study hardness of approximation of its variants.

 Let us consider two extreme variations of the \MCN problem. The first variant is the \MCN problem itself, where we need to draw an input graph $G$ in the plane with fewest crossings. The second variant is where we need to compute a subgraph $G'$ of the input graph $G$ that is planar, while maximizing $|E(G')|$. The latter problem has a simple constant-factor approximation algorithm, obtained by letting $G'$ be any spanning forest of $G$ (this is since a planar $n$-vertex graph may only have $O(n)$ edges).

In this paper we study a variation of the \MCN problem, that we call \BCS, which can be viewed as an intermediate problem between these two extremes. In the \BCS problem, given an $n$-vertex graph $G$ and an integer $L>0$, the goal is to compute a subgraph $H\subseteq G$, such that $H$ has a plane drawing with at most $L$ crossings, while maximizing $|E(H)|$. This problem is only interesting when the bound $L$ on the number of crossings  is $\Omega(n)$. This is since, from the Crossing Number Inequality \cite{ajtai82,leighton_book}, if $|E(G)|\geq 4|V(G)|$, then the crossing number of $G$ is at least $ \Omega(|E(G)|^3/|V(G)|^2)$. Therefore, for $L=O(n)$, a spanning tree provides a constant-factor approximation to the problem. 
We emphasize that the focus here is on dense graphs, whose crossing number may be as large as $\Omega(n^4)$. 

The \BCS problem was implicitly used in  \cite{NDP-grid-hardness} for proving hardness of approximation of \NDP, as an intermediate problem, in the reduction from \WGPwB to \NDP. Their work suggests that  there may be a connection between \WGP  and  \BCS, even though the two problems appear quite different. In this paper we prove that the two problems are roughly equivalent from the approximation perspective: if there is an efficient factor $\alpha(n)$-approximation algorithm for \WGP, then there is an efficient $O(\alpha(n)\cdot\poly\log n)$-approximation algorithm for \BCS. On the other hand, an efficient $\alpha(n)$-approximation algorithm for \BCS implies an efficient  $O((\alpha(n))^2\cdot \poly\log n)$-approximation algorithm for \WGP. Combined with our conditional hardness of approximation for \WGP, we get that, assuming the \conj, for some constant $0<\eps\leq 1/2$ there is no efficient $2^{(\log n)^{\eps}}$-approximation algorithm for \BCS, unless $\NP\subseteq \BPTIME(n^{O(\log n)})$.

\paragraph{Dense $k$-Coloring.}
The fourth and last problem that we consider is \DkC. In this problem, the input is an $n$-vertex graph $G$ and an integer $k$, such that $n$ is an integral multiple of $k$. The goal is to partition $V(G)$ into $n/k$ disjoint subsets $S_1,\ldots,S_{n/k}$, of cardinality $k$ each, so as to maximize $\sum_{i=1}^{n/k}|E(S_i)|$.  This problem can be viewed as an intermediate problem between \DkS and \WGP. The connection to \WGP seems clear: in both problems, the goal is to compute a large collection of subgraphs of the input graph $G$, that contain many edges of $G$. While in \WGP we place a limit on the number of edges in each subgraph, in \DkC we require that each subgraph contains exactly $k$ vertices. The connection to the \DkS problem is also clear: while in \DkS  the goal is to compute a single dense subgraph of $G$ containing $k$ vertices, in \DkC we need to partition $G$ into many dense subgraphs, containing $k$ vertices each. We show reductions between the \DkC and the \DkS problem in both directions, that provide very similar guarantees to the reductions between \WGP and \DkS. In particular,  our results show that, assuming the \conj, for some constant $0<\eps\leq 1/2$, there is no efficient $2^{(\log n)^{\eps}}$-approximation algorithm for \DkC, unless $\NP\subseteq \BPTIME(n^{O(\log n)})$.

\paragraph{The Low-Degree CSP Conjecture.}
We now turn to describe our new conjecture on the hardness of 2-\CSP's.
We consider the following bipartite version of the Constraint Satisfaction Problem with 2 variables per constraint (2-\CSP).
The input consists of two sets $X$ and $Y$ of variables, together with  an integer $A\geq 1$. Every variable in $X\cup Y$ takes values in $[A]=\set{1,\ldots,A}$. We are also given a collection $\cset$ of constraints, where each constraint $C(x,y)\in \cset$ is defined over a pair of variables $x\in X$ and $y\in Y$. For each such constraint, we are given a truth table that, for every pair of assignments $a$ to $x$ and $a'$ to $y$, specifies whether $(a,a')$ \emph{satisfy} the constraint. 
The \emph{value} of the \CSP is the largest fraction of constraints that can be simultaneously satisfied by an assignment to the variables.
For given values $0<s<c\leq 1$, the $(c,s)$-\GCSP problem is the problem of distinguishing \CSP's of value at least $c$ from those of value at most $s$.

We can associate, to each constraint $C=C(x,y)\in \cset$, a bipartite graph $G_C=(L,R,E)$, where $L=R=[A]$, and there is an edge $(a,a')$ in $E$ iff the assignments $a$ to $x$ and $a'$ to $y$ satisfy $C$. 
Notice that instance $\iset$ of the \BCSP problem is completely defined by $X,Y,A,\cset$, and the graphs in $\set{G_C}_{C\in \cset}$, so we will denote $\iset=(X,Y,A,\cset,\set{G_C}_{C\in \cset})$.
We let the \emph{size} of instance $\iset$ be $\size(\iset)=|\cset|\cdot A^2+|X|+|Y|$. We sometimes refer to $A$ as the \emph{size of the alphabet for instance $\iset$}.
We say that instance $\iset$ of 2-\CSP is \emph{$d$-to-$d'$} iff for every constraint $C$, every vertex of $G_C$  that lies in $L$ has degree at most $d$, and every vertex that lies in $R$ has degree at most $d'$. (We note that this is somewhat different from the standard definition, that requires that all vertices in $L$ have degree exactly $d$ and all vertices of $R$ have degree exactly $d'$. In the standard definition, the alphabet sizes for variables in $X$ and $Y$ may be different, that is, variables in $X$ take values in $[A]$ and variables of $Y$ take values in $[A']$ for some integers $A,A'$. However, this difference is insignificant to our discussion, and it is more convenient for us to use this slight variation of the standard definition).

The famous Unique-Games Conjecture of Khot \cite{Khot-UGC} applies to $1$-to-$1$ CSP's. The conjecture states that, for any $0<\eps<1$, there is a large enough value $A$, such that the $(1-\eps,\eps)$-\GCSP problem is \NP-hard for $1$-to-$1$ instances with alphabet size $A$. 
The conjecture currently remains open, though interesting progress has been made on the algorithmic side: the results of \cite{arora2015subexponential} provide an algorithm for the problem with running time $2^{n^{O(1/\eps^{1/3})}}$.

A conjecture that is closely related to the Unique-Games Conjecture is the $d$-to-$1$ Conjecture of  Khot \cite{Khot-UGC}. The conjecture states that,  for every $0<\eps<1$, and $d>0$, there is a large enough value $A$, such that the $(1,\eps)$-\GCSP problem in $d$-to-$1$ instances with alphabet size $A$ is \NP-hard.

H{\aa}stad \cite{Hastad} proved the following nearly optimal hardness of approximation results for \CSP's: he showed that for every $0<\eps<1$, there are values $d$ and $A$, such that the problem of $(1,\eps)$-\GCSP in $d$-to-$1$ instances with alphabet size $A$ is \NP-hard. The value $d$, however, depends exponentially on $\poly(1/\eps)$ in this result.
In contrast, in the $d$-to-$1$ Conjecture, both $d$ and $\eps$ are fixed, and $d$ may not have such a strong dependence on $1/\eps$.

On the algorithmic side, the results of \cite{arora2015subexponential, steurer2010subexponential} provide an algorithm for  $(c,s)$-\GCSP on $d$-to-$1$ instances. The running time of the algorithm is $2^{n^{O(1/(\log(1/s))^{1/2})}}$, where the $O(\cdot)$ notation hides factors that are polynomial in $d$ and $A$.

A recent breakthrough in this area is the proof of the $2$-to-$2$ conjecture (now theorem), that builds on a long sequence of work \cite{barak2015making,2-to-2-7,2-to-2-6,2-to-2-5,2-to-2-4,2-to-2-3,2-to-2-2,2-to-2}. The theorem proves that for every $0<\eps<1$, there is a large enough value $A$, such that the $(1-\eps,\eps)$-\GCSP problem is \NP-hard on $2$-to-$2$ instances with alphabet size $A$.

In this paper, we propose the following conjecture, that we refer to as \emph{\conj}, regarding the hardness of $\GCSP$ in $d$-to-$d$ instances.

\begin{conjecture}[\conj]
	\label{conj: main}
	There is a constant $0<\eps\leq 1/2$, such that it is \NP-hard to distinguish between $d(n)$-to-$d(n)$ instances of 2-\CSP of size $n$, that have value at least $1/2$, and those of value at most $s(n)$, where 
	$d(n)=2^{(\log n)^{\eps}}$ and $s(n)= 1/2^{64 (\log n)^{1/2+\eps}}$. 
\end{conjecture}

We now compare this conjecture to existing conjectures and results in this area that we are aware of.
First, in contrast to the $d$-to-$1$ conjecture, we allow the parameter $d$ and the soundness parameter $s$ to be functions of $n$ -- the size of the input instance. Note that the size of the input instance depends on the alphabet size $A$, so, unlike in the setting of the $d$-to-$1$ conjecture, $A$ may no longer be arbitrarily large compared to $d$ and $s$. 

The hardness of approximation result of H{\aa}stad \cite{Hastad} 
for $d$-to-$d$ \CSP's only holds when $d$ depends exponentially on $\poly(1/s)$,  (in particular it may not extend to the setting where $s(n)= 1/2^{64 (\log n)^{1/2+\eps}}$, since the size $n$ of the instance depends polynomially on $d(n)$).

We can also combine standard constant hardness of approximation results for \CSP's (such as, for example, $3$-SAT) with the Parallel Repetition theorem, to obtain \NP-hardness of $(1,s(n))$-\GCSP on $d(n)$-to-$d(n)$ instances. Using this approach, if we start from an instance of \CSP of size $N$ and a constant hardness gap (with perfect completeness), after $\ell$ rounds of parallel repetition, we obtain hardness of  $(1,s)$-\GCSP on $d$-to-$d$ instances with $s=2^{O(\ell)}$, $d=2^{O(\ell)}$, and the resulting instance size $n=N^{O(\ell)}$.
By setting the number of repetition to be $\ell=\Theta\left((\log N)^{(1/2+\eps)/(1/2-\eps)}\right )$, we can ensure the desired bound  $s(n)=1/2^{64 (\log n)^{1/2+\eps}}$. However, in this setting, we also get that $d(n)=2^{\Omega((\log n)^{1/2+\eps})}$, which is significantly higher than the desired value $d(n)=2^{(\log n)^{\eps}}$.

Lastly, one could attempt to  combine the recent proof of the $2$-to-$2$ conjecture with Parallel Repetition in order to reap the benefits of both approaches, but the resulting parameters also fall short of the ones stated in the conjecture.

From the above discussion, one can view the \conj as occupying a middle ground between the $d$-to-$1$ conjecture, and the results one can obtain via standard techniques of amplifying a constant hardness of a \CSP, such as 3SAT, via Parallel Repetition.

We now proceed to discuss our results and techniques in more detail.

\subsection{A More Detailed Overview of our Results and Techniques}

In addition to posing the \conj that we already described above, we prove conditional hardness of approximation of the four problems that we consider. We also prove that all four problems are roughly equivalent approximation-wise. We now discuss the conditional hardness of approximation for \DkS and the connections between the four problems that we establish.

\paragraph{Conditional Hardness of \DkS.}
Our first result is a conditional hardness of \DkS. Specifically, we prove that, assuming \Cref{conj: main} holds and $\P\ne \NP$,  for some $0<\eps\leq 1/2$, there is no efficient approximation algorithm for \DkS problem that achieves approximation factor $2^{(\log N)^{\eps}}$, where $N$ is the number of vertices in the input graph.

We now provide a brief overview of our techniques. The proof of the above result employs a Cook-type reduction, and follows some of the ideas that were introduced in  \cite{NDP-grid-hardness}. We assume for contradiction that there is a factor-$\alpha$ algorithm $\aset$ for the \DkS problem, where $\alpha=2^{(\log N)^{\eps}}$. Given an input instance $\iset$ of the $2$-\CSP problem of size $n$, that is a $d(n)$-to-$d(n)$ instance, we construct a constraint graph $H$ representing $\iset$. We gradually decompose graph $H$ into a collection $\hset$ of disjoint subgraphs, such that, for each subgraph $H'\in \hset$, we can either certify that the value of the corresponding instance of $2$-\CSP is at most $1/4$, or it is at least $\beta$, for some carefully chosen parameter $\beta$. In order to compute the decomposition, we start with $\hset=\set{H}$. If, for a graph $H'\in \hset$, we certified that the corresponding instance of $2$-\CSP has value at most $1/4$, or at least $\beta$, then we say that graph $H'$ is \emph{inactive}. Otherwise, we say that it is \emph{active}. As long as $\hset$ contains at least one active graph, we perform iterations. In each iteration, we select an arbitrary graph $H'\in \hset$ to process. In order to process $H'$, we consider an \emph{assignment graph} $G'$ associated with $H'$, that contains a vertex for every variable-assignment pair $(x,a)$, where $x$ is a variable whose corresponding vertex belongs to $H'$. We view $G'$ as an instance of the \DkS problem, for an appropriately chosen parameter $k$, and apply the approximation algorithm $\aset$ for \DkS to it. Let $S$ be the set of vertices of $G'$ that Algorithm $\aset$ computes as a solution to this instance. We exploit the set $S$ of vertices in order to either (i) compute a large subset $E'\subseteq E(H')$ of edges, such that, if we denote by $\cset'\subseteq \cset$ the set of constraints corresponding to $E'$, then at most $1/4$ of the constraints of $\cset'$ can be simultaneously satisfied; or (ii) compute a large subset $E'\subseteq E(H')$ of edges as above, and certify that at least a $\beta$-fraction of such constraints can be satisfied; or (iii) compute a subgraph $H''\subseteq H'$, such that $|V(H'')|\ll |V(H')|$, and the number of edges contained in graphs $H''$ and $H'\setminus V(H'')$ is sufficiently large compared to $E(H')$. In the former two cases, we replace $H'$ with graph $H'[E']$ in $\hset$, and graph $H'[E']$ becomes inactive. In the latter case, we replace $H'$ with two graphs: $H''$ and $H'\setminus V(H'')$, that both remain active. The algorithm terminates once every graph in $\hset$ is inactive. The crux of the analysis of the algorithm is to show that, when the algorithm terminates, the total number of edges lying in the subgraphs $H'\in \hset$ is high, compared to $|E(H)|$. This algorithm for decomposing graph $H$ into subgraphs and its analysis employ some of the techniques and ideas introduced in \cite{NDP-grid-hardness}, and is very similar in spirit to the hardness of approximation proof of the \WGPwB problem, though  details are different. We employ this decomposition algorithm multiple times, in order to obtain a partition $(E_0,E_1,\ldots,E_z)$ of the set $E(H)$ of edges into a small number of subsets, such that, among the constraints corresponding to the edges of $E_0$, at most a $1/4$-fraction can be satisfied by any assignment to $X\cup Y$, and, for all $1\leq i\leq z$, a large fraction of constraints corresponding to edges of $E_i$ can be satisfied by some assignment. Depending on the cardinality of the set $E_0$ of edges we then determine whether $\iset$ is a \yi or a \ni.

\paragraph{Reductions from \DkC and \WGP to \DkS.}
We show that, if there is an efficient factor $\alpha(n)$-approximation algorithm for the \DkS problem, then there is an efficient $O(\alpha(n^2)\cdot\poly\log n)$-approximation algorithm for \DkC, and an efficient $O(\alpha(n^2)\cdot \poly\log n)$-approximation algorithm for \WGP. The two reductions are very similar, so we focus on describing the first one. We believe that the reduction is of independent interest, and uses unusual techniques.

We assume that there is an $\alpha(n)$-approximation algorithm for the \DkS problem. In order to obtain an approximation algorithm for \DkC, we start by formulating a natural LP-relaxation for the problem. Unfortunately, this LP-relaxation has a large number of variables: roughly $n^{\Theta(k)}$, where $n$ is the number of vertices in the 
input graph and $k$ is the parameter of the \DkC problem instance. We then show an efficient algorithm, that, given a solution to the LP-relaxation, whose support size is bounded by $\poly(n)$, computes an approximate integral solution to the \DkC problem.

The main challenge is that, since the LP relaxation has $n^{\Theta(k)}$ variables, it is unclear how to solve it efficiently. We consider the dual linear program, that has $\poly(n)$ variables and $n^{\Theta(k)}$ constraints. Using the $\alpha(n)$-approximation algorithm for \DkS as a subroutine, we design an approximate separation oracle for the dual LP, that allows us to solve the original LP-relaxation for \DkC, obtaining a solution whose support size is bounded by $\poly(n)$. By applying the LP-rounding approximation algorithm to this solution, we obtain the desired approximate solution to the input instance of \DkC.

\paragraph{Reductions from \DkS to \WGP and \DkC.}
We prove that, if there is an efficient $\alpha(n)$-approximation algorithm for \DkC, then there is a randomized algorithm for the \DkS problem, whose running time is $n^{O(\log n)}$, that with high probability obtains an $O(\alpha(n^{O(\log n)})\cdot\log n)$-approximate solution to the input instance of the problem. We also show a similar reduction from \DkS to \WGP, but now the resulting approximation factor for \DkS becomes $O((\alpha(n^{O(\log N)}))^3\cdot\log^2 n)$. By combining these reductions with our conditional hardness result for \DkS, we get that,  assuming the \conj, for some constant $0<\eps\leq 1/2$, there is no efficient $2^{(\log n)^{\eps}}$-approximation algorithm for \WGP and for \DkC, unless $\NP\subseteq \BPTIME(n^{O(\log n)})$.

The two reductions are very similar; we focus on the reduction to \DkC in this overview. Our construction is inspired by the results of \cite{khanna2000hardness}, and we borrow some of our ideas from them. 
Assume that there is an efficient $\alpha(n)$-approximation algorithm for \DkC. Let $G$ be an instance of the \DkS problem. The main difficulty in the reduction is that it is possible that $G$ only contains one very dense subgraph induced by $k$ vertices, while the \DkC problem requires that the input graph $G$ can essentially be partitioned into many such dense subgraphs. To overcome this difficulty, we construct a random ``inflated'' bipartite graph $H$, that contains $n^{O(\log n)}$ vertices, where $n=|V(G)|$. Every vertex of $G$ is mapped to some vertex of $H$ at random, while every edge of $G$ is mapped to a large number of edges of $H$. This allows us to ensure that, if $G$ contains a subgraph $G'$ induced by a set of $k$ vertices, where $|E(G')|=R$, then graph $H$ can essentially be partitioned into a large number of subgraphs that contain $k$ vertices each, and many of them contain close to $R$ edges. Therefore, we can apply our $\alpha(n)$-approximation algorithm for \DkC to the new graph $H$. The main challenge in the reduction is that, while this approximation algorithm is guaranteed to return a large number of disjoint dense subgraphs of $H$, since every edge of $G$ contributes many copies to $H$, it is not clear that one can extract a single dense subgraph of $G$ from dense subgraphs of $H$. The main difficulty in the reduction is to ensure that, on the one hand, a single $k$-vertex dense subgraph in $G$ can be translated into $|V(H)|/k$ dense subgraphs of $H$; and, on the other hand, a dense $k$-vertex subgraph of $H$ can be translated into a dense subgraph of $G$ on $k$ vertices. We build on and expand the ideas from \cite{khanna2000hardness} in order to ensure these properties.

\paragraph{Reductions between \WGP and \BCS.}
Lastly, we provide reductions between \WGP and \BCS in both directions. First, we show that, if there is an efficient factor $\alpha(n)$-approximation algorithm for \WGP, then there is an efficient $O(\alpha(n)\cdot\poly\log n)$-approximation algorithm for \BCS. On the other hand, an efficient $\alpha(n)$-approximation algorithm for \BCS implies an efficient $O((\alpha(n))^2\cdot \poly\log n)$-approximation algorithm for \WGP. Combined with our conditional hardness of approximation for \WGP, we get that, assuming the \conj, for some constant $0<\eps\leq 1/2$, there is no efficient $2^{(\log n)^{\eps}}$-approximation algorithm for \WGP, unless $\NP\subseteq \BPTIME(n^{O(\log n)})$.

Both these reductions exploit the following connection between crossing number and graph partitioning: if a graph $G$ has a drawing with at most $L$ crossings, then there is a balanced cut in $G$, containing at most $O\left(\sqrt{L+\Delta\cdot |E(G)|}\right )$ edges, where $\Delta$ is maximum vertex degree in $G$. This result can be viewed as an extension of the classical Planar Separator Theorem of \cite{lipton1979separator}.
Another useful fact exploited in both reductions is that any graph $G$ with $m$ edges has a plane drawing with at most $m^2$ crossings. In particular, if $\hset=\set{H_1,\ldots,H_r}$ is a solution to an instance of the \WGP problem on graph $G$, then there is a drawing of graph $H=\bigcup_{i=1}^rH_i$, in which the number of crossings is bounded by $r\cdot h^2$. These two facts establish a close relationship between the \WGP and \BCS problems, that are exploited in both our reductions.

We have now obtained a chain of reductions that show that all four problems, \DkS, \DkC, \WGP, and \BCS are almost equivalent from approximation viewpoint (if one considers sufficiently large approximation factors, and allows randomized quasi-polynomial time algorithms). We also obtain conditional hardness of approximation results for all four problems based on the \conj.

\paragraph{Organization.}
We start with preliminaries in \Cref{sec: prelim}. In \Cref{sec: conditional} we provide the conditional hardness of approximation proof for the \DkS problem. In \Cref{sec: DkC and WGP to DkS} we provide our reductions from \DkC and \WGP to \DkS, and in \Cref{sec: DkS to DkC and WGP} we provide reductions in the opposite direction. 
Lastly in \Cref{sec: WGP and BCS} we provide reductions between \WGP and \BCS.

\section{Preliminaries}
\label{sec: prelim}

By default, all logarithms are to the base of $2$. For a positive integer $N$, we denote by $[N]=\set{1,2,\ldots,N}$. All graphs are finite, simple and undirected. We say that an  event holds with high probability if the probability of the event is $1-1/n^c$ for a large enough constant $c$, where $n$ is the number of vertices in the input graph.

\subsection{General Notation}
Let $G$ be a graph and let $S$ be a subset of its vertices. We denote by $G[S]$ the subgraph of $G$ induced by $S$.
For two disjoint subsets $A,B$ of vertices of $G$, we denote by $E_G(A,B)$ the set of all edges with one endpoint in $A$ and the other endpoint in $B$, and we denote by $E_G(A)$ the set of all edges with both endpoints in $A$.
Given a graph $G$ and a vertex $v\in V(G)$, we denote by $\deg_G(v)$ the degree of $v$ in $G$. For a subset $S$ of vertices of $G$, its \emph{volume} is $\vol_G(S)=\sum_{v\in S}\deg_G(v)$.
 We sometimes omit the subscript $G$ if it is clear from the context.


Given a graph $G$, a drawing $\varphi$ of $G$ is an embedding of  $G$ into the plane, that maps
every vertex $v$ of $G$ to a point (called the \emph{image of $v$} and denoted by $\varphi(v)$), and every edge $e$ of $G$ to a simple curve (called the \emph{image of $e$} and denoted by $\varphi(e)$), that connects the images of its endpoints. If $e$ is an edge of $G$ and $v$ is a vertex of $G$, then the image of $e$ may only contain the image of $v$ if $v$ is an endpoint of $e$. Furthermore, if some point $p$ belongs to the images of three or more edges of $G$, then $p$ must be the image of a common endpoint of all edges $e$ with $p\in \phi(e)$. We say that two edges $e,e'$ of $G$ \emph{cross} at a point $p$, if $p\in \phi(e)\cap \phi(e')$, and $p$ is not the image of a shared endpoint of these edges.
Given a graph $G$ and a drawing $\varphi$ of $G$ in the plane, we use $\crn(\varphi)$ to denote the number of crossings in $\varphi$, and the crossing number of $G$, denoted by $\CrN(G)$, is the minimum number of crossings in any drawing of $G$.

\subsection{Problem Definitions and Additional Notation}


%
%

In this paper we consider the following four problems: \DkS, \DkC, \WGP and \BCS. We now define the problems, along with some additional notation.

\paragraph{Densest $k$-Subgraph.} In the \DkS problem, the input is a graph $G$ and an integer $k>0$. The goal is to compute a subset $S\subseteq V(G)$ of $k$ vertices, maximizing $|E_G(S)|$. We denote an instance of the problem by  $\pdks(G, k)$, and we denote the value of the optimal solution to instance $\pdks(G, k)$ by $\optdks(G, k)$.

We also consider a bipartite version of the \DkS problem, called \newline $\DBS$. 
This problem was first studied in \cite{dks_average_hardness}.
The input to the problem is a bipartite graph $G=(A,B,E)$ and positive integers $k_1,k_2$. The goal is to compute a subset $S\subseteq V(G)$ of vertices with $|S\cap A|=k_1$ and $|S\cap B|=k_2$, such that $|E_G(S)|$ is maximized.
An instance of this problem is denoted by $\pbdks(G, k_1,k_2)$, and the value of the optimal solution to instance $\pbdks(G, k_1,k_2)$ is denoted by $\optbdks(G, k_1,k_2)$.
The following lemma shows that the \DBS problem and the \DkS problem are roughly equivalent from the approximation viewpoint. Similar results were also shown in prior work. For completeness, we provide the proof in \Cref{apd: Proof of DkS and Dk1k2S}. 

\begin{lemma}
\label{lem: DkS and Dk1k2S}
Let $\alpha: \mathbb{Z^+}\to \mathbb{Z^+}$ be an increasing function such that $\alpha(n)=o(n)$. Then the following hold:
\begin{itemize}
\item If there exists an $\alpha(n)$-approximation algorithm for the \DkS problem with running time at most $T(n)$, where $n$ is the number of vertices in the input graph, then there exists an  $O(\alpha(N^2))$-approximation algorithm for the \BDkS problem, with running time $O(T(N^2)\cdot \poly(N))$, where $N$ is the number of vertices in the input graph. Moreover, if the algorithm for \DkS is deterministic, then so is the algorithm for \BDkS.

\item Similarly, if there exists an efficient $\alpha(N)$-approximation algorithm for the \BDkS problem, where $N$ is the number of vertices in the input graph, then there exists an efficient $O(\alpha(2n))$-approximation algorithm for the \DkS problem, where $n$ is the number of vertices in the input graph. Moreover, if the algorithm for \BDkS is deterministic, then so is the algorithm for \DkS.
\end{itemize}
\end{lemma}

\paragraph{Dense $k$-Coloring.}
The input to the \DkC problem consists of an $n$-vertex graph $G$ and an integer $k>0$, such that $n$ is an integral multiple of $k$. The goal is to compute a partition of $V(G)$ into  $n/k$ subsets $S_1, \ldots, S_{n/k}$ of cardinality $k$ each, while maximizing $\sum_{i = 1}^{n/k}|E_G(S_i)|$. 
An instance of the \DkC problem is denoted by $\pcl(G, k)$, and the value of the optimal solution to instance $\pcl(G, k)$ is denoted by $\optcl(G, k)$.

\paragraph{$(r,h)$-Graph Partitioning.}
  The input to the \WGP problem consists of a graph $G$, and integers $r, h > 0$. The goal is to compute $r$ vertex-disjoint subgraphs $H_1, \ldots, H_r$ of $G$, such that for all $1\leq i\leq r$, $|E(H_i)|\leq h$, while maximizing $\sum_{i = 1}^r|E(H_i)|$.
An instance of the \WGP problem is denoted by $\pwgp(G, r, h)$, and the value of the optimal solution to instance $\pwgp(G, r, h)$ is denoted by $\optwgp(G, r, h)$.

\paragraph{Maximum Bounded-Crossing Subgraph.}

In the \MBCS problem, the input is a graph $G$ and an integer $L>0$. The goal is to compute a subgraph $H\subseteq G$ with $\CrN(H)\leq L$, while maximizing $|E(H)|$.
An instance of the \MBCS problem is denoted by $\pcn(G, L)$, and the value of the optimal solution to instance $\pcn(G, L)$ is denoted by $\optcn(G, L)$. We note that we can assume that $L\le |V(G)|^4$, as otherwise the optimal solution is the whole graph $G$, since the crossing number of a simple graph $G$ is at most $|E(G)|^2\leq |V(G)|^4$.

\subsection{Chernoff Bound}

We use the following standard version of Chernoff Bound (see. e.g., \cite{dubhashi2009concentration}).

\begin{lemma}[Chernoff Bound]
	\label{lem: Chernoff}
	Let $X_1,\ldots,X_n$ be independent randon variables taking values in $\set{0,1}$. Let $X=\sum_{1\le i\le n}X_i$, and let $\mu=\expect{X}$. Then for any $t>2e\mu$,
	\[\Pr\Big[X>t\Big]\le 2^{-t}.\]
	Additionally, for any $0\le \delta \le 1$,
	\[\Pr\Big[X<(1-\delta)\cdot\mu\Big]\le e^{-\frac{\delta^2\cdot\mu}{2}}.\]
\end{lemma}

\subsection{Auxiliary Lemma}

We use the following simple lemma.

\begin{lemma}
\label{lem: size reducing}
There is an efficient algorithm, that, given a graph $G$, a subset $S$ of its vertices, and a parameter $\frac{2}{|S|}<\beta<1$, computes a set $S'\subseteq S$ of vertices, such that $|S'|\le \beta\cdot |S|$, and $|E_G(S')|\ge \Omega(\beta^2\cdot |E_G(S)|)$ holds.
\end{lemma}

\begin{proof}
Consider the graph $G'=G[S]$ and denote $|S|=k$.
We iteratively remove the lowest-degree vertex from $G'$, until $G'$ contains $\floor{\beta k}$ vertices. Once the algorithm terminates, we output $S'=V(G')$. It now remains to show that $|E_{G'}(S')|\ge \Omega(\beta^2|E(G')|)$.

Observe that, if $H$ is an $n$-vertex graph, and $v$ is a lowest-degree vertex of $H$, then the degree of $v$ in $H$ is at most $2|E(H)|/n$. Therefore, if vertex $v$ is removed from $H$, then $|E(H)|$ decreases by at most a factor $(1-2/n)$. Therefore,
\[
\frac{|E_{G'}(S')|}{|E(G')|}\ge \bigg(1-\frac{2}{k}\bigg)\bigg(1-\frac{2}{k-1}\bigg) \cdots \bigg(1-\frac{2}{\floor{\beta k}+1}\bigg)=\frac{\floor{\beta k}\cdot(\floor{\beta k}-1)}{k\cdot(k-1)}=\Omega(\beta^2).
\]
\end{proof}

\section{Conditional Hardness of \DkS}
\label{sec: conditional}

\subsection{Low-Degree CSP Conjecture}
We consider the \BCSP problem, that is defined as follows.
The input to the problem consists of two sets $X,Y$ of variables, together with an integer $A>1$. Every variable $z\in X\cup Y$ takes values in set $[A]=\set{1,\ldots,A}$. We are also given a collection $\cset$ of constraints, where each constraint $C(x,y)\in \cset$ is defined over a pair of variables $x\in X$ and $y\in Y$. For each such constraint, we are given a truth table that, for every pair of assignments $a$ to $x$ and $a'$ to $y$, specifies whether $(a,a')$ \emph{satisfy} constraint $C(x,y)$. 
The \emph{value} of the \CSP is the largest fraction of constraints that can be simultaneously satisfied by an assignment to the variables.

We associate with each constraint $C=C(x,y)\in \cset$, a bipartite graph $G_C=(L,R,E)$, where $L=R=[A]$, and there is an edge $(a,a')$ in $E$ iff the assignments $a$ to $x$ and $a'$ to $y$ satisfy $C$. Notice that instance $\iset$ of the \BCSP problem is completely defined by $X,Y,A,\cset$, and the graphs in $\set{G_C}_{C\in \cset}$, so we will denote $\iset=(X,Y,A,\cset,\set{G_C}_{C\in \cset})$.
The \emph{size} of instance $\iset$ is defined to be $\size(\iset)=|\cset|\cdot A^2+|X|+|Y|$.

Consider some instance $\iset=(X,Y,A,\cset,\set{G_C}_{C\in \cset})$ of \BCSP. We say that $\iset$ is a \emph{$d$-to-$d$ instance} if, for every constraint $C$, every vertex of graph $G_C=(L,R,E)$  has degree at most $d$.

Consider now some functions $d(n),s(n):\reals^+\rightarrow \reals^+$. We assume that, for all $n$, $d(n)\geq 1$ and $s(n)<1$.
In a $(d(n),s(n))$-\LDCSP problem, the input is an instance $\iset=(X,Y,A,\cset,\set{G_C}_{C\in \cset})$ of \BCSP, such that, if we denote by $n=\size(\iset)$, then the instance is $d(n)$-to-$d(n)$. We say that $\iset$ is a  \yi, if there is some assignment to the variables of $X\cup Y$ that satisfies at least $|\cset|/2$ of the constraints, and we say that it is a \ni, if the largest number of constraints of $\cset$ that can be simultaneously satisfied by any assignment is at most $s(n)\cdot |\cset|$. Given an instance $\iset$ of $(d(n),s(n))$-\LDCSP problem, the goal is to distinguish between the case where $\iset$ is a \yi and the case where $\iset$ is a \ni. If $\iset$ is neither a \yi nor a \ni, the output of the algorithm can be arbitrary.
We now state our conjecture regarding hardness of $(d(n),s(n))$-\LDCSP, that is a restatement of  \Cref{conj: main} from the Introduction.

\begin{conjecture}[\conj]
	\label{conj: main1}
	There is a constant $0<\eps\leq 1/2$, such that the $(d(n),s(n))$-\LDCSP problem is \NP-hard for $d(n)=2^{(\log n)^{\eps}}$ and $s(n)= 1/2^{64 (\log n)^{1/2+\eps}}$. 
\end{conjecture}

\subsection{Conditional Hardness of \DkS}

In the remainder of this section, we prove the following theorem on the conditional hardness of \DkS.

\begin{theorem}
	\label{thm: main_DkS}
	Assume that \Cref{conj: main1} holds and that $\P\ne \NP$. Then for some $0<\eps\leq 1/2$, there is no efficient approximation algorithm for \DkS problem that achieves approximation factor $2^{(\log N)^{\eps}}$, where $N$ is the number of vertices in the input graph.
\end{theorem}

In fact we will prove a slightly more general theorem, that will be useful for us later.

\begin{theorem}
	\label{thm: main_DkS2}
	Suppose there is an algorithm for the \DkS problem, that, given an instance $\pdks(G,k)$ with $|V(G)|=N$, in time at most $T(N)$, computes a factor $2^{(\log N)^{\eps}}$-approximate solution to the problem, for some constant $0<\eps\le 1/2$. Then there is an algorithm, that, given an instance $\iset$ of $(d(n),s(n))$-\LDCSP problem of size $n$, where $d(n)=2^{(\log n)^{\eps}}$ and $s(n)= 1/2^{64 (\log n)^{1/2+\eps}}$, responds ``YES'' or ''NO'', in time $O(\poly(n)\cdot T(\poly(n)))$. If $\iset$ is a \yi, the algorithm is guaranteed to respond ``YES'', and if it is a \ni, it is guaranteed to respond ``NO''. 
\end{theorem}

%
%
\Cref{thm: main_DkS} immediately follows from \Cref{thm: main_DkS2}.
The remainder of this section is dedicated to proving \Cref{thm: main_DkS2}. A central notion that we use is a \emph{constraint graph} that is associated with an instance $\iset$ of $2$-CSP.

\paragraph{Constraint Graph.}

Let $\iset=(X,Y,A,\cset,\set{G_C}_{C\in \cset})$ be an instance of the \BCSP problem. 
The \emph{constraint graph} associated with instance $\iset$ is denoted by $H(\iset)$, and it is defined as follows.
The set of vertices of $H(\iset)$ is the union of two subsets: set $V=\set{v(x)\mid x\in X}$ of vertices representing the variables of $X$, and set $U=\set{v(y)\mid y\in Y}$ of vertices representing the variables of $Y$. For convenience, we will not distinguish between the vertices of $V$ and the variables of $X$, so we will identify each variable $x\in X$ with its corresponding vertex $v(x)$. Similarly, we will not distinguish between vertices of $U$ and variables of $Y$. The set of edges of $H(\iset)$ contains, for every constraint $C=C(x,y)\in \cset$, edge $e_C=(x,y)$. We say that edge $e_C$ \emph{represents} the constraint $C$.
Notice that, if $E'$ is a subset of edges of $H(\iset)$, then we can define a set $\Phi(E')\subseteq \cset$ of constraints that the edges of $E'$ represent, namely: $\Phi(E')=\set{C\in \cset\mid e_C\in E'}$.
Next, we define bad sets of constraints and bad collections of edges.

\begin{definition}[Bad Set of Constraints and Bad Collection of Edges]
Let $\cset'\subseteq \cset$ be a collection of constraints of an instance $\iset=(X,Y,A,\cset,\set{G_C}_{C\in \cset})$ of \BCSP. We say that $\cset'$ is a \emph{bad set of constraints} if the largest number of constraints of $\cset'$ that can be simultaneously satisfied by any assignment to the variables of $X\cup Y$ is at most $\frac{|\cset'|}{4}$.
If $E'\subseteq E(H(\iset))$ is a set of edges of $H(\iset)$, whose corresponding set $\Phi(E')$ of constraints is bad, then we say that $E'$ is a \emph{bad collection of edges}.
\end{definition}

The next observation easily follows from the definition of a bad set of constraints.

\begin{observation}\label{obs: bad constraints}
	Let $\iset=(X,Y,A,\cset,\set{G_C}_{C\in \cset})$ be an instance of bipartite 2-CSP, and let $\cset',\cset''\subseteq \cset$ be two disjoint sets of constraints that are both bad. Then $\cset'\cup \cset''$ is also a bad set of constraints.
\end{observation}

Next, we define good subsets of constraints and good subgraphs of the constraint graph $H(\iset)$.

\begin{definition}[Good Set of Constraints and Good Subgraphs of $H(\iset)$]
	Let $\cset'\subseteq \cset$ be a collection of constraints of an instance $\iset=(X,Y,A,\cset,\set{G_C}_{C\in \cset})$ of \BCSP, and let $0<\beta\leq 1$ be a parameter. We say that $\cset'$ is a \emph{$\beta$-good set of constraints}, if there is an assigmnet to variables of $X\cup Y$ that satisfies at least  $\frac{|\cset'|}{\beta}$ constraints of $\cset'$.
	If $E'\subseteq E(H(\iset))$ is a set of edges of $H(\iset)$, whose corresponding set $\Phi(E')$ of constraints is $\beta$-good, then we say that $E'$ is a \emph{$\beta$-good collection of edges}. Lastly, if $H'\subseteq H(\iset)$ is a subgraph of the constraint graph, and the set $E(H')$ of edges is $\beta$-good, then we say that graph $H'$ is \emph{$\beta$-good}.
\end{definition}

The next observation easily follows from the definition of a good set of constraints.

\begin{observation}\label{obs: good constraints}
	Let $\iset=(X,Y,A,\cset,\set{G_C}_{C\in \cset})$ be an instance of bipartite 2-CSP, let $0<\beta\leq 1$ be a parameter, and let $H',H''$ be two subgraphs of $H(\iset)$ that are both $\beta$-good and disjoint in their vertices. Then graph $H'\cup H''$ is also $\beta$-good.
\end{observation}

The observation follows from the fact that, since graphs $H',H''$ are disjoint in their vertices, if we let $\cset'=\Phi(E(H'))$, $\cset''=\Phi(E(H''))$ be the sets of constraints associated with the edge sets of both graphs, then the variables participating in the constraints of $\cset'$ are disjoint from the variables participating in the constraints of $\cset''$.

The following theorem is key in proving \Cref{thm: main_DkS2}.

\begin{theorem}\label{thm: DkS reduction main}
	Assume that there exists a constant $0<\eps\leq 1/2$, and an $\alpha(N)$-approximation algorithm $\aset$ for the \DkS problem, whose running time is at most $T(N)$, where $N$ is the number of vertices in the input graph, and $\alpha(N)=2^{(\log N)^{\eps}}$. Then there is an algorithm, whose input consists of an instance $\iset=(X,Y,A,\cset,\set{G_C}_{C\in \cset})$ of \BCSP and parameter $n$ that is greater than a large enough constant, so that $\size(\iset)\leq n$ holds, and $\iset$ is a $d(n)$-to-$d(n)$ instance of \BCSP, for $d(n)\leq 2^{(\log n)^{\eps}}$.
	Let $\beta=2^{8(\log n)^{1/2+\eps}}$, and let $r=\ceil{\beta\cdot \log n}$. The algorithm  returns a partition $(E^b,E_1,\ldots,E_r)$ of $E(H(\iset))$, such that $E^b$ is a bad set of edges, and for all $1\leq i\leq r$, set $E_i$ of edges is $\beta^3$-good. The running time of the algorithm is $O(T(\poly(n))\cdot \poly(n)$.
\end{theorem}

The proof of \Cref{thm: main_DkS2} easily follows from \Cref{thm: DkS reduction main}. 	Assume that there exists a constant $0<\eps\leq 1/2$, and an $\alpha(N)$-approximation algorithm $\aset$ for the \DkS problem, whose running time is at most $T(N)$, where $N$ is the number of vertices in the input graph, and $\alpha(N)=2^{(\log N)^{\eps}}$. We show an algorithm for the  $(d(n),s(n))$-\LDCSP problem, for  $d(n)=2^{(\log n)^{\eps}}$ and $s(n)= 1/2^{64 (\log n)^{1/2+\eps}}$. Let $\iset=(X,Y,A,\cset,\set{G_C}_{C\in \cset})$ be an input instance of the 
 \BCSP problem, with $\size(\iset)=n$, so that $\iset$ is a $d(n)$-to-$d(n)$ instance of \BCSP, for $d(n)\leq 2^{(\log n)^{\eps}}$. If $n$ is bounded by a constant, then we can determine whether $\iset$ is a \yi or a \ni by exhaustively trying all assignments to its variables. Therefore, we assume that $n$ is greater than a large enough constant.
 We apply the algorithm from \Cref{thm: DkS reduction main} to this instance $\iset$. Let $(E^b,E_1,\ldots,E_r)$ be the partition of the edges of $E(H(\iset))$ that the algorithm returns. We now consider two cases.
 
 Assume first that $|E^b|> 2|\cset|/3$. Let $\cset^b\subseteq \cset$ be the set of all constraints that correspond to the edges of $E^b$. Recall that set $\cset^b$ of constraints is bad, so in any assignment, at most $\frac{|\cset'|}{4}$ of the constraints in $\cset$ may be satisfied. Therefore, if $f$ is any assignment to variables of $X\cup Y$, the number of constraints in $\cset$ that are not satisfied by $f$ is at least $\frac{3|\cset'|}{4}> \frac{|\cset|}{2}$. Clearly, $\iset$ may not be a \yi in this case. Therefore, if $|E^b|>2|\cset|/3$, we report that $\iset$ is a \ni.
 
 If $|E^b|\leq 2|\cset|/3$, then we report that $\iset$ is a \yi. It is now enough to show that, if  $|E^b|\leq 2|\cset|/3$, then instance $\iset$ may not be a \ni. In other words, it is enough to show that there is an assignment that satisfies more than $\frac{|\cset|}{2^{64 (\log n)^{1/2+\eps}}}$ constraints. Indeed, since $|E^b|\leq 2|\cset|/3$, there is an index $1\leq i\leq r$, with $|E_i|\geq \frac{|\cset|}{3r}$. Since set $E_i$ of edges is $\beta^3$-good, there is an assignment to the variables of $X\cup Y$, that satisfies at least $\frac{|E_i|}{\beta^3}\geq\frac{|\cset|}{3r\beta^3}$ constraints that correspond to the edges of $E_i$. Recall that $\beta=2^{8(\log n)^{1/2+\eps}}$ and $r=\ceil{\beta\cdot \log n}$. Therefore, $3r\beta^3\leq 6\beta^4\log n\leq 2^{64(\log n)^{1/2+\eps}}$. We conclude that there is an assignment satisfying at least $|\cset|/2^{64(\log n)^{1/2+\eps}}$ constraints, and so $\iset$ may not be a \ni. It is easy to verify that the running time of the algorithm is $O(T(\poly(n))\cdot \poly(n)$.
 
 To conclude, we have shown that, if there is an  $\alpha(N)$-approximation algorithm $\aset$ for the \DkS problem, with running time at most $T(N)$, where $N$ is the number of vertices in the input graph, and $\alpha(N)=2^{(\log N)^{\eps}}$, then there is an  algorithm for the  $(d(n),s(n))$-\LDCSP problem, for $d(n)=2^{(\log n)^{\eps}}$ and $s(n)= 1/2^{8 (\log n)^{1/2+\eps}}$, whose running time is $O(T(\poly(n))\cdot \poly(n)$.

In the remainder of this section we prove \Cref{thm: DkS reduction main}.

\subsection{Proof of \Cref{thm: DkS reduction main}}

The following theorem is the main technical ingredient of the proof of \Cref{thm: DkS reduction main}.
\begin{theorem}\label{thm: partition or solution for no instance}
	Assume that there exists an $\alpha(N)$-approximation algorithm $\aset$ for the \BDkS problem, whose running time is at most $T(N)$,  where $N$ is the number of vertices in the input graph. Then there is an algorithm, that, given an instance $\iset=(X,Y,A,\cset,\set{G_C}_{C\in \cset})$ of \BCSP and parameters $n,\beta\geq 1$, so that $\size(\iset)\leq n$,  $\beta \geq 2^{30}(\alpha(n))^3(\log n)^{12}$, and $\iset$ is a $d(n)$-to-$d(n)$ instance of \BCSP, for some function $d(n)$, in time $O(T(n)\cdot\poly(n))$, does one of the following:
	
	\begin{itemize}
		\item either correctly establishes that graph $H(\iset)$ is $\beta^3$-good; or
		
		\item computes a bad set $\cset'\subseteq \cset$ of constraints, with $|\cset'|\geq \frac{|\cset|}{8\log^2n}$; or
		\item  computes a subgraph $H'=(X',Y',E')$ of $H(\iset)$, for which the following hold:
		\begin{itemize}
			\item $|X'|\leq \frac{2d(n)\cdot |X|}{\beta}$;
			\item $|Y'|\leq \frac{2d(n)\cdot |Y|}{\beta}$; and
			\item $|E'|\geq \frac{\vol_H(X'\cup Y')}{2048d(n)\cdot \alpha(n)\cdot \log^4n}$.
		\end{itemize}
	\end{itemize} 
\end{theorem}

We prove \Cref{thm: partition or solution for no instance} in Section \Cref{subsec: proof of inner thm for DkS}, after we complete the proof of \Cref{thm: DkS reduction main} using it.
We start with the following corollary of \Cref{thm: partition or solution for no instance}.

\begin{corollary}\label{thm: general partition of solution}
	Assume that there exists an $\alpha(N)$-approximation algorithm $\aset$ for the \DS problem, whose running time is at most $T(N)$, where $N$ is the number of vertices in the input graph. Then there is an algorithm, whose input consists of an instance $\iset=(X,Y,A,\cset,\set{G_C}_{C\in \cset})$ of \BCSP and parameters $n,\beta\geq 1$, so that $\size(\iset)\leq n$,  $\beta \geq 2^{30}(\alpha(n))^3(\log n)^{12}$, and $\iset$ is a $d(n)$-to-$d(n)$ instance of \BCSP. The algorithm  returns a partition $(E_1,E_2)$ of $E(H(\iset))$, where $E_1$ is a bad set of edges, and:

\begin{itemize}
	\item either the algorithm correctly certifies that $E_2$ is a $\beta^3$-good set of edges; or

\item it computes a subgraph $H'=(X',Y',E')$ of $H(\iset)$, with $E(H')\subseteq E_2$, for which the following hold:
	\begin{itemize}
		\item $|X'|\leq \frac{2d(n)\cdot |X|}{\beta}$;
		\item $|Y'|\leq \frac{2d(n)\cdot |Y|}{\beta}$; and
		\item $|E'|\geq \frac{|E^*_2|}{2048d(n)\cdot \alpha(n)\cdot \log^4n}$, where $E^*_2$ is a set of edges containing every edge $e\in E_2$ with exactly one endpoint in $V(H')$.
	\end{itemize}
\end{itemize} 
The running time of the algorithm is $O(T(n)\cdot \poly(n))$.
\end{corollary}

\begin{proof}
The algorithm is iterative. We start with $E_1=\emptyset$, $E_2=E(H(\iset))$ and $H=H(\iset)$. We then iterate. In every iteration, we compute a graph $H'=H\setminus E_1$. We denote by $\cset'=\Phi(E(H'))$ the set of all constraints of $\cset$ corresponding to the edges of $H'$.  Notice that graph $H'$ naturally defines a $d(n)$-to-$d(n)$ instance $\iset'$ of \BCSP, whose size is at most $n$, that corresponds to the subset $\cset'\subseteq \cset$ of constraints. 
We apply the algorithm from \Cref{thm: partition or solution for no instance} to instance $\iset'$. If the outcome of the algorithm is a  bad set $\cset''\subseteq \cset'$ of constraints, then we let $\tilde E=\set{e_C\mid C\in \cset''}$ be the set of edges of $H'$ corresponding to the constraints of $\cset''$. We add the edges of $\tilde E$ to $E_1$, remove them from $E_2$, and continue to the next iteration.

If the algorithm from \Cref{thm: partition or solution for no instance} certifies that graph $H'$ is $\beta^3$-good, then we terminate the algorithm with the current partition $(E_1,E_2)$ of $E(H)$, and certify that the set $E_2$ of edges is $\beta^3$-good.

Otherwise, the outcome of the algorithm from \Cref{thm: partition or solution for no instance} must be a subgraph $H''=(X',Y',E')$ of $H'$, with $|X'|\leq \frac{2d(n)\cdot |X|}{\beta}$ and $|Y'|\leq \frac{2d(n)\cdot |Y|}{\beta}$. The algorithm also guarantees that $|E'|\geq \frac{\vol_{H'}(X'\cup Y')}{2048d(n)\cdot \alpha(n)\cdot \log^4n}$.

Let $E^*_2$ be the set of edges containing every edge $e\in E_2$ with exactly one endpoint in $V(H'')$. Since $E(H')=E_2$, it is immediate to verify that $|E^*_2|\leq \vol_{H'}(X'\cup Y')$. Therefore, we are guaranteed that $|E'|\geq \frac{|E^*_2|}{2048d(n)\cdot \alpha(n)\cdot \log^4n}$. We return the current partition $(E_1,E_2)$ of $E(H)$ and subgraph $H''$ of $H(\iset)$, and terminate the algorithm.

It is easy to verify that the algorithm consists of at most $O(\poly(n))$ iterations, and the running time of each iteration is at most $O(T(n)\cdot \poly(n))$. Therefore, the total running time of the algorithm is at most $O(T(n)\cdot \poly(n))$.
\end{proof}

Next, we obtain the following corollary.

\begin{corollary}\label{cor: a phase outer DkS}
	Assume that there exists a constant $0<\eps\leq 1/2$, and an  $\alpha(N)$-approximation algorithm $\aset$ for the \BDkS problem, whose running time is at most $T(N)$, where $N$ is the number of vertices in the input graph, and  $\alpha(N)=2^{(4\log N)^{\eps}}$. Then there is an algorithm, whose input consists of an instance $\iset=(X,Y,A,\cset,\set{G_C}_{C\in \cset})$ of \BCSP and parameter $n$ that is greater than a large enough constant, so that $\size(\iset)\leq n$ holds, and $\iset$ is a $d(n)$-to-$d(n)$ instance of \BCSP, for $d(n)\leq 2^{(\log n)^{\eps}}$.
	Let $\beta=2^{8(\log n)^{1/2+\eps}}$. The algorithm  returns a partition $(E_1,E_2,E_3)$ of $E(H(\iset))$, where $E_1$ is a bad set of constraints, $E_2$ is a $\beta^3$-good set of constraints, and $|E_1\cup E_2|\geq \frac{|E(H(\iset))|}{\beta}$. The running time of the algorithm is $O(T(n)\cdot \poly(n))$.
\end{corollary}

\begin{proof}
	Throughout the proof, we assume that there exists a constant $0<\eps\leq 1/2$, and an $\alpha(N)$-approximation algorithm $\aset$ for the \BDkS problem, whose running time is at most $T(N)$, where $N$ is the number of vertices in the input graph, and $\alpha(N)=2^{(2\log N)^{\eps}}$. 
Assume that we are given an instance $\iset=(X,Y,A,\cset,\set{G_C}_{C\in \cset})$  of \BCSP, together with a parameter $n$ that is greater than a large enough constant, so that  $\size(\iset)\leq n$, and $\iset$ is a $d(n)$-to-$d(n)$ instance of \BCSP, for $d(n)\leq 2^{(\log n)^{\eps}}$.
For convenience, we denote $H=H(\iset)$.  Our algorithm uses a parameter $\eta=2^{12} d(n)\cdot \alpha(n)\cdot \log^4n$.

The algorithm is iterative. Over the course of the algorithm, we maintain a collection $\hset$ of subgraphs of $H$, and another subgraph $H^g$ of $H$. We will ensure that, throughout the algorithm, all graphs in $\hset\cup \set{H^g}$ are mutually disjoint in their vertices. We denote by $E^g=E(H^g)$ and $E^1=\bigcup_{H'\in \hset}E(H')$. Additionally, we maintain another set $E^b$ of edges of $H$, that is disjoint from $E^g\cup E^1$, and we denote by $E^0=E(H)\setminus (E^g\cup E^b\cup E^1)$ the set of all remaining edges of $H$.
We ensure that the following invariants hold throughout the algorithm.

\begin{properties}{I}
	\item set $E^g=E(H^g)$ of edges is $\beta^3$-good; \label{inv: good edges}
	\item set $E^b$ of edges is bad;\label{inv: bad edges} and
	
	\item all graphs in $\hset\cup \set{H^g}$ are disjoint in their vertices. \label{inv: disjoint graphs}
	
\end{properties}

Intuitively, we will start with the set $\hset$ containing a single graph $H$, and $E^g=E^b=E^0=\emptyset$. As the algorithm progresses, we will iteratively add edges to sets $E^g,E^b$ and $E^0$, while partitioning the graphs in $\hset$ into smaller subgraphs. The algorithm will terminate once $\hset=\emptyset$. The key in the analysis of the algorithm is to ensure that $|E^0|$ is relatively small when the algorithm terminates. We do so via a charging scheme: we assign a budget to every edge of $E^1\cup E^g\cup E^b$, that evolves over the course of the algorithm, and we keep track of this budget over the course of the algorithm.

In order to define vertex budgets, we will assign, to every graph $H\in \hset$ a \emph{level}, that is an integer between $0$ and $\ceil{\log n}$. We will ensure that, throughout the algorithm, the following additional invariants hold:

\begin{properties}[3]{I}
\item If $H'\in \hset$ is a level-$i$ graph, then the budget of every edge $e\in E(H')$ is at most $\eta^i$; and \label{prop: charge each edge}

\item Throughout the algorithm's execution, the total budget of all edges in $E^g\cup E^b\cup E^1$ is at least $|E(H)|$.\label{prop: total budget}
\end{properties}

Intuitively, at the end of the algorithm, we will argue that the level of every graph in $\hset$ is not too large, and that the budget of every edge in $E^g\cup E^b\cup E^1$ is not too large. Since the total budget of all edges in  $E^g\cup E^b\cup E^1$ is at least $|E(H)|$, it will then follow that $|E^g\cup E^b\cup E^1|$ is sufficiently large.
We now proceed to describe the algorithm.

Our algorithm will repeatedly use the algorithm from \Cref{thm: general partition of solution}, with the same functions $\alpha(N), d(n)$, and parameter $\beta$. In order to be able to use the corollary, we need to estalish that $\beta \geq 2^{30}(\alpha(n))^3(\log n)^{12}$. This is immediate to verify since $\beta=2^{8(\log n)^{1/2+\eps}}$, $\alpha(n)=2^{(4\log n)^{\eps}}$, and $n$ is large enough.

\paragraph{Initialization.}
At the beginning of the algorithm, we set $E^0=E^g=E^b=\emptyset$, and we  let $\hset$ contain a single graph $H$, which is assigned level $0$. Note that $E^1=E(H)$ must hold. Every edge $e\in E(H)$ is assigned budget $b(e)=1$. Clearly, the total budget of all edges of $E^1\cup E^g\cup E^b$ is $B=\sum_{e\in E^1\cup E^g\cup E^b}b(e)=|E(H)|$.

The algorithm performs iterations, as long as $\hset\neq \emptyset$. In every iteration, we select an arbitrary graph $H'\in \hset$ to process. We now describe a single iteration.

\paragraph{Iteration description.}
We now describe an iteration where some graph $H'\in \hset$ is processed. We assume that graph $H'$ is assigned level $i$. Notice that graph $H'$ naturally defines an instance $\iset'=(X',Y',A,\cset',\set{G_C}_{C\in \cset'})$ of \BCSP, where $X'=V(H')\cap X$, $Y'=V(H')\cap Y$, $\cset'=\set{C\in \cset\mid e_C\in E(H')}$, and the graphs $G_C$ for constraints $C\in \cset'$ remain the same as in instance $\iset$. Clearly, $\size(\iset')\leq \size(\iset)\leq n$, and $H(\iset')=H'$. Furthermore, instance $\iset'$ remains a $d(n)$-to-$d(n)$ instance. We apply the algorithm from \Cref{thm: general partition of solution} to instance $\iset'$, with parameters $n$ and $\beta$ remaining unchanged. Consider the partition $(E_1,E_2)$ of $E(H')$ that the algorithm returns. Recall that the set $E_1$ of edges is bad. We add the edges of $E_1$ to set $E^b$. From Invariant \ref{inv: bad edges} and \Cref{obs: bad constraints}, set $E^b$ of edges continues to be bad. If the algorithm from \Cref{thm: general partition of solution} certified that $E_2$ is a $\beta^3$-good set of edges, then we update graph $H^g$ to be $H^g\cup (H'\setminus E_1)$, and we add the edges of $E_2$ to set $E^g$. We then remove graph $H'$ from $\hset$, and continue to the next iteration. Note that, from \Cref{obs: good constraints} and Invariants \ref{inv: good edges} and \ref{inv: disjoint graphs}, the set $E^g$ of edges continues to be $\beta^3$-good. It is easy to verify that all remaining invariants also continue to hold.

From now on we assume that the algorithm from \Cref{thm: general partition of solution} returned a subgraph $H''=(X'',Y'',E'')$ of $H'$, with $E''\subseteq E_2$, such that $|X''|\leq \frac{2d(n)\cdot |X'|}{\beta}$ and $|Y''|\leq \frac{2d(n)\cdot |Y'|}{\beta}$. In particular, $|V(H'')|=|X''|+|Y''|\leq \frac{2d(n)}{\beta}\cdot (|X'|+|Y'|)\leq  \frac{2d(n)}{\beta}\cdot |V(H')|$.
Additinally, if we denote by $E^*_2$ the subset of edges of $E_2$ containing all edges with exactly one endpoint in $X''\cup Y''$, then
$|E''|\geq \frac{|E^*_2|}{2048d(n)\cdot \alpha(n)\cdot \log^4n}$ must hold.
We let $H^*$ be the graph obtained from $H'\setminus E_1$, by deleting the vertices of $H''$ from it, so $V(H^*)\cup V(H'')=V(H')$, and $E(H^*)\cup E(H'')\cup E^*_2=E_2$. We remove graph $H'$ from $\hset$, and we add graphs $H''$ and $H^*$ to $\hset$, with graph $H''$ assigned level $(i+1)$, and graph $H^*$ assigned level $i$. We also add the edges of $E_2^*$ to $E^0$, and we update the set $E^1$ of edges to contain all edges of $\bigcup_{\tilde H\in \hset}E(\tilde H)$. Since we did not modify graph $H^g$ in the current iteration, it is immediate to verify that Invariants \ref{inv: good edges}--\ref{inv: disjoint graphs} continue to hold. Next, we update the budgets of edges, in order to ensure that 
Invariants \ref{prop: charge each edge} and \ref{prop: total budget} continue to hold. Intuitively, the edges of $E^*_2$ are now added to set $E^0$, so we need to distribute their budget among the edges of $E(H'')$, in order to ensure that the total budget of all edges in $E^g\cup E^b\cup E^1$ does not decrease. This will ensure that Invariant \ref{prop: total budget} continues to hold. At the same time, since the level of graph $H''$ is $(i+1)$, while the level of graph $H'$ was $i$, we can increase the budgets of the edges of $E(H')$ and still maintain Invariant \ref{prop: charge each edge}.

Formally, recall that \Cref{thm: general partition of solution} guarantees that $|E^*_2|\leq |E''|\cdot (2048d(n)\cdot  \alpha(n)\cdot \log^4n)=\frac{|E''|\cdot \eta} 2$. From Invariant \ref{prop: charge each edge}, the current budget of every edge in $E''\cup E^*_2$ is bounded by $\eta^i$. Therefore, at the beginning of the current iteration:

\[\sum_{e\in E''\cup E^*_2}b(e)\leq \eta^i\cdot \left(|E^*_2|+|E''|\right )\leq \eta^i\cdot |E''|\cdot \left(1+\frac{\eta}{2}\right )<\eta^{i+1}\cdot |E''|.
\]

We set the budget of every edge in $E''$ to be $\eta^{i+1}$, and leave the budgets of all other edges unchanged. It is easy to verify that $\bigcup_{e\in E^g\cup E^b\cup E^1}b(e)$ does not decrease in the current iteration, so Invariant \ref{prop: total budget} continues to hold. It is also easy to verify that Invariant \ref{prop: charge each edge} continues to hold. Therefore, all invariants continue to hold at the end of the iteration. This completes the description of an iteration.

The algorithm terminates when $\hset=\emptyset$. Clearly, we obtain a partition $(E^g,E^b,E^0)$ of $E(H)$ into disjoint subsets, where the set $E^b$ of edges is bad, and the set $E^g$ of edges is $\beta^3$-good. It remains to show that $|E^g\cup E^b|\geq \frac{|E(H)|}{\beta}$.
We use the edge budgets in order to prove this. Let $L^*$ be the largest level of any subgraph of $H$ that belonged to $\hset$ at any time during the algorithm. We start with the following key observation.

\begin{observation}\label{obs: graph level}
 $L^*\leq (\log n)^{1/2-\eps}$.
\end{observation}
\begin{proof}
	Consider any graph $H''$ that was added to set $\hset$ at any time during the algorithm's execution, and assume that $H''$ was assigned level $i$. Consider the iteration during which $H''$ was added to $\hset$, and let $H'\in \hset$ be the graph that was processed during that iteration. We refer to graph $H'$ as the \emph{parent-graph} of $H''$. Note that the level of $H'$ is either $i$ or $(i-1)$. Assume that it is the latter. Then, from the algorithm's description, $|V(H'')|\leq  \frac{2d(n)}{\beta}\cdot |V(H')|$ must hold.
	
	We can now construct a \emph{partitioning tree}, that contains a vertex $v(H')$ for every graph $H'$ that was ever present in $\hset$ over the course of the algorithm, an an edge between vertices $v(H')$ and $v(H'')$ whenever graph $H'$ is a parent-graph of graph $H''$. The root of the tree is $v(H)$. Consider now again some graph $H''$, and the unique path $P$ in the partitioning tree, connecting $v(H)$ to $v(H'')$. Denote the vertices on this path by $v(H)=v(H_0),v(H_1),\ldots,v(H_r)=v(H'')$, and assume that these vertices appear on path $P$ in this order. For all $1\leq i\leq r$, denote the level of graph $H_i$ by $L_i$. Then $0=L_1\leq L_2\leq\cdots\leq L_r$ must hold. Moreover, for every index $0<i\leq r$, either $L_i=L_{i-1}$; or $L_i=L_{i-1}+1$ hold. In the latter case, $|V(H_i)|\leq |V(H_{i-1})|\cdot \frac{2d(n)}{\beta}$. Denote $\Delta= \frac{\log n}{\log\left(\frac{\beta}{2d(n)}\right )}$. We claim that $L_r\leq \Delta$. Indeed, assume for contradiction that $L_r>\Delta$. Then there is a collection $J\subseteq \set{1,\ldots,r}$ of at least $\Delta+1$ indices $i$, for which $L_i=L_{i-1}+1$. But then:
	
	\[|V(H'')|\leq n\cdot\left(\frac{2d(n)}{\beta}\right )^{\Delta+1}<1, \]
	
	a contradiction.
	We conclude that $L^*\leq \frac{\log n}{\log\left(\frac{\beta}{2d(n)}\right )}$. Substituting $\beta=2^{8(\log n)^{1/2+\eps}}$ and $d(n)\leq 2^{(\log n)^{\eps}}$, we get that:
	
	\[ L^*\leq \frac{\log n}{\log\left(2^{6(\log n)^{1/2+\eps}}\right )}\leq (\log n)^{1/2-\eps}. \]
	
	\end{proof}


 From Invariant \ref{prop: charge each edge}, throughout the algorithm, for every edge $e\in E^1$, $b(e)\leq \eta^{L^*}$ must hold. Once an edge is added to $E^b\cup E^g$, its budget does not change. Therefore, at the end of the algorithm, the budget of every edge in $ E^g\cup E^b$ is at most $\eta^{L^*}$. On the other hand, from Invariant \ref{prop: total budget}, at the end of the algorithm, the total budget of all edges in $E^1\cup E^g\cup E^b$ is at least $|E(H)|$. Therefore, at the end of the algorithm:

\[|E^g\cup E^b|\geq \frac{|E(H)|}{\eta^{L^*}}. \]

We now bound $\eta^{L^*}$. Recall that $\eta=2^{12} d(n)\cdot \alpha(n)\cdot \log^4n\leq 2^{4(\log n)^{\eps}}$, since
$d(n)\leq 2^{(\log n)^{\eps}}$,
$\alpha(n)=2^{(4\log n)^{\eps}}$, and $n$ is large enough. Since, from \Cref{obs: graph level}, $L^*\leq (\log n)^{1/2-\eps}$, we get that $\eta^{L^*}\leq 2^{4(\log n)^{1/2}}<\beta$, since $\beta=2^{8(\log n)^{1/2+\eps}}$. Therefore, $|E^g\cup E^b|\geq |E(H)|/\beta$ as required.

Lastly, it is easy to verify that the algorithm has at most $\poly(n)$ iterations, and the running time of each iteration is bounded by $O(T(n)\cdot\poly(n))$, so the total running time of the algorithm is at most $O(T(n)\cdot\poly(n))$.
\end{proof}

We are now ready to complete the proof of \Cref{thm: DkS reduction main}.
	Assume that there exists a constant $0<\eps\leq 1/2$, and an  $\alpha(N)$-approximation algorithm $\aset$ for the \DkS problem, whose running time is at most $T(N)$, where $N$ is the number of vertices in the input graph, and $\alpha(N)=2^{(\log N)^{\eps}}$. From \Cref{lem: DkS and Dk1k2S}, there exists an  $\alpha'(N)$-approximation algorithm $\aset$ for the \BDkS problem, where $N$ is the number of vertices in the input graph, and $\alpha'(N)\leq O(\alpha(N^2))\leq O\left(2^{(2\log N)^{\eps}}\right )$. The running time of the algorithm is at most $O(T(N^2)\cdot \poly(N))$. Denote $T'(N)= O(T(N^2)\cdot \poly(N))$ this bound on the running time of the algorithm, and let $\alpha''(N)=2^{(4\log N)^{\eps}}$. Then there is an $\alpha''(N)$-approximation algorithm for \BDkS with running time at most $O(T'(N))$.  Indeed, if $N$ is greater than a sufficiently large constant, then we can use Algorithm $\aset$, to obtain a solution whose approximation factor is $\alpha'(N)\leq O\left(2^{(2\log N)^{\eps}}\right )\leq 2^{(4\log N)^{\eps}}\leq \alpha''(N)$. Otherwise, we can solve the problem exactly via exhaustive search.

	Assume now that we are given an instance $\iset=(X,Y,A,\cset,\set{G_C}_{C\in \cset})$ of \BCSP and parameter $n$ that is greater than a large enough constant, so that $\size(\iset)\leq n$ holds, and $\iset$ is a $d(n)$-to-$d(n)$ instance of \BCSP, for $d(n)\leq 2^{(\log n)^{\eps}}$.
	Let $\beta=2^{8(\log n)^{1/2+\eps}}$, and let $r=\ceil{\beta\cdot \log n}$. For convenience, we denote $H=H(\iset)$. Initially, we set $E^b=\emptyset$. Our algorithm performs $r$ iterations, where for all $1\leq j\leq r$, in iteration $j$ we construct the set $E_j\subseteq E(H)$ of edges, that is $\beta^3$-good, and possibly adds some edges to set $E^b$. We ensure that, throughout the algorithm, the set $E^b$ of edges is bad.
	
	Initially, $E^b=\emptyset$. We now describe the $j$th iteration. We assume that sets $E_1,\ldots,E_{j-1}$ of edges of $H$ were already defined. We construct graph $H_j$, that is obtained from graph $H$, by deleting the edges of $E_1\cup\cdots\cup E_{j-1}\cup E^b$ from it. Notice that graph $H_j$ naturally defines an instance $\iset_j=
	(X,Y,A,\cset_j,\set{G_C}_{C\in \cset_j})$ of \BCSP, with $H_j=H(\iset_j)$, where $\cset_j=\set{C\in \cset\mid e_C\in E(H_j)}$. We apply the algorithm from  \Cref{cor: a phase outer DkS} to graph $H_j$, with parameters $n, \beta$, and $d(n)$ remaining unchanged. Consider a partition $(E^1,E^2,E^3)$ of $E(H_j)$ that the algorithm returns. We add the edges of $E^1$ to set $E^b$. Since both sets of edges are bad, from \Cref{obs: bad constraints}, set $E^b$ of edges continues to be bad. We also set $E_{j}=E^2$, which is guaranteed to be a $\beta^3$-good set of edges from \Cref{cor: a phase outer DkS}. Recall that \Cref{cor: a phase outer DkS} also guarantees that $|E^1\cup E^2|\geq |E(H_j)|/\beta$. We then continue to the next iteration.
	
	Since, from the above discussion, for all $1\leq j<r$, $|E(H_{j+1})|\leq \left (1-\frac{1}{\beta}\right )|E(H_j)|$, and since $r=\ceil{\beta\cdot \log n}$, at the end of the algorithm, we are guaranteed that the final collection $E^b,E_1,\ldots,E_r$ of subsets of edges indeed partitions $E(H)$.
	
	Notice that the running time of a single iteration is bounded by $O(T'(n)\cdot \poly(n))\leq O(T(\poly(n))\cdot \poly(n))$. Since the number of iterations is bounded by $\poly(n)$, the total running time of the algorithm is bounded by $O(T(\poly(n))\cdot \poly(n))$.

In order to complete the proof of \Cref{thm: DkS reduction main}, it is now enough to prove \Cref{thm: partition or solution for no instance}, which we do next.

\subsection{Proof of \Cref{thm: partition or solution for no instance}}
\label{subsec: proof of inner thm for DkS}

The proof partially relies on ideas and techniques from  \cite{NDP-grid-hardness}.
Assume that there exists an $\alpha(N)$-approximation algorithm $\aset$ for the \BDkS problem, whose running time is at most $T(N)$,  where $N$ is the number of vertices in the input graph. Let $\iset=(X,Y,A,\cset,\set{G_C}_{C\in \cset})$ be the input instance of \BCSP, with $\size(\iset)\leq n$.
For convenience, we denote $H=H(\iset)$. If $|E(H)|\leq \beta^3$, then  graph $H$ is $\beta^3$-good, since we can compute an assignment to the variables of $X\cup Y$ that satisfies at least one constraint of $\cset$. Therefore, we assume from now on that $|E(H)|>\beta^3$. We can also assume that graph $H$ contains no isolated vertices, as isolated vertices of $H$ correspond to variables that do not participate in any constraints, and can be discarded.

The proof consists of four steps. In the first step, in order to simplify the proof, we will regularize graph $H$, by computing a ``nice'' subgraph $\tH\subseteq H$. In the second step, we will define an \emph{assignment graph} associated with $\tilde H$, and we will use it in order to obtain an instance of \BDkS, to which algorithm $\aset$ will then be applied. In the last two steps, we will use the outcome of algorithm $\aset$ in order to either correctly establish that graph $H$ is $\beta^3$-good, or to compute a bad subset of constraints, or a subgraph $H'$ of $H$ as required. We now describe each of the three steps in turn.

\subsubsection{Step 1: Regularization}

In this step we will compute a subgraph $\tilde H$ of $H$ that has a convenient structure. We refer to graphs with such structure as \emph{nice subgraphs} of $H$, and define them next.

\begin{definition}[Nice Subgraph of $H$]
		Let $\tilde H=(\tilde X,\tilde Y,\tilde E)$ be a subgraph of $H$, and let $d_1,d_2\geq 1$ be parameters. We say that $\tilde H$ is a $(d_1,d_2)$-nice subgraph of $H$, if the following hold:
		\begin{itemize}
			\item For every vertex $x\in \tilde X$, $d_1\leq \deg_{H}(x)< 2d_1$; 
			
			\item For every vertex $y\in \tilde Y$, $d_2\leq \deg_{H}(y)< 8d_2\log n$ and
			$d_2\leq \deg_{\tilde H}(y)< 2d_2$; and
			
			\item $|\tilde E|\geq \frac{d_1}{4\log n}\cdot |\tilde X|$.
		
		\end{itemize}
		We say that $\tilde H$ is a \emph{nice} subgraph of $H'$ if it is a $(d_1,d_2)$-nice subgraph of $H$ for any pair $d_1,d_2\geq 1$ of parameters.
	\end{definition}

	The first step of our algorithm is summarized in the following claim, that allows us to compute a nice subgraph $\tilde H$ of $H$ that contains many edges of $H$.
	
	\begin{claim}\label{claim: regularizing constraint graph}
		There is an algorithm with running time $O(\poly(n))$, that computes parameters $d_1,d_2>1$, and a subgraph $\tilde H$ of $H$, such that $\tilde H$ is a $(d_1,d_2)$-nice subgraph of $H$, and  $|E(\tilde H)|\geq \frac{|E(H)|}{8\log^2n}$.
	\end{claim}

	\begin{proof}
		The proof uses standard regularization techniques, and consists of three steps. Denote $H=(X,Y,E)$. 
		
		In the first step, we partition the vertices of $X$ into  groups $S_0,\ldots,S_q$, for $q=\ceil{\log n}$, where for all $0\leq i\leq q$, $S_i=\set{x\in X\mid 2^i\leq \deg_{H}(x)<2^{i+1}}$.
		We also partition the set $E$ of edges into subsets $E_0,\ldots,E_q$, where for all $0\leq i\leq q$, set $E_i$ contains all edges $e\in E$ that are incident to vertices of $S_i$. Clearly, there is an index $0\leq i^*\leq q$, with $|E_{i^*}|\geq \frac{|E(H)|}{2\log n}$. We let $\tilde X=S_{i^*}$, and we let $H_1$ be the graph whose vertex set is $\tilde X\cup Y$, and edge set is $E_{i^*}$. We also define $d_1=2^{i^*}$. Clearly, for every vertex $x\in \tilde X$, $d_1\leq \deg_H(x)<2d_1$. This completes the first regularization step.


We now proceed to describe our second step, in which we consider the vertices of $y\in Y$ one by one. We say that a vertex $y\in Y$ is \emph{bad}, if $\deg_{H_1}(y)<\frac{\deg_{H}(y)}{4\log n}$. Let $Y''\subseteq Y$ be the set of all bad vertices, and let $Y'=Y\setminus Y''$ be the set of all remaining vertices of $Y$, that we refer to as \emph{good vertices}.
We use the following observation.

\begin{observation}\label{obs: few edges incident to bad vertices}
	$\sum_{y\in Y''}\deg_{H_1}(y)\leq \frac{|E(H_1)|}{2}$.
\end{observation}
\begin{proof}
	Since, for every bad vertex $y$, $\deg_{H_1}(y)<\frac{\deg_{H}(y)}{4\log n}$, we get that:
	
	\[\sum_{y\in Y''}\deg_{H_1}(y)<\sum_{y\in Y''}\frac{\deg_{H}(y)}{4\log n}\leq \frac{|E(H)|}{4\log n}. \]
	
	Since, as observed above, $|E(H_1)|\geq \frac{|E(H)|}{2\log n}$, the observation follows.
\end{proof}

We let $H_2$ be a graph that is obtained from $H_1$, by discarding the vertices of $Y''$ from it. Therefore, $V(H_2)=\tilde X\cup Y'$. Additionally, from \Cref{obs: few edges incident to bad vertices}, $|E(H_2)|\geq \frac{|E(H_1)|}{2}\geq \frac{|E(H)|}{4\log n}$.

Lastly, in our third step, we perform a geometric grouping of the vertices of $Y'$ by their degree in $H_2$. 
Specifically, we let $r=\ceil{\log n}$, and we partition the vertices of $Y'$ into sets $S'_0,\ldots,S'_r$, where for $0\leq j\leq r$, $S'_j=\set{y\in Y'\mid 2^j\leq \deg_{H_2}(y)<2^{j+1}}$.
As before, we also partition the set $E(H_2)$ of edges into subsets $E'_0,\ldots,E'_r$, where for $0\leq j\leq r$ set $E'_j$ contains all edges $e\in E(H_2)$ that are incident to the vertices of $S'_j$. As before, there must be an index $0\leq j^*\leq r$ with $|E'_{j^*}|\geq \frac{|E(H_2)|}{2\log n}\geq \frac{|E(H)|}{8\log^2n}$. We set $\tilde Y=S'_{j^*}$, $d_2=2^{j^*}$, and we let $\tilde H$ be the graph whose vertex set is $\tilde X\cup \tilde Y$, and edge set is $E'_{j^*}$. We now verify that this graph has all required properties.

First, as observed already, for every vertex $x\in \tilde X$, $d_1\leq \deg_H(x)<2d_1$. Let $\tilde E=E(\tilde H)$. As observed already, $|\tilde E|\geq \frac{|E(H)|}{8\log^2n}$. Moreover, since, for every vertex $x\in \tilde X$, $\deg_{H_1}(x)=\deg_H(x)\geq d_1$, we get that $|E(H_1)|\geq d_1\cdot |\tilde X|$, and so $|\tilde E|\geq \frac{|E(H_2)|}{2\log n}\geq \frac{|E(H_1)|}{4\log n}\geq \frac{d_1}{4\log n}\cdot |\tilde X|$.

Consider now some vertex $y\in \tilde Y$. From the definition of graph $\tilde H$, it is immediate to verify that $\deg_{\tilde H}(y)=\deg_{H_2}(y)$. Therefore, $d_2\leq \deg_{\tilde H}(y)< 2d_2$. 
Clearly, $\deg_H(y)\geq \deg_{\tilde H}(y)\geq d_2$. Lastly, since vertex $y$ is good, we get that:

\[\deg_{\tilde H}(y)=\deg_{H_2}(y)=\deg_{H_1}(y)\geq \frac{\deg_H(y)}{4\log n}.  \]

Since $\deg_{\tilde H}(y)< 2d_2$, we get that $\deg_H(y)\leq (4\log n)\deg_{\tilde H}(y)<8d_2\log n$. We conclude that
$d_2\leq \deg_{H}(y)< 8d_2\log n$.
\end{proof}

\subsubsection{Step 2: Assignment Graph and Reduction to \BDkS}

Recall that we have computed, in the first step, a subgraph $\tilde H=(\tilde X,\tilde Y,\tilde E)$ of the graph $H=H(\iset)$. Since every edge of $H$ is associated with a distinct constraint in $\cset$, we can define a collection $\tilde \cset\subseteq \cset$ of constraints corresponding to the edges of $\tilde H$: $\tilde \cset=\set{C\in \cset\mid e_C\in \tilde E}$.

Next, we define a bipartite graph $G=(U,V,\hat E)$, called \emph{assignment graph}, that is associated with graph $\tilde H$.
For every variable  $z\in \tilde X\cup \tilde Y$, we define a set $R(z)=\set{v(z,a)\mid 1\leq a\leq A}$ of vertices that represent the possible assignments to variable $z$. 
We then set $U=\bigcup_{x\in \tilde X}R(x)$, and $V=\bigcup_{y\in \tilde Y}R(y)$. The set of vertices of $G$ is defined to be $U\cup V$.

In order to define the edges, 
	consider any constraint $C=C(x,y)\in \tilde \cset$. We define a set $E(C)$ of at most $d(n)\cdot A$ edges corresponding to $C$, as follows: we add an edge between vertex $v(x,a)$ and vertex $v(y,a')$ to $E(C)$ if assignments $a$ to $x$ and $a'$ to $y$ satisfy the constraint $C$. Since instance $\iset$ is a $d(n)$-to-$d(n)$ instance of Bipartite 2-CSP, every vertex of $R(x)\cup R(y)$ is incident to at most $d(n)$ edges of $E(C)$.
	We then let $E(G)=\bigcup_{C\in \tilde \cset}E(C)$.

Let $k_1=|\tilde X|$ and $k_2=|\tilde Y|$. We can then view graph $G$, together with parameters $k_1$ and $k_2$ as an instance of the  \BDkS problem, $\pdks(G, k_1,k_2)$. Notice that $|V(G)|\leq |\tilde \cset|\cdot A^2\leq |\cset|\cdot A^2\leq \size(\iset)\leq n$. We apply Algorithm $\aset$ to instance $\pdks(G, k_1,k_2)$ of \BDkS, and we let $S$ be the solution that the algorithm returns. Denote $G'=G[S]$, so the value of the solution is $|E(G')|$. Assume first that $|E(G')|<\frac{|\tilde \cset|}{4 \alpha(n)}$. We use the following observation to show that, in this case, the set $\tilde \cset$ of constraints is bad.

\begin{observation}\label{obs: bad se of constraint}
	If the set $\tilde \cset$ of constraints is not bad, then the value of the optimal solution to instance $\pdks(G, k_1,k_2)$ of \BDkS is at least $\frac{|\tilde \cset|}{4}$.
\end{observation}	
\begin{proof}	
Assume that the set $\tilde \cset$ of constraints is not bad. Then there is an assignment $f$ to variables of $\tilde X\cup \tilde Y$ that satisfies more than $\frac{|\tilde \cset|}{4}$ constraints of $\tilde \cset$. For each variable $z\in \tilde X\cup \tilde Y$, we denote the corresponding assignment by $f(z)$. Let $S'$ be the set of vertices of $G$ that contains, for every variable $x\in \tilde X$, vertex $v(x,f(x))$, and for every variable $y\in \tilde Y$, vertex $v(y,f(y))$. Then $S'$ is a valid solution to instance $\pdks(G, k_1,k_2)$ of \BDkS. Moreover, for every constraint $C\in \tilde \cset$ that is satisfied by the assignment $f$, an edge of $E(C)$ must be contained in $G[S']$. Therefore, the value of solution $S'$ is at least $\frac{|\tilde \cset|}{4}$. 
\end{proof}

From \Cref{obs: bad se of constraint}, if the set $\tilde \cset$ of constraints is not bad, then algorithm $\aset$ must have returned a solution whose value is at least $\frac{|\tilde \cset|}{4 \alpha(n)}$. Therefore, if the value of the solution $S$ that the algorithm returns is less than $\frac{|\tilde \cset|}{4\cdot \alpha(n)}$, then we terminate the algorithm, and return $\tilde \cset$ as a bad set of constraints. Recall that, from \Cref{claim: regularizing constraint graph}, $|\tilde \cset|=|E(\tilde H)|\geq \frac{|E(H)|}{8\log^2n}=\frac{|\cset|}{8\log^2n}$.
From now on we assume that the value of the soution $S$ is at least $\frac{|\tilde \cset|}{4 \alpha(n)}$.
We will use the set $S$ of vertices of $G$, in order to either correctly certify that graph $H$ is $\beta^3$-good, or to compute a subgraph $H'\subseteq H$ with the required properties. It will be convenient for us to further regularize graph $G'$, which we do in the next step.

\subsubsection{Step 3: Further Regularization}
It would be convenient for us to further regularize graph $G'$, by computing a subgraph $G''\subseteq G'$, with $|E(G'')|$ roughly comparable to $|E(G')|$, that has the following additional properties. First, for all  variables $x\in \tilde X$ with $R(x)\cap V(G'')\neq \emptyset$, the cardinalities of the sets $R(x)\cap V(G'')$ of vertices are roughly equal to each other (to within factor $2$). Similarly, for all  variables $y\in \tilde Y$ with $R(y)\cap V(G'')\neq \emptyset$, the cardinalities of the sets $R(y)\cap V(G'')$ of vertices are roughly equal to each other. Lastly, for all constraints $C\in \tilde \cset$ with $E(C)\cap E(G'')\neq \emptyset$, 
the cardinalities of the sets $E(C)\cap E(G'')$ of edges are roughly equal to each other. In this step, we compute a subgraph $G''\subseteq G'$ with all these properties. The algorithm is summarized in the following claim.
	
\begin{claim}\label{claim: regularization}
There is an  algorithm with running time $O(\poly(n))$, that computes a subgraph $G''\subseteq G'$, subsets $X^*\subseteq \tilde X$, $Y^*\subseteq \tilde Y$ of variables, a subset $\cset^*\subseteq \tilde \cset$ of constraints, and integers $q,q',r\geq 1$, such that, if we denote, for every variable $z\in X\cup Y$,  $R'(z)= R(z)\cap V(G'')$, and for every constraint $C\in \cset$, $E'(C)= E(C)\cap E(G'')$, then the following hold:
		
\begin{itemize}
\item for every variable $x\in X^*$, $2^q\leq |R'(x)|< 2^{q+1}$, and for $x\not\in X^*$, $R'(x)=\emptyset$;
\item for every variable $y\in Y^*$, $2^{q'}\leq |R'(y)|< 2^{q'+1}$, and for $y\not\in Y^*$, $R'(y)=\emptyset$;
\item for every constraint $C\in \cset^*$, $2^r\leq |E'(C)|<2^{r+1}$, and for $C\not\in \cset^*$, $E'(C)=\emptyset$;  and
\item $|\cset^*|\geq \frac{|\tilde \cset|}{2^{r+6}\cdot \alpha(n)\cdot \log^3n}$.
\end{itemize}
\end{claim}
	
\begin{proof}
The proof follows a standard regularization process. Let $E^0=E(G')$, so that $|E^0|\geq \frac{|\tilde \cset|}{4\cdot \alpha(n)}$.

Our first step regularizes the variables of $\tilde X$. We group the variables of $\tilde X$ into sets $J_0,J_1,\ldots,J_{\ceil{\log A}}$. For all $0\leq i\leq \ceil{\log A}$, we let $J_i=\set{x\in \tilde X\mid 2^i\leq |R(x)\cap V(G')|<2^{i+1}}$.
Note that, if $R(x)\cap V(G')=\emptyset$, then variable $x$ does not belong to any of the groups that we have defined. We also partition the edges of $E^0$ into groups $E_0,E_1,\ldots,E_{\ceil{\log A}}$, where for all $0\leq i\leq \ceil{\log A}$, group $E_i$ contains all edges $e\in E^0$ that are incident to the vertices of $\bigcup_{x\in J_i}(R(x)\cap V(G'))$. It is easy to verify that $(E_0,\ldots,E_{\ceil{\log A}})$ is indeed a partition of the set $E^0$ of edges. Therefore, there is an index $0\leq q\leq \ceil{\log A}$ with $|E_q|\geq \frac{|E^0|}{2\log A}\geq \frac{|E^0|}{2\log n}$. We then set $X^*=J_q$. For each such variable $x\in X^*$, we set $R'(x)=R(x)\cap V(G')$, and for each variable $x\in \tilde X\setminus X^*$, we set $R'(x)=\emptyset$. We also let $E^1=E_q$. From the above discussion,  $|E^1|\geq \frac{|E^0|}{2\log n}$, and, for every variable $x\in X^*$, $2^q\leq |R'(x)|<2^{q+1}$. Note that all edges of $E^1$ are incident to vertices of $\bigcup_{x\in X^*}R'(x)$.

Our second step is to regularize the variables of $\tilde Y$, exactly as before. 
We group the variables of $\tilde Y$ into sets $J'_0,J'_1,\ldots,J'_{\ceil{\log A}}$. For all $0\leq i'\leq \ceil{\log A}$, we let $J'_{i'}=\set{y\in \tilde Y\mid 2^{i'}\leq |R(x)\cap V(G')|<2^{i'+1}}$.
As before, if $R(y)\cap V(G')=\emptyset$, then variable $y$ does not belong to any of the sets that we have defined. We also partition the edges of $E^1$ into sets $E'_0,E'_1,\ldots,E'_{\ceil{\log A}}$, where for all $0\leq i'\leq \ceil{\log A}$, set $E_{i'}$ contains all edges $e\in E^1$ that are incident to the vertices of $\bigcup_{y\in J_{i'}}(R(y)\cap V(G'))$. As before, $(E'_0,\ldots,E'_{\ceil{\log A}})$ is a partition of the set $E^1$ of edges. Therefore, there is an index $0\leq q'\leq \ceil{\log A}$ with $|E'_{q'}|\geq \frac{|E^1|}{2\log A}\geq \frac{|E^1|}{2\log n}\geq \frac{|E^0|}{4\log^2n}$. We then set $Y^*=J'_{q'}$. For each variable $y\in Y^*$, we let $R'(y)=R(y)\cap V(G')$, and for each variable $y\in \tilde Y\setminus Y^*$, we set $R'(y)=\emptyset$. We also let $E^2=E'_{q'}$. From the above discussion,  $|E^2|\geq \frac{|E^0|}{4\log^2 n}$, and, for every variable $y\in Y^*$, $2^{q'}\leq |R'(y)|<2^{y+1}$. Notice that every edge of $E^2$ is incident to a vertex of $\bigcup_{x\in X^*}R'(x)$ and a vertex of $\bigcup_{y\in Y^*}R'(y)$.
		
Our third and final step is to regularize the constraints. Recall that for each constraint $C=C(x,y)\in \tilde \cset$, we defined a set $E(C)$ of edges. Since $|R(x)|= A$ and $|R(y)|=A$, $|E(C)|\leq A^2\leq n$ must hold. We group all constraints $C\in \tilde \cset  $ into 
sets $\cset_0,\cset_1,\ldots,\cset_{\ceil{\log n}}$, where for all $0\leq j\leq \ceil{\log n}$, set $\cset_j$ contains all constraints $C\in \tilde \cset$ with $2^j\leq |E(C)\cap E^2|<2^{j+1}$. Note that, if $E(C)\cap E^2=\emptyset$, then constraint $C$ does not belong to any set. Next, we define a partition $E_0'',E_1'',\ldots,E_{\ceil{\log n}}''$ of the set $E^2$ of edges: for $0\leq j\leq \ceil{\log n}$, set $E''_j$ contains all edges in $\bigcup_{C\in \cset_j}(E(C)\cap E^2)$. It is easy to verify that $(E''_0,\ldots,E_{\ceil{\log n}}'')$ is indeed a partition of $E^2$. Therefore, there must be an index $0<r\leq \ceil{\log n}$, with $|E''_r|\geq \frac{|E^2|}{2\log n}\geq \frac{|E^0|}{8\log^3n}$. We let $\cset^*=\cset_r$ and $E^3=E''_r$.
Since every constraint $C\in \cset^*$ contributes at most $2^{r+1}$ edges to $E^3$, we get that:

\[|\cset^*|\geq \frac{|E^3|}{2^{r+1}} \geq \frac{|E^0|}{2^{r+4}\log^3n}\geq \frac{|\tilde \cset|}{2^{r+6} \cdot \alpha(n)\cdot \log^3n}. \]

 For every constraint $C\in \cset^*$, we let $E'(C)=E(C)\cap E^2$, and for every constraint $C\in \tilde \cset\setminus\cset^*$, we let $E'(C)=\emptyset$. We are now ready to define graph $G''$. Its vertex set is $\left (\bigcup_{x\in X^*}R'(X)\right )\cup \left (\bigcup_{y\in Y^*}R'(y)\right )$, and its edge set is $\bigcup_{C\in \cset^*}E'(C)=E^3$. 
It is immediate to verify, from the above discussion, that graph $G'$, sets $X^*,Y^*$ of variables, and set $\cset^*$ of constraints have all required properties.
\end{proof}

In the next observation, we establish some useful bounds on the cardinalities of the sets $X^*,Y^*$ of variables, and the set $\cset^*$ of constraints.
	
\begin{observation}\label{obs: bounds on sizes of sets}
All of the following bounds hold:

\begin{itemize}
	\item  $|X^*|\leq \frac{|\tilde  X|}{2^q}$; 
	\item 	$|Y^*|\leq \frac{|\tilde Y|}{2^{q'}}$;
	\item $|\cset^*|\leq 2d_1|X^*|\leq \frac{2d_1|\tilde X|}{2^q}$; and
	\item $|\cset^*|\leq 2d_2|Y^*|\leq \frac{2d_2|\tilde Y|}{2^{q'}}$.
\end{itemize}
\end{observation}
	
\begin{proof}
	From the definition of the instance $(G,k_1,k_2)$ of \BDkS, graph $G'$ contains at most $k_1=|\tilde X|$ vertices of $\bigcup_{x\in \tilde X}R(x)$, and at most $k_2=|\tilde Y|$ vertices of $\bigcup_{y\in \tilde Y}R(y)$. Since, for every variable $x\in X^*$, $R'(x)\subseteq V(G')$ and $|R'(x)|\geq 2^q$, we get that $|X^*|\leq \frac{|\tilde X|}{2^q}$. Similarly, $|Y^*|\leq \frac{|\tilde Y|}{2^{q'}}$. 
	
	Recall that, since $\tilde H$ is a nice subgraph of $H$, for every vertex $x\in \tilde H$, $\deg_{\tilde H}(x)\leq \deg_H(x)\leq 2d_1$, and so $x$ may participate in at most $2d_1$ constraints of $\tilde \cset$. Since $\cset^*\subseteq \tilde \cset$ and $X^*\subseteq \tilde X$, every variable $x\in X^*$ may participate in at most $2d_1$ constraints of $\cset^*$. Therefore, $|\cset^*|\leq 2d_1|X^*|$. 
	
	Similarly, since $\tilde H$ is a nice subgraph of $H$, for every vertex $y\in \tilde Y$, $\deg_{\tilde H}(y)\leq 2d_2$. Using the same arguments as before, $|\cset^*|\leq 2d_2|Y^*|$.
	\end{proof}

Recall that, from \Cref{claim: regularization}, $|\cset^*|\geq \frac{|\tilde \cset|}{2^{r+5}\cdot \alpha(n)\cdot \log^3n}$. Since graph $\tilde H$ is a nice subgraph of $H$, we get that, for every vertex $y\in \tilde Y$, $\deg_{\tilde H}(y)\geq d_2$. Therefore, $|\tilde \cset|\geq |\tilde Y|\cdot d_2$, and so:

\begin{equation}
|\cset^*|\geq \frac{d_2\cdot |\tilde Y|}{2^{r+6}\cdot \alpha(n)\cdot \log^3n}.  \label{eq2}
\end{equation}

Similarly, from the definition of a nice subgraph, $|\tilde \cset|\geq \frac{d_1}{4\log n}\cdot |\tilde X|$, and so:

\begin{equation}
|\cset^*|\geq \frac{d_1\cdot |\tilde X|}{2^{r+8}\cdot \alpha(n)\cdot \log^4n}.  \label{eq3}
\end{equation}

	Lastly, we show that both $2^q,2^{q'}$ are close to $2^r$, in the following corollary of \Cref{obs: bounds on sizes of sets}.
	
	\begin{corollary}\label{cor: q q' and r}
		The following inequalities hold:
		
		\begin{itemize}
			\item 	$\frac{2^r}{2d(n)}\leq 2^{q}\leq 2^{r+8}\cdot \alpha(n)\cdot \log^4n$; and
			\item $\frac{2^r}{2d(n)}\leq 2^{q'}\leq 2^{r+6}\cdot \alpha(n)\cdot \log^3n$.
		\end{itemize}
	\end{corollary}
	
	\begin{proof}
		Consider some constraint $C=C(x,y)\in \cset^*$, and recall that $|E'(C)|\geq 2^r$. From the definition of the $d$-to-$d$ instances, every vertex $v(x,a)\in R(x)$ may be incident to at most $d(n)$ edges of $E(C)$. Since all edges of $E'(C)$ are incident to vertices of $R'(x)$, and $|R'(x)|\leq 2^{q+1}$, we get that $|E'(C)|\leq d(n)\cdot |R'(x)|\leq d(n)\cdot 2^{q+1}$. We conclude that $2^r\leq d(n)\cdot 2^{q+1}$.
		
		Similarly, every vertex $v(y,a')\in R(y)$ may be incident to at most $d(n)$ edges of $E(C)$. Since all edges of $E'(C)$ are incident to vertices of $R'(y)$, and $|R'(y)|\leq 2^{q'+1}$, we get that $|E'(C)|\leq d(n)\cdot 2^{q'+1}$.
		This proves the inequalities $\frac{2^r}{2d(n)}\leq 2^{q}$ and $\frac{2^r}{2d(n)}\leq 2^{q'}$.

		Next, we prove that 	$2^{q'}\leq 2^{r+6}\cdot \alpha(n)\cdot \log^3n$. 
	Recall that, from Inequality \ref{eq2}, $|\cset^*|\geq	\frac{d_2\cdot |\tilde Y|}{2^{r+6}\cdot \alpha(n)\cdot \log^3n}$.
		On the other hand, from the definition of a nice subgraph, every variable $y\in Y^*$ may participate in at most $2d_2$ constraints of $\cset^*$, and, from \Cref{obs: bounds on sizes of sets}, 	$|Y^*|\leq \frac{|\tilde Y|}{2^{q'}}$. Therefore:

		\begin{equation}
		|\cset^*|\leq 2d_2\cdot |Y^*|\leq  \frac{2d_2|\tilde Y|}{2^{q'}}.
		\end{equation}
		
		Combining the two inequalities, we get that:	$2^{q'}\leq 2^{r+6}\cdot \alpha(n)\cdot \log^3n$.

		Lastly, we prove that 	$2^q\leq 2^{r+8}\cdot \alpha(n)\cdot \log^4n$. 
Recall that, from Inequality \ref{eq3}, 
$|\cset^*|\geq \frac{d_1\cdot |\tilde X|}{2^{r+7}\cdot \alpha(n)\cdot \log^4n}$ holds.
As before, from the definition of a nice subgraph, every variable $x\in X^*$ may participate in at most $2d_1$ constraints of $\cset^*$, and, from \Cref{obs: bounds on sizes of sets}, 	$|X^*|\leq \frac{|\tilde X|}{2^{q}}$. Therefore:

\begin{equation}
|\cset^*|\leq 2d_1\cdot |X^*|\leq  \frac{2d_1\cdot |\tilde X|}{2^{q}}.
\end{equation}

Combining the two inequalities together, we get that:	$2^{q}\leq 2^{r+8}\cdot \alpha(n)\cdot \log^4n$.		
\end{proof}

\subsubsection{Step 4: Certifying that $H$ is a Good Graph or Computing a Subgraph of $H$}

We consider two cases, depending on whether $2^r\leq \beta$ holds. We start by showing that, if $2^r\leq \beta$, then graph $H$ is $\beta^3$-good.

\begin{observation}\label{obs: good graph}
	If $2^r\leq \beta$, then graph $H$ is $\beta^3$-good.
\end{observation}
\begin{proof}
We show that there exists an assignment to variables of $X\cup Y$ that satisfies at least $|\cset|/\beta^3$ constraints of $\cset$. In order to do it, we show a randomized algorithm that computes assignments to variables of $X\cup Y$, such that the expected number of satisfied constraints is at least $|\cset|/\beta^3$.

The assignments are computed as follows. Consider a variable $x\in X$. If $x\not\in X^*$, then we assign to $x$ an arbitrary value from $[A]$. Assume now that $x\in X^*$. Recall that we have defined a set $R'(x)\subseteq R(x)$ of vertices, whose cardinality is at most $2^{q+1}$. Set $R'(x)$ of vertices naturally defines a collection $\hat A(x)=\set{a\in [A]\mid v(x,a)\in R'(x)}$ of assignments to variable $x$, with $|\hat A(x)|\leq 2^{q+1}$. We choose an assignment $a\in \hat A(x)$ uniformly at random, and assign value $a$ to $x$.

Assignments to variables of $Y$ are defined similarly. Consider a variable $y\in Y$. If $y\not\in Y^*$, then we assign to $y$ an arbitrary value from $[A]$. Assume now that $y\in Y^*$. Recall that we have defined a set $R'(y)\subseteq R(y)$ of vertices, whose cardinality is at most $2^{q'+1}$. Set $R'(y)$ of vertices naturally defines a collection $\hat A(y)=\set{a\in [A]\mid v(y,a)\in R'(y)}$ of assignments to variable $y$, with $|\hat A(y)|\leq 2^{q'+1}$. We choose an assignment $a'\in \hat A(y)$ uniformly at random, and assign value $a'$ to $y$.

Recall that we have computed, in \Cref{claim: regularization}, a collection $\cset^*\subseteq \cset$ of constraints, with $|\cset^*|\geq \frac{|\tilde \cset|}{2^{r+6}\cdot \alpha(n)\cdot \log^3n}$. Consider now any constraint $C=C(x,y)\in \cset^*$, and recall that $x\in X^*,y\in Y^*$ must hold. Recall that graph $G''$ contains a collection $E'(C)\subseteq E(C)$ of edges, with $|E'(C)|\geq 2^r$. Consider now any such edge $e=(v(x,a),v(y,A))$. We say that edge $e$ \emph{wins} if $x$ is assigned value $a$, and $y$ is assigned value $a'$. The probability that edge $e$ wins is at least $\frac{1}{2^{q+1}\cdot 2^{q'+1}}$. Notice that at most one edge of $E'(C)$ may win, and so the probability that any edge of $E'(C)$ wins is at least $\frac{|E'(C)|}{2^{q+1}\cdot 2^{q'+1}}\geq \frac{2^r}{2^{q+1}\cdot 2^{q'+1}}$. If any edge of $E'(C)$ wins, the constraint $C$ is satisfied by the assignment that the algorithm chooses. Therefore, the probability that a constraint $C\in \cset^*$ is satisfied is at least $ \frac{2^r}{2^{q+1}\cdot 2^{q'+1}}$.

Overall, the expected number of constraints that are satisfied by the assignment is at least:

\[\frac{|\cset^*|\cdot 2^r}{2^{q+1}\cdot 2^{q'+1}}\geq \frac{|\tilde \cset|}{2^{8}\cdot \alpha(n)\cdot 2^q\cdot 2^{q'}\cdot \log^3n}\]

Recall that, from \Cref{claim: regularizing constraint graph}, $|\tilde \cset|=|E(\tilde H)|\geq \frac{|E(H)|}{8\log^2n}=\frac{|\cset|}{8\log^2n}$, and, from \Cref{cor: q q' and r}, $2^q\cdot 2^{q'}\leq 2^{2r+14}\cdot (\alpha(n))^2\cdot \log^7n\leq 2^{14}\cdot \beta^2\cdot  (\alpha(n))^2\cdot \log^7n$, since we have assumed that $2^r\leq \beta$.
Therefore, the expected number of constraints of $\cset$ that are satisfied by the assignment is at least:

\[\frac{|\tilde \cset|}{2^{8}\cdot \alpha(n)\cdot 2^q\cdot 2^{q'}\cdot \log^3n}\geq \frac{|\cset|}{2^{11}\cdot \alpha(n)\cdot 2^q\cdot 2^{q'}\cdot \log^5n}\geq  \frac{|\cset|}{2^{25}\cdot \beta^2\cdot (\alpha(n))^3 \cdot \log^{12}n}\geq \frac{|\cset|}{\beta^3},\]

since $\beta \geq 2^{27}(\alpha(n))^3\log^{12}n$.

We conclude that there is an assignment to the variables of $X\cup Y$ that satisfies at least $|\cset|/\beta^3$ constraints of $\cset$, and so graph $H$ is $\beta^3$-good.
\end{proof}

If $2^r\leq \beta$, then we terminate the algorithm and report that graph $H$ is $\beta$-good. 

From now on we assume that $2^r>\beta$. In this case, we return a subgraph  a subgraph $H'=(X',Y',E')$ of $H(\iset)$, that is defined as follows: $X'=X^*$, $Y'=Y^*$, and $E'=E^*$. We now verify that this graph has all required properties.

Recall that, from \Cref{obs: bounds on sizes of sets}, $|X^*|\leq  \frac{|\tilde  X|}{2^q}\leq \frac{|X|}{2^q}$, from \Cref{cor: q q' and r}, $2^q\geq \frac{2^r}{2d(n)}$, and, from our assumption, $2^r>\beta$. Therefore:

\[|X^*|\leq  \frac{|X|}{2^q}\leq \frac{|X|\cdot 2d(n)}{2^r}\leq \frac{2d(n)}{\beta}\cdot |X|. \]

Similarly, from \Cref{obs: bounds on sizes of sets}, $|Y^*|\leq  \frac{|\tilde  Y|}{2^{q'}}\leq \frac{|Y|}{2^{q'}}$, from \Cref{cor: q q' and r}, $2^{q'}\geq \frac{2^r}{2d(n)}$, and, from our assumption, $2^r>\beta$. Therefore:

\[|Y^*|\leq  \frac{|Y|}{2^{q'}}\leq \frac{|Y|\cdot 2d(n)}{2^r}\leq \frac{2d(n)}{\beta}\cdot |Y|. \]

It now only remains to show that  $|E^*|\geq \frac{\vol_H(X^*\cup Y^*)}{256d(n)\cdot \alpha(n)\cdot \log^4n}$.

Recall that, from Inequality \ref{eq3}, 
$|\cset^*|\geq \frac{d_1\cdot |\tilde X|}{2^{r+8}\cdot \alpha(n)\cdot \log^4n}$ holds. Since, from \Cref{cor: q q' and r}, $2^r\leq 2^q\cdot 2d(n)$, we get that $|E^*|=|\cset^*|\geq \frac{d_1\cdot |\tilde X|}{2^{q+9}\cdot d(n)\cdot \alpha(n)\cdot \log^4n}$. At the same time, from the definition of a nice graph, for every vertex $x\in X^*$, $\deg_H(x)\leq 2d_1$, so $\vol_H(X^*)\leq 2d_1\cdot |X^*|\leq \frac{2d_1\cdot |\tilde X|}{2^q}$, since $|X^*|\leq \frac{|\tilde  X|}{2^q}$ from \Cref{obs: bounds on sizes of sets}. Therefore, 
$|E'|\geq \frac{\vol_H(X^*)}{1024d(n)\cdot \alpha(n)\cdot \log^4n}
$.

Similarly, from Inequality \ref{eq2}, $|\cset^*|\geq	\frac{d_2\cdot |\tilde Y|}{2^{r+6}\cdot \alpha(n)\cdot \log^3n}$.
Since, from \Cref{cor: q q' and r}, $2^r\leq 2^{q'}\cdot 2d(n)$, we get that $|E'|=|\cset^*|\geq \frac{d_2\cdot |\tilde Y|}{2^{q'+7}\cdot d(n)\cdot \alpha(n)\cdot \log^3n}$. At the same time, from the definition of a nice graph, for every vertex $y\in Y^*$, $\deg_H(y)\leq 2d_2$, so $\vol_H(Y^*)\leq 2d_2\cdot |Y^*|\leq \frac{2d_2\cdot |\tilde Y|}{2^{q'}}$, since $|Y^*|\leq \frac{|\tilde  Y|}{2^{q'}}$ from \Cref{obs: bounds on sizes of sets}. Therefore, 
$|E^*|\geq \frac{\vol_H(Y^*)}{1024d(n)\cdot \alpha(n)\cdot \log^3n}
$.

Altogether, we get that $|E'|\geq \frac{\vol_H(X^*\cup Y^*)}{2048d(n)\cdot \alpha(n)\cdot \log^4n}$.

Note that every step of the algorithm, except Step 2, has running time $O(\poly(n))$, while the running time of Step 2 is $O(T(n)+\poly(n))$. Therefore, the total running time of the algorithm is at most $O(T(n)\cdot \poly(n))$.

	
	
	
	
	
	

\section{Reductions from \DkC and \WGP to \DkS}
\label{sec: DkC and WGP to DkS}

In this section we provide reductions from the \DkC and \WGP problems to \DkS, by proving the following theorem.

\begin{theorem}
\label{thm: alg_DkS gives alg_DkC}
Let $\alpha: \mathbb{Z^+}\to \mathbb{Z^+}$ be an increasing function, such that $\alpha(n)\leq o(n)$. 
Assume that there is an efficient $\alpha(n)$-approximation algorithm for the \DkS problem, where $n$ is the number of vertices in the input graph. Then both of the following hold:

\begin{itemize}
\item there is an efficient randomized algorithm that, given an instance of \DkC whose graph contains $N$ vertices, with high probability computes an $O(\alpha(N^2)\cdot\poly\log N)$-approximate solution to this instance; and

\item there is an efficient randomized algorithm that, given an instance of \WGP whose graph contains $N$ vertices, with high probability computes an $O(\alpha(N^2)\cdot\poly\log N)$-approximate solution to this instance.
\end{itemize}
\end{theorem}

We provide the proof of the first assertion of the theorem, by showing a reduction from \DkC to \DkS. The proof of the second assertion is similar and is  deferred to Section  \ref{apd: WGP to DkS} of Appendix.
We start by considering an LP-relaxation of the \DkC problem, whose number of variables is at least $N\choose k$. Due to this high number of variables, we cannot solve it directly. We first show an algorithm, that, given an approximate fractional solution to this LP-relaxation, whose support size is polynomial in $N$, computes an approximate integral solution to the \DkC problem instance. We then show an efficient algorithm that computes an approximate solution to the LP-relaxation, whose support is relatively small. In order to do so, we design an approximate separation oracle to the dual LP of the LP-relaxation.

\subsection{An LP-Relaxation and Its Rounding}

Let $\pcl(G,k)$ be an input instance of the \DkC problem, and denote $|V(G)|=N$.
We let $\rset$ the collection of all subsets of $V(G)$ containing at most $k$ vertices, that is: $\rset=\set{S\subseteq V(G)\mid |S|\leq k}$. For every set $S\in \rset$ of vertices, we denote by $m(S)=|E_G(S)|$.
We consider the following LP-relaxation of the \DkC problem, that has a variable $x_S$ for every set $S\in \rset$ of vertices.

\begin{eqnarray*}
	\mbox{(LP-P)}&&\\
		\max &\sum_{S\in \rset} m(S)\cdot x_S&\\
	\mbox{s.t.}&&\\
	&\sum_{\stackrel{S\in \rset:}{v\in S}} x_S\leq 1 &\forall v\in V(G)\\
	&\sum_{S\in \rset}x_S \leq N/k\\
	&x_S\geq 0&\forall S\in \rset
\end{eqnarray*}

It is easy to verify that (LP-P) is an LP-relaxation of the \DkC problem. Indeed, consider a solution $(S_1,\ldots,S_r)$ to the input instance $\pcl(G,k)$, where $r=N/k$. For all $1\leq i\leq r$, we set $x_{S_i}=1$, and for every other set $S\in \rset$, we set $x_S=0$. This provides a feasible solution to (LP-P), whose value is precisely $\sum_{i=1}^r|E_G(S_i)|$. We denote the value of the optimal solution to (LP-P) by $\opt_{\textnormal{LP-P}}$. From the above discussion, $\opt_{\textnormal{LP-P}}\ge \optcl(G,k)$. 

Note that the number of variables in (LP-P) is at least $N\choose k$, and so we cannot solve it directly.  We will show below an efficient algorithm that provides an approximate solution to (LP-P), whose support is reasonably small. Before we do so, we provide an LP-rounding algorithm, by proving the following claim.

\begin{claim}\label{claim: LP-rounding}
	There is an efficient randomized algorithm, whose input consists of an instance $\pcl(G,k)$ of the \DkC problem with $N=|V(G)|$, such that $N$ is greater than a large enough constant, and a solution $\set{x_S\mid S\in \rset}$ to (LP-P), in which the number of variables $x_S$ with $x_S>0$ is bounded by $\poly(N)$, and $\sum_{S\in \rset}m(S)\cdot x_S\geq \opt_{\textnormal{LP-P}}/\beta$, for some parameter $1\leq \beta\leq N^3$; the solution is given by only specifying values of variables $x_S$ that are non-zero. The algorithm with high probability returns an integral solution $(S_1,\ldots,S_{N/k})$ to instance $\pcl(G,k)$, such that $\sum_{i=1}^{N/k}|E_G(S_i)|\geq \frac{\optcl(G,k)}{2000\beta \log^3N}$.
\end{claim}

\begin{proof}
We assume that we are given a solution  $\set{x_S\mid S\in \rset}$ to (LP-P), in which the number of variables $x_S$ with $x_S>0$ is bounded by $\poly(N)$. Denote $C=\sum_{S\in \rset}m(S)\cdot x_S$, so that $C\geq \frac{\opt_{\textnormal{LP-P}}}{\beta}\geq \frac{\optcl(G,k)}{\beta}$ holds. We denote by $\rset'\subseteq \rset$ the collection of all sets $S\in \rset$ with $x_S>0$. Assume first that there is any set $S\in \rset'$ with $m(S)\geq \frac{C}{300\log^3N}$. Then we can obtain the desired solution $(S_1,\ldots,S_{N/k})$  as follows. We start with $S_1=S$ and $S_2=\cdots=S_{N/k}=\emptyset$. We then iteratively add vertices of $V(G)\setminus S$ to sets $S_1,\ldots,S_{N/k}$ arbitrarily, to ensure that the cardinality of each such set is exactly $k$. It is immediate to verify that $\sum_{i=1}^{N/k}|E_G(S_i)|\geq \frac{C}{300\log^3N}\geq \frac{\optcl(G,k)}{300\beta \log^3N}$. Therefore, we assume from now on that, for every set $S\in \rset'$, $m(S)<\frac{C}{300\log^3N}$.

We construct another collection $\rset''\subseteq \rset'$ of subsets of vertices of $G$ as follows. For every set $S\in \rset'$, we add $S$ to $\rset''$ independently, with probability $x_S$. Clearly, $\expect{\sum_{S\in \rset''}m(S)}=\sum_{S\in \rset}m(S)\cdot x_S=C$.
	
We say that a bad event $\event_1$ happens if some vertex $v\in V(G)$ lies in more than $5\log N$ sets in $\rset''$. We say that a bad event $\event_2$ happens if $|\rset''|> (5N\log N)/k$. We say that a bad event $\event_3$ happens if $\sum_{S\in \rset''}m(S)<\frac{C}{8}$.
Lastly, we say that a bad event $\event$ happens if either of the events $\event_1, \event_2$, or $\event_3$ happen.
We start with the following simple observation.

\begin{observation}\label{obs: first bad event}
	$\prob{\event}\leq \frac{2}{N^3} $.
\end{observation}

\begin{proof}
	For every set $S\in \rset'$, let $Y_S$ be the random variable whose value is $1$ if $S\in \rset''$ and $0$ otherwise. 
	
	Consider any vertex $v\in V(G)$. Denote by $Z_v$ the number of vertex sets in $\rset''$ containing $v$. Clearly, $Z_v=\sum_{\stackrel{S\in \rset':}{v\in S}}Y_S$. Therefore: 
	
	$$\expect{Z_v}=\expect{\sum_{\stackrel{S\in \rset':}{v\in S}}Y_S}=\sum_{\stackrel{S\in \rset':}{v\in S}}x_S\leq 1,$$ 
		
		where the last inequality follows from the constraints of (LP-P). 
		
		By applying the Chernoff Bound from \Cref{lem: Chernoff}, we get that $\prob{Z_v>5\log N}\leq 1/N^5$. Using the union bound over all vertices $v\in V(G)$, we get that $\prob{\event_1}\leq 1/N^4$.

		Notice that $\expect{|\rset''|}=\expect{\sum_{S\in \rset}Y_S}=\sum_{S\in \rset}x_S\leq N/k$ from the constraints of (LP-P). Bad Event $\event_2$ happens if $|\rset''|=\sum_{S\in \rset}Y_S> (5N\log N)/k$. By applying the Chernoff Bound from \Cref{lem: Chernoff} to the variables of $\set{Y_S\mid S\in \rset'}$, we get that $\prob{\event_2}=\prob{\sum_{S\in \rset'}Y_S>(5N\log N)/k}\leq 1/N^5$.

Lastly, we bound the probability of Event $\event_3$.
Recall that we have assumed that, for every set $S\in \rset'$, $m(S)<\frac{C}{300\log^3N}$ holds. We partition the collection $\rset'$ of vertex subsets into $\rho=\ceil{2\log N}$ collections $\rset_0,\ldots,\rset_{\rho-1}$, as follows. For all $0\leq i<\rho$, we let $\rset_i=\set{S\in \rset'\mid 2^i\leq m(S)<2^{i+1}}$.
For all $0\leq i<\rho$, we denote $C_i=\sum_{S\in \rset_i}m(S)\cdot x_S$. Clearly, $\sum_{i=0}^{\rho-1}C_i=C$. We say that an index $0\leq i<\rho$ is \emph{bad}, if $C_i<\frac{C}{8\log N}$, and otherwise we say that $i$ is a \emph{good index}. Let $I^g,I^b\subseteq \set{0,\ldots,\rho-1}$ be the collections of good and bad indices, respectively. Since $\rho\leq 4\log N$, we get that:

\[ \sum_{i\in I^b}C_i\leq (4\log N)\cdot \frac{C}{8\log N}\leq \frac{C}2. \]

Therefore, $\sum_{i\in I^g}C_i\geq \frac{C}2$ holds. Consider now a good index $i\in I^g$. Recall that, for every set $S\in \rset'$, $m(S)\leq \frac{C}{300\log^3N}$ holds, and so $2^i\leq \frac{C}{300\log^3N}$ must hold. Moreover, since $\sum_{S\in \rset_i}m(S)\cdot x_S= C_i\geq \frac{C}{8\log N}$, we get that:

\[ \frac{C}{8\log N}\leq \sum_{S\in \rset_i}m(S)x_S\leq 2^{i+1}\cdot \sum_{S\in \rset_i}x_S\leq \frac{C}{150\log^3N}\cdot \sum_{S\in \rset_i}x_S.  \]
		
We conclude that $\sum_{S\in \rset_i}x_S\geq 16\log^2N$ holds for every good index $i\in I^g$.

For a good index $i\in I^g$, we denote by $\rset'_i=\rset_i\cap \rset''$, and we let $\tilde \event_i$ be the bad event that $|\rset'_i|<\frac{\sum_{S\in \rset_i}x_S}{2}$.
From the Chernoff Bound (\Cref{lem: Chernoff}), and the fact that $\sum_{S\in \rset_i}x_S\geq 16\log^2N$,  we get that $\prob{\tilde \event_i}\leq e^{-2\log^2N}<N^{-4}$, if $N$ is sufficiently large.
By applying the Union Bound to all indices $i\in I^g$, we get that the probability that any of the events in $\set{\tilde \event_i\mid i\in I^g}$ happens is bounded by $1/N^3$. Note that, if neither of the events in $\set{\tilde \event_i\mid i\in I^g}$ happen, then: 

\[
\begin{split}
\sum_{S\in \rset''}m(S)& \geq \sum_{i\in I^g}2^i\cdot |\rset'_i|\\
&\geq \sum_{i\in I^g}2^i\cdot \frac{\sum_{S\in \rset_i}x_S}{2}\\
&\geq \sum_{i\in I^g}\sum_{S\in \rset_i}\frac{m(S)\cdot x_S}{4}\\
&\geq \sum_{i\in I^g}\frac{C_i}{4}\\
&\geq \frac{C}{8}.
\end{split}
\]		

Therefore, if neither of the events in $\set{\tilde \event_i\mid i\in I^g}$ happen, then Event $\event_3$ also does not happen. We conclude that $\prob{\event_3}\leq 1/N^3$.

		Finally, from the Union Bound, we get that:

		$$\prob{\event}\leq \prob{\event_1}+\prob{\event_2}+\prob{\event_3}\leq \frac{1}{N^4}+\frac{1}{N^5}+\frac{1}{N^3}\leq \frac{2}{N^3}.$$
\end{proof}

Observe that we can efficiently check whether Event $\event$ happened. If Event $\event$ happens, then we terminate the algorithm with a FAIL. We assume from now on that Event $\event$ did not happen. 
In this case, we are guaranteed that $\sum_{S\in \rset''}m(S)\geq \frac{C}{8}\geq \frac{\optcl(G,k)}{8\beta}$.
We denote $\rset''=\set{S_1,S_2,\ldots,S_z}$, where the sets are indexed according to their value $m(S)$, so that $m(S_1)\geq m(S_2)\geq \cdots\geq m(S_z)$. We then let $\sset=\set{S_1,\ldots,S_{N/k}}$ (if $z<N/k$, then  we set $S_{z+1}=\cdots=S_{N/k}=\emptyset$). For all $1\leq i\leq N/k$, we denote $E_i=E_G(S_i)$, so $|E_i|=m(S_i)$, and we denote $E'=\bigcup_{i=1}^{N/k}E_i$. Recall that, since Event $\event$ did not happen, $|\rset''|\leq (5N\log N)/k$ holds. Therefore:

\[ |E'|\geq \frac{\sum_{S\in \rset''}m(S)}{5\log N}\geq \frac{\optcl(G,k)}{40\beta\log N}. \]

Note that the vertex sets in the family $\sset$ may not be mutually disjoint. However, since Event $\event$ did not happen, every vertex of $V(G)$ may lie in at most $5\log N$ such sets. We now construct a new collection $\sset'=\set{S_1',\ldots,S_{N/k}'}$ of sets of vertices, as follows. Consider any vertex $v\in V(G)$, and let $S_{i_1},S_{i_2},\ldots,S_{i_a}\in \sset$ be the sets of $\sset$ containing $v$. Vertex $v$ chooses an index $i^*\in \set{i_1,\ldots,i_a}$ at random, and is then added to $S'_{i^*}$.

Note that for all $1\leq j\leq N/k$, for every vertex $v\in S_j$, the probability that $v\in S'_j$ is at least $1/(5\log N)$. We say that an edge $e=(u,v)\in E_j$ \emph{survives} if both $u,v\in S'_j$. We denote by $E''\subseteq E'$ the set of all edges that survive. Since $\prob{u\in S'_j}\geq 1/(5\log N)$, $\prob{v\in S'_j}\geq 1/(5\log N)$, and the two events are independent, we get that the probability that edge $e$ survives is at least $1/(25\log^2N)$. Overall, we get that:

\[\expect{|E''|}\geq \frac{|E'|}{25\log^2N}\geq \frac{\optcl(G,k)}{1000\beta\log^3 N} .\]

We obtain a final solution $\sset^*$ to instance $\pcl(G,k)$ of the \DkC problem by starting with $\sset^*=\sset'$, and then partitioning the vertices of $V(G)\setminus \left(\bigcup_{j=1}^{N/k}S'_j\right )$ among the sets of $\sset^*$ arbitrarily, until each such set contains exactly $k$ vertices. Clearly, the value of the resulting solution is at least $|E''|$.

So far we have obtained a randomized algorithm that either returns FAIL (with probability at most $2/N^3$), or it returns a solution to instance $\pcl(G,k)$ of the \DkC problem, whose expected value is at least $\frac{\optcl(G,k)}{1000\beta\log^3 N}$.

Let $p'$ be the probability that the algorithm returned a solution of value at least   $\frac{\optcl(G,k)}{3000\beta\log^3 N}$, given that it did not return FAIL. Note that the expected solution value, assuming the algorithm did not return FAIL, is at most $\frac{\optcl(G,k)}{2000\beta\log^3 N}+p'\cdot \optcl(G,k)$. Since this expectation is also at least $\frac{\optcl(G,k)}{1000\beta\log^3 N}$, we get that $p'\geq \frac{1}{2000\beta\log^3N}$. Overall, the probability that our algorithm successfully returns a solution of value at least $\frac{\optcl(G,k)}{2000\beta\log^3 N}$ is $p'\cdot \prob{\neg\event}\geq \Omega\left(\frac{1}{\beta\log^3N}\right )$. By repeating the algorithm $\poly(N)$ times we can ensure that it successfully computes a solution of value at least $\frac{\optcl(G,k)}{2000\beta\log^3 N}$ with high probability.
\end{proof}

\subsection{Approximately Solving the LP-Relaxation}
\label{subsec: solve the LP}
As observed already,  (LP-P) has $\binom{N}{k}$ variables, and so we cannot solve it directly. Instead, we will use the Ellipsoids algorithm with an approximate separation oracle to its dual LP, that appears below. This LP has a variable $z$, and, additionally, for every vertex $v\in V(G)$, it has a variable $y_v$.

\begin{eqnarray*}
	\mbox{(LP-D)}&&\\
	\min& \frac{N}{k}\cdot z+\sum_{v\in V(G)} y_v\\
	\mbox{s.t.}&&\\
	& z+\sum_{v\in S} y_v\geq m(S) &\forall S\in \rset\\
	&z\ge 0\\
	&y_v\geq 0&\forall v\in V(G)
\end{eqnarray*}

We denote the value of the optimal solution to (LP-D) by $\opt_{\textnormal{LP-D}}$.

Next, we define an approximate separation oracle, and provide such a separation oracle to (LP-D).

\paragraph{Approximate Separation Oracle.}
Consider the following general minimization Linear Program, whose variables are $\set{x_1,\ldots,x_n}$.

\begin{eqnarray*}
	\mbox(P)&&\\
	\min& \sum_{i=1}^nc_ix_i&\\
	\mbox{s.t.}&&\\
	& \sum_{i=1}^nA_{j,i}x_{i}\ge b_j &\forall 1\le j\le m\\
	&x_i\geq 0 &\forall 1\le i\le n
\end{eqnarray*}


Let $\beta: \mathbb{Z^+}\to \mathbb{Z^+}$ be an increasing function. A \emph{randomized $\beta(n)$-approximate separation oracle} for (P) is an efficient randomized algorithm (that is, the running time of the algorithm is bounded by a polynomial function of its input size). The input to the algorithm is a set $\set{x_i}_{i=1}^n$ of non-negative real values. 
The algorithm either returns  ``accept'', or it returns an LP-constraint (called a \emph{violated constraint}) that does not hold for the given values $x_1,\ldots,x_n$. We say that the algorithm \emph{errs} if it returns ``accept'' and yet there is some index $1\leq j\leq m$ for which $\sum_{i=1}^nA_{j,i}x_{i}< b_j/\beta(n)$ holds. We require that the probability that the algorithm errs is at most $2/3$.

For the case where a linear program has a very large number of constraints, or the constraints are not given explicitly, one can use an approximate separation oracle, combined with the Ellipsoids algorithm, in order to compute an approximate LP-solution in time polynomial in the number of variables of the LP, provided that there is an Ellipsoid containing the feasible region, whose volume is not too large. For the case where a linear program has a large number of variables, but a relatively small number of constraints, we can use an approximate separation oracle for its dual LP in order to solve the original LP approximately. 
We start by providing an approximate separation oracle for (LP-D). We then show that this separation oracle can be used in order to obtain an approximate solution to (LP-P) in time $\poly(N)$.

\begin{lemma}\label{lemma: separation oracle}
Assume that there is an efficient $\alpha(n)$-approximation algorithm for the \DkS problem, where $\alpha$ is an increasing function, and $n$ is the number of vertices in the input graph. Then	there is a randomized $\beta(N)$-approximate separation oracle for (LP-D), where $N$ is the number of variables in the input graph $G$, and $\beta(N)=O(\alpha(N^2)\cdot \log^2N)$.
\end{lemma}

\begin{proof}
Recall that we are given as input real values $z$ and $\set{y_v\mid v\in V(G)}$. Clearly, we can efficiently check whether $z\geq 0$, and whether $y_v\geq 0$ for all $v\in V(G)$. If this is not the case, we can return the corresponding violated constraint.

We say that a set $S\in \rset$ of vertices is \emph{bad} if $z+\sum_{v\in S}y_v<m(S)/\beta$ holds, where $\beta=c\cdot \alpha(N^2)\cdot \log^2N$, and $c$ is a large enough constant whose value we set later.
Our goal is to design an efficient algorithm that either returns a violated constraint of the LP (that is, a set $S\in \rset$ of vertices for which 
$z+\sum_{v\in S}y_v<m(S)$ holds); or it returns ``accept''. We require that, if there exists a bad set $S$ of vertices, then the probability that the algorithm returns ``accept'' is at most $2/3$.

It will be convenient for us to slightly modify the input values in $\set{y_v\mid v\in V(G)}$, as follows. We let $m$ be the smallest integral power of $2$ that is greater than $|E(G)|$. First, for every vertex $v\in V(G)$ with $y_v>m$, we let $y'_v=m$, and for every vertex $v\in V(G)$ with $y_v<1/4$, we set $y'_v=0$. For each remaining vertex $v$, we let $y'_v$ be the smallest integral power of $2$ that is greater than $4y_v$. Note that for every vertex $v$ with $y'_v\neq 0$, $1\leq y'_v\leq 4m$ holds, and $y'_v$ is an integral power of $2$. We also set $z'=2z$. We say that a set $S\in \rset$ of vertices is \emph{problematic} if 
 $z'+\sum_{v\in S}y'_v<8m(S)/\beta$ holds. We
need the following two observations regarding the new values $\set{y'_v\mid v\in V(G)}$.

\begin{observation}\label{obs: relating old and new LP values}
	If $S\in \rset$ is a bad set of vertices, then it is a problematic set of vertices. 
\end{observation}
\begin{proof}
	Recall that, if $S$ is a bad set of vertices, then $z+\sum_{v\in S}y_v<m(S)/\beta$ must hold. Since, for every vertex $v\in V(G)$, $y'_v\leq 8y_v$ holds, and $z'=2z$, we get that:
	
\[z'+\sum_{v\in S}y'_v\leq 2z+8\sum_{v\in S}y_v\leq 8\left(z+\sum_{v\in S}y_v  \right )<8m(S)/\beta.\]
	
Therefore, set $S$ is problematic.	
\end{proof}

\begin{observation}\label{obs: relating old to new 2}
	Assume that there exists a set $S\in \rset$ of vertices, for which $z'+\sum_{v\in S}y'_v<m(S)$ holds. Let $S'\subseteq S$ be the set of vertices obtained from $S$ by deleting every vertex $v\in S$ that has no neighbors in $S$. Then $z+\sum_{v\in S'}y_v<m(S')$ holds. 
\end{observation}

\begin{proof}
Since, for every vertex $v\in S\setminus S'$, no neighbor of $v$ lies in $S$, we get that $m(S')=|E_G(S')|=|E_G(S)|=m(S)$. We partition the vertices of $S'$ into two subsets: set $X$ containing all vertices $v\in S'$ with $y_v<1/4$, and set $Y$ containing all remaining vertices. Clearly, $\sum_{v\in X}y_v<\frac{|X|}{4}\leq \frac{m(S')} 2$ (since $m(S')\geq |S'|/2\geq |X|/2$, as graph $G[S']$ contains no isolated vertices).

Assume for contradiction that $z+\sum_{v\in S'}y_v\geq m(S')$. Then:

\[ z+\sum_{v\in Y}y_v\geq m(S')-\sum_{v\in X}y_v\geq m(S')/2.\]

We now consider two cases. The first case is when there is some vertex $v\in Y$ with $y_v\geq m$. In this case, $y'_v\geq m$ holds, and $z'+\sum_{v\in S}y'_v\geq m>m(S)$ holds, a contradiction.

Otherwise,  for every vertex $v\in Y$, $y'_v\geq 4y_v$ holds. Since $z'=2z$ also holds, we get that:

\[z'+\sum_{v\in S}y'_v\geq z'+\sum_{v\in Y}y'_v\geq 2z+4\sum_{v\in Y}y_v\geq m(S')=m(S),\]

a contradiction. 
\end{proof}

From now on we focus on values $z',\set{y'_v\mid v\in V(G)}$. It is now enough to design an efficient randomized algorithm, that either computes a set $S\in \rset$ of vertices, for which $z'+\sum_{v\in S}y'_v<m(S)$ holds, or returns ``accept''. It is enough to ensure that, if there is a problematic set $S\in \rset$ of vertices, then the algorithm returns ``accept'' with probability at most $2/3$. Indeed, if there is a bad set $S\in \rset$ of vertices, then, from  \Cref{obs: relating old and new LP values}, there is a problematic set of vertices, and the algorithm will return ``accept'' with probability at most $2/3$. On the other hand, if the algorithm computes a set $S\in \rset$ of vertices, for which $z'+\sum_{v\in S}y'_v<m(S)$ holds, then we can return the set $S'\subseteq S$ of vertices from the statement of \Cref{obs: relating old to new 2}, that defines a violated constraint with respect to the original LP-values.

Our algorithm computes a random partition $(A,B)$ of the vertices of $G$, where every vertex $v\in V(G)$ is independently added to $A$ or to $B$ with probability $1/2$ each. Let $q=\log(8m)$. For all $1\leq i\leq q$, we define a set $A_i\subseteq A$ of vertices: $A_i=\set{v\in A\mid y'_v=2^{i-1}}$, and we let $A_0=\set{v\in A\mid y'_v=0}$. Clearly, $(A_0,\ldots,A_q)$ is a partition of the set $A$ of vertices.

We compute  a partition $(B_0,\ldots,B_q)$ of the vertices of $B$ similarly. For all $0\leq i,j\leq  q$, we denote by $E_{i,j}$ the set of all edges $e=(u,v)$ with $u\in A_i$ and $v\in B_j$, and we define a bipartite graph $G_{i,j}$, whose vertex set is $A_i\cup B_j$, and edge set is $E_{i,j}$.

Recall that we have assumed that there is an efficient $\alpha(n)$-approximation algorithm for the \DkS problem, where $n$ is the number of vertices in the input graph. From \Cref{lem: DkS and Dk1k2S}, there exists an efficient $O(\alpha(\hat n^2))$-approximation algorithm for the \BDkS problem, where $\hat n$ is the number of vertices in the input graph.
We denote this algorithm by $\aset'$.

For every pair $0\leq i,j\leq q$ of integers, and every pair $k_1,k_2\geq 0$ of integers with $k_1+k_2\leq k$, we apply Algorithm $\aset'$ for the \BDkS problem to graph $G_{i,j}$, with parameters $k_1$ and $k_2$. Let $S_{i,j}^{k_1,k_2}$ be the output of this algorithm, and let $m_{i,j}^{k_1,k_2}$ be the number of edges in the subgraph of $G_{i,j}$ that is induced by the set $S_{i,j}^{k_1,k_2}$ of vertices. We say that the application of algorithm $\aset'$ is \emph{successful} if $z'+\sum_{v\in S_{i,j}^{k_1,k_2}}y'_v<m_{i,j}^{k_1,k_2}$, and otherwise it is \emph{unsuccessful}. If, for any quadruple $(i,j,k_1,k_2)$ of indices, the application of algorithm $\aset'$ was successful, then we return the resulting set $S=S_{i,j}^{k_1,k_2}$ of vertices. Clearly, $|S|\leq k_1+k_2\leq k$, so $S\in \rset$ holds. Moreover, we are guaranteed that $z'+\sum_{v\in S}y'_v<m_{i,j}^{k_1,k_2}<m(S)$, as required. If every application of algorithm $\aset'$ is unsuccessful, then we return ``accept''. The following observation will finish the proof of \Cref{lemma: separation oracle}.

\begin{observation}\label{obs: if problematic then fail}
	Suppose there is a problematic set $S\in \rset$ of vertices. Then the probability that the algorithm returns ``accept'' is at most $2/3$.
\end{observation}

\begin{proof}
	Let $S\in \rset$ be a problematic set of vertices, so $z'+\sum_{v\in S}y'_v<8m(S)/\beta$ holds. Let $E'=E_G[S]$, so $|E'|=m(S)$. Denote $A_S=A\cap S,B_S=B\cap S$, and let $E''\subseteq E'$ be the set of edges $e$, such that exactly one endpoint of $e$ lies in $A$. Clearly, for every edge $e\in  E'$, $\prob{e\in E''}=1/2$. Therefore, $\expect{|E''|}=|E'|/2$. Let $\event'$ be the bad event that $|E''|<|E'|/8$, and let $p=\prob{\event'}$. Clearly:
	
	\[ \expect{|E''|}\leq p\cdot \frac{|E'|}8+(1-p)\cdot |E'|=|E'|\left(1-\frac{7p}{8}\right ). \]
	
	Since $ \expect{|E''|}=|E'|/2$, we get that $p\leq 2/3$. Next, we show that, if Event $\event'$ does not happen, then the algorithm does not return ``accept''. 
	
	From now on we assume that Event $\event'$ did not happen, so $|E''|\geq |E'|/8$. 
	Therefore: 
	
	$$z'+\sum_{v\in S}y'_v <\frac{8m(S)}{\beta}\leq \frac{64|E''|}{\beta}$$
	
	 holds.

	Clearly, there must be a pair $0\leq i,j\leq q$ of indices, such that $|E''\cap E_{i,j}|\geq \frac{|E''|}{4q^2}\geq\frac{|E''|}{128\log^2 m}$. We now fix this pair $i,j$ of indices, and denote $A'_i=A_i\cap S$ and $B'_j=B_j\cap S$. We also denote $k_1=|A'_i|$ and let $k_2=|B'_j|$.
	Clearly, $k_1+k_2\leq k$ holds. 
	Denote $M_{i,j}=|E''\cap E_{i,j}|$.
	From our choice of indices $i,j$, we get that:

	\[ z'+\sum_{v\in A'_i\cup B'_j}y'_v\leq z'+\sum_{v\in S}y'_v\leq  \frac{64|E''|}{\beta}\leq \frac{M_{i,j}}{\beta}\cdot (2^{13}\log^2m).\]
	
	 Notice that the set $S'=A'_i \cup B'_j$ of vertices provides a solution to the instance of the \BDkS problem on graph $G_{i,j}$ with parameters $k_1,k_2$, whose value is at least $M_{i,j}$. Let $S''=S_{i,j}^{k_1,k_2}$ be the set of vertices obtained by applying Algorithm $\aset'$ to graph $G_{i,j}$ with parameters $k_1,k_2$. Since $|V(G_{i,j})|\leq N$, and since $\aset'$ is an $O(\alpha(N^2))$-approximation algorithm for \BDkS, we are guaranteed that $|E_G(S'')|\geq \Omega\left(\frac{M_{i,j}}{\alpha(N^2)}\right )$. Recall that $|A\cap S''|\leq k_1$; $A\cap S''\subseteq A_i$, and all vertices  $v\in A_i$ have an identical value $y'_v$. Therefore, $\sum_{v\in A\cap S''}y'_v\leq \sum_{v\in A'_i}y'_v$. Using a similar reasoning, 
	 $\sum_{v\in B\cap S''}y'_v\leq \sum_{v\in B'_j}y'_v$. Overall, we then get that:

	 \[
	 \begin{split}
	 z'+\sum_{v\in S''}y'_v& \leq z'+\sum_{v\in S'}y'_v\\
	 &\leq   \frac{M_{i,j}}{\beta}\cdot (2^{13}\cdot \log^2m)\\
	 &\leq O\left (\frac{|E_G(S'')|\cdot \alpha(N^2)\cdot 2^{13}\cdot \log^2m}{\beta}\right )
	 \end{split}\]

Recall that $\beta=c\cdot \alpha(N^2)\cdot \log^2N$. By letting the value of the constant $c$ be high enough, we can ensure that $z'+\sum_{v\in S''}y'_v<|E_G(S'')|$, and so the application of algorithm $\aset'$ to graph $G_{i,j}$ with parameters $k_1$ and $k_2$ is guaranteed to be successful. Therefore, if Event $\event'$ does not happen, and we set $c$ to be a large enough constant, then our algorithm does not return ''accept''. Since $\prob{\event'}\leq 2/3$, the observation follows.
\end{proof}
\end{proof}

\paragraph{Approximately Solving (LP-P).}

We use standard methods for solving (LP-P) using approximate separation oracle for (LP-D).

For a collection $\rset'\subseteq \rset$ of vertex subsets, we define a linear program ($\mbox{LP}(\rset')$), which is obtained from (LP-D) by only including the constraints associated with the subsets in $\rset'$:

\begin{eqnarray*}
	\mbox{(LP($\rset'$))}&&\\
	\min& \frac{N}{k}\cdot z+\sum_{v\in V(G)} y_v\\
	\mbox{s.t.}&&\\
	& z+\sum_{v\in S} y_v\geq m(S) &\forall S\in \rset'\\
	&z\ge 0\\
	&y_v\geq 0&\forall v\in V(G)
\end{eqnarray*}

We denote by $\opt(\rset')$ the value of the optimal solution to ($\mbox{LP}(\rset')$).
Since a solution to (LP-D) defines a solution to ($\mbox{LP}(\rset')$), we get that $\opt(\rset')\leq \opt_{\textnormal{LP-D}}$
We use the following lemma that allows us to compute a small collection $\rset'\subseteq \rset$ of vertex subsets, such that $\opt(\rset')$ is within a factor $\beta(N)=O(\alpha(N^2)\cdot \log^2N)$ of $\opt_{\textnormal{LP-D}}$.

\begin{claim}\label{lem: important inequalities}
	Assume that there is an efficient $\alpha(n)$-approximation algorithm for \DkS, where $\alpha$ is an increasing function of $n$, and $n$ is the number of vertices in the input graph.
	Then there is a randomized algorithm with running time $O(\poly(N))$, that computes a collection $\rset'\subseteq \rset$ of subsets of vertices with $|\rset'|\leq O(\poly(N))$, such that, with high probability, $\opt(\rset')\geq \Omega\left(\opt_{\textnormal{LP-D}}/\beta(N)\right )$, where $\beta(N)=O(\alpha(N^2)\cdot \log^2N)$.
\end{claim}

We prove \Cref{lem: important inequalities} below, after we provide an algorithm for approximately solving (LP-P) using it. 

Recall that for every set $S\in \rset$ of vertices, there is a variable $x_S$ in the primal LP, (LP-P). We now consider the dual LP to $\mbox{LP}(\rset')$, that is defined as follows:

\begin{eqnarray*}
	\mbox{(LP-P2)}&&\\
	\max &\sum_{S\in \rset'} m(S)\cdot x_S&\\
	\mbox{s.t.}&&\\
	&\sum_{\stackrel{S\in \rset':}{v\in S}} x_S\leq 1 &\forall v\in V(G)\\
	&\sum_{S\in \rset'}x_S \leq N/k\\
	&x_S\geq 0&\forall S\in \rset'
\end{eqnarray*}

Notice that this Linear Program can be obtained from (LP-P) by eliminating all variables $x_S$ for $S\in \rset\setminus\rset'$. Since $|\rset'|\leq \poly(N)$, this new linear program has at most $O(\poly(N))$ variables, and it has at most $\poly(N)$ constraints. Therefore, we can solve it in time $O(\poly(N))$ using standard algorithms for LP-solving. Let $\set{x'_S\mid S\in \rset'}$ be the resulting solution. From the Strong Duality Theorem, we get that:

\[ \sum_{S\in \rset'} m(S)\cdot x'_S=\opt(\rset').\]

From  \Cref{lem: important inequalities}, with high probability:

\[\opt(\rset') \geq \Omega\left(\frac{\opt_{\textnormal{LP-D}}}{\beta(N)}\right )=\Omega\left(\frac{\opt_{\textnormal{LP-P}}}{\beta(N)}\right ). \]

Altogether, we get that with high probability, $\sum_{S\in \rset'} m(S)\cdot x'_S\geq \Omega\left(\frac{\opt_{\textnormal{LP-P}}}{\beta(N)}\right )$.
We can extend the solution $\set{x'_S\mid S\in \rset'}$ to (LP-P2) to obtain a feasible solution $\set{x'_S\mid S\in \rset}$ to (LP-P) by setting the value $x'_S$ for all sets $S\in \rset\setminus \rset'$ to $0$. It is immediate to verify that the resulting solution to (LP-P) is feasible, and its value remains unchanged. Therefore, we have obtained a randomized algorithm, with running time bounded by $O(\poly(N))$, that with high probability computes a solution to (LP-P), whose value is at least $\Omega\left(\frac{\opt_{\textnormal{LP-P}}}{\beta(N)}\right )$; here, $\beta(N)=O(\alpha(N^2)\cdot \log^2N)$.

We are now ready to complete the reduction from $\DkC$ to $\DkS$ from \Cref{thm: alg_DkS gives alg_DkC}. Let $\alpha: \mathbb{Z^+}\to \mathbb{Z^+}$ be an increasing function, such that $\alpha(n)\leq o(n)$, and assume that 
there is an efficient $\alpha(n)$-approximation algorithm for the \DkS problem, where $n$ is the number of vertices in the input graph.

Consider now the input instance $\pcl(G,k)$ of the \DkC problem with $N=|V(G)|$. We can assume w.l.o.g. that $N$ is greater than a large enough constant, as otherwise we can solve the problem exactly via exhaustive search. Since $\alpha(n)\leq o(n)$, we can also assume that $\beta(N)\leq N^3$. 
We use the randomized algorithm described above, that, in time $O(\poly(N))$, with high probability computes a $\beta(N)$-approximate solution to (LP-P). Recall that the number of variables of (LP-P) with non-zero LP-value is bounded by $O(\poly(N))$. Next, we apply the algorithm from \Cref{claim: LP-rounding} in order to round the resulting LP solution. The algorithm 
with high probability returns an integral solution $(S_1,\ldots,S_{N/k})$ to instance $\pcl(G,k)$, such that $\sum_{i=1}^{N/k}|E_G(S_i)|\geq \Omega\left( \frac{\optcl(G,k)}{\beta(N)\cdot \log^3N}\right )$. Since $\beta(N)=O(\alpha(N^2)\cdot \log^2N)$, with high probability we obtain an $O(\alpha(N^2)\cdot \poly\log N)$-approximate solution to the input instance of \DkC. In order to complete the reduction  from $\DkC$ to $\DkS$ from \Cref{thm: alg_DkS gives alg_DkC}, it is now enough to prove \Cref{lem: important inequalities}, which we do next. The proof uses standard techniques and is only included for completeness.

\begin{proofof}{\Cref{lem: important inequalities}}
Let $m=|E(G)|$. Notice that we can assume w.l.o.g. that in an optimal solution to (LP-D), for every vertex $v\in V(G)$, $y_v\leq m$ holds, and $z\leq m$. Indeed, if this is not the case, then we can modify the solution by setting $y_v=\min\set{y_v,m}$ for every vertex $v\in V(G)$, and $z=\min\set{z,m}$. It is easy to verify that this remains a feasible solution, and its value does not grow. For convenience, for every subset $\rset'\subseteq \rset$ of vertex subsets, we define the following LP:

	\begin{eqnarray*}
		\mbox{($P'(\rset'))$}&&\\
		\min& \frac{N}{k}\cdot z+\sum_{v\in V(G)} y_v\\
		\mbox{s.t.}&&\\
		& z+\sum_{v\in S} y_v\geq m(S) &\forall S\in \rset'\\
		&0\leq z\leq m\\
		&0\leq y_v\leq m&\forall v\in V(G)
	\end{eqnarray*}

We denote by $\opt'(\rset')$ the value of the optimal solution of the above LP.
From the above discussion, it is enough to provide 	a randomized algorithm with running time $O(\poly(N))$, that computes a collection $\rset'\subseteq \rset$ of subsets of vertices, such that, with high probability, $\opt'(\rset')\geq \Omega\left(\opt'(\rset)/\beta(N)\right )$, where $\beta(N)=O(\alpha(N^2)\cdot \log^2N)$. This is since $\opt(\rset')=\opt'(\rset')$, and $\opt'(\rset)=\opt(\rset)$ holds.

Clearly, $\opt'(\rset)\leq 2Nm$ must hold. Let $M$ be the smallest integral power of $2$ that is greater than $2Nm$. For all $0<C\leq \log M$ and $\rset'\subseteq \rset$, we consider a feasibility LP, that is obtained from $P'(\rset')$, by adding the constraint that 	$\frac{N}{k}\cdot z+\sum_{v\in V(G)} y_v\leq 2^C$:

\begin{eqnarray*}
	\mbox{($F(\rset',C))$}&&\\
	& \frac{N}{k}\cdot z+\sum_{v\in V(G)} y_v\leq 2^C\\
	& z+\sum_{v\in S} y_v\geq m(S) &\forall S\in \rset'\\
	&0\leq z\leq m\\
	&0\leq y_v\leq m&\forall v\in V(G)
\end{eqnarray*}

The key to the proof of \Cref{lem: important inequalities} is the following observation.

\begin{observation}\label{obs: important inequalities for fixed C}
	Assume that there is an efficient $\alpha(n)$-approximation algorithm for \DkS, where $\alpha$ is an increasing function of $n$, and $n$ is the number of vertices in the input graph.
Then	there is a randomized algorithm with running time $O(\poly(N))$, that, given a value $0\leq C\leq \log M$, either:
	
	\begin{itemize}
		\item computes a collection $\rset'(C)\subseteq \rset$ of at most $O(\poly(N))$ subsets of vertices, such that the linear program ($F(\rset',C)$) is infeasible; or
		
		\item computes values $0\leq z'\leq m$ and $0\leq y'_v\leq m$ for all $v\in V(G)$, such that $\frac{N}{k}\cdot z+\sum_{v\in V(G)} y_v\leq 2^C$, and, with high probability, for every vertex set $S\in \rset$, $z'+\sum_{v\in S} y'_v\geq m(S)/\beta(N)$.
	\end{itemize}
\end{observation}

We provide the proof of \Cref{obs: important inequalities for fixed C} below, after we complete the proof of \Cref{lem: important inequalities} using it. We apply the algorithm from \Cref{obs: important inequalities for fixed C} to every value $0\leq C\leq \log M$. We say that the application of the algorithm for value $C$ is \emph{successful} if the algorithm returns values $0\leq z'\leq m$ and $0\leq y'_v\leq m$ for all $v\in V(G)$; otherwise we say that it is unsuccessful. We say that the algorithm \emph{errs} if it is successful, and yet there is a set $S\in \rset$ of vertices for which $z'+\sum_{v\in S} y'_v< m(S)/\beta(N)$.

Let $C^*$ be the smallest value of $C$, such that the algorithm from \Cref{obs: important inequalities for fixed C}, when applied to $C$, was successful. Let $z'$ and $\set{y'_v\mid v\in V(G)}$ be the values of the LP-variables returned by the algorithm. We let $\hat \event$ be the bad event that the algorithm from \Cref{obs: important inequalities for fixed C}, when applied to value $C^*$, errs. The probability of $\hat \event$ is at most $1/\poly(N)$. Consider the following solution to $(P'(\rset))$: we set $z^*=\min\set{z'\cdot \beta(N),m}$, and for all $v\in V(G)$, we set $y^*_v=\min\set{y'_v\cdot \beta(N),m}$. It is immediate to verify that, if Event $\hat \event$ did not happen, then we obtain a feasible solution to 
$(P'(\rset))$, whose value is at most $2^{C^*}\cdot \beta(N)$. Notice that, from the choice of the value $C^*$, LP ($F(\rset'(C^*-1),C^*-1)$) does not have a feasible solution. Since the constraints in ($F(\rset'(C^*-1),C^*-1$) are a subset of the constraints in ($F(\rset,C^*-1)$), it follows that ($F(\rset,C^*-1)$) does not have a feasible solution, and so $\opt'(\rset)\geq 2^{C^*-1}$. Therefore, if Event $\hat \event$ did not happen, we obtain a feasible solution $(z^*,\set{y^*_v}_{v\in V(G)})$ to $(P'(\rset))$, whose 
value is at most $2\beta(N)\cdot \opt'(\rset)$.

The final collection of subsets of vertices that our algorithm returns is $\rset'=\bigcup_{C=0}^{C^*-1}\rset'(C)$. Clearly, $|\rset'|\leq \poly(N)$. It is also immediate to verify that $\opt(\rset')\geq 2^{C^*-1}$. Indeed, for all values $0\leq C<C^*$, linear program ($F(\rset'(C),C)$) is infeasible, and, since $\rset'(C)\subseteq \rset'$, linear program ($F(\rset',C)$) is also infeasible. Since every constraint of $P'(\rset')$ is also a constraint of $LP(\rset)$, we get that, if Event $\hat \event$ did not happen, then   $(z^*,\set{y^*_v}_{v\in V(G)})$ is a feasible solution to $P'(\rset')$, whose 
value is at most $\beta(N)\cdot 2^{C^*}$. 
To summarize, if Event $\hat \event$ did not happen, then:

$$2^{C^*-1}\leq \opt'(\rset')\leq 2^{C^*}\cdot \beta(N)$$

and 

$$2^{C^*-1}\leq \opt'(\rset)\leq 2^{C^*}\cdot \beta(N)$$

hold.

Therefore, if Event $\hat \event$ did not happen, then $\opt'(\rset')\geq 2^{C^*-1}\geq \opt'(\rset)/(2\beta(N))$.
It now remains to prove \Cref{obs: important inequalities for fixed C}. The proof is standard; we only provide its sketch below.

\begin{proofof}{\Cref{obs: important inequalities for fixed C}}
We fix a value $0\leq C\leq \log M$, and consider the corresponding Linear Program ($F(\rset,C))$. The idea of the proof is simple: we employ the Ellipsoids algorithm, together with the separation oracle from \Cref{lemma: separation oracle} (after we reduce its error probability by repeating the algorithm a number of times). We then let $\rset'(C)$ be the collection of all vertex subsets $S\in \rset$, such that the separation oracle returns the constraint associated with $S$ over the course of the algorithm.

We now provide more details. Recall first the Ellipsoids algorithm for solving a feasibility Linear Program ($F$) on $N'$ variables. The algorithm proceeds in iterations. The input to the $i$th iteration is an $N'$-dimensional  Ellipsoid $E_i$, that contains the feasible region of ($F$). Let $x_i$ denote the center point of the ellipsoid. If the algorithm is given a constraint $A^j$ of the Linear Program $(F)$ that is violated by point $x_i$, then it produces a new ellipsoid $E_{i+1}$, that contains the feasible region of ($F$), whose volume is at most $(1-1/\poly(N'))$ times the volume of $E_i$. The running time of a single iteration is $O(\poly(N'))$.

Typically, we assume that there is an initial ellipsoid $E_1$, whose volume is at most $2^{\poly(N')}$, that contains the feasible region of $(F)$, which needs to be supplied to the Ellipsoids algorithm.  We can also typically assume that, if $(F)$ has a feasible solution, then the volume of the feasible region of $(F)$ is at least $L=2^{-\poly(N')}$ (if this is not the case, the feasible region can be slightly inflated artificially by adding a small amount of slack to the constraints; in our case, since we are only solving the LP approximately, this is immaterial). If the above two conditions hold, the algorithm can proceed for at most $\poly(N')$ iterations, before the volume of the current ellipsoid becomes smaller than $L$, and the algorithm then correctly declares that $(F)$ does not have a feasible solution. In every iteration, the constraint violated by the  center $x_i$ of the current ellipsoid $E_i$ is supplied by a separation oracle. If the separation oracle declares that $x_i$ is (approximately) feasible solution, then the algorithm halts.

We now turn to consider the linear program ($F(\rset,C))$. Since the LP constraints require that the values of all LP-variables are between $0$ and $m$, it is easy to verify that the feasible region of the LP is contained in the $(N+1)$-dimensional sphere $E_1$, whose radius is bounded by $\poly(m)$, and volume is at most $2^{\poly(m)}\leq 2^{\poly(N)}$. We initially set $\rset'(C)=\emptyset$. We apply the Ellipsoids algorithm to this LP, with the initial ellipsoid $E_1$.

We now consider the $i$th iteration of the algorithm, whose input is an ellipsoid $E_i$, together with its center point $(z^i,\set{y^i_v}_{v\in V(G)})$. We manually check the constraints $\frac{N}{k}\cdot z^i+\sum_{v\in V(G)} y^i_v\leq 2^C$; $0\leq z^i\leq m$; and $0\leq y^i_v\leq m$ for all $v\in V(G)$. If any of these constraints does not hold, then we return it as a violated constraint. Assume now that all these constraints hold. We apply the algorithm from \Cref{lemma: separation oracle} to the current values $(z^i,\set{y^i_v}_{v\in V(G)})$; we do so $N$ times. If, in each of these iterations, the algorithm returns ``accept'', then we terminate our algorithm, and return the current solution  $(z^i,\set{y^i_v}_{v\in V(G)})$. 
Observe that we are guaranteed that 
$0\leq z^i\leq m$; $0\leq y^i_v\leq m$ for all $v\in V(G)$; and $\frac{N}{k}\cdot z^i+\sum_{v\in V(G)} y^i_v\leq 2^C$. Moreover, unless the algorithm from \Cref{lemma: separation oracle}  erred in each of its $N$ applications, we are guaranteed that, for every vertex set $S\in \rset$, $z^i+\sum_{v\in S} y^i_v\geq m(S)/\beta(N)$ holds. The probability that the algorithm from \Cref{lemma: separation oracle}   errs is at most $2/3$, so with high probability, we are guaranteed that for all $S\in \rset$,  $z^i+\sum_{v\in S} y^i_v\geq m(S)/\beta(N)$.

Assume now that in some application of the algorithm from \Cref{lemma: separation oracle} to the current values $(z^i,\set{y^i_v}_{v\in V(G)})$ we obtain a violated constraint of (LP-D). That is, we obtain a set $S\in \rset$ of vertices, for which $z^i+\sum_{v\in S} y^i_v< m(S)$ holds. In this case, we add $S$ to set $\rset'(S)$, and we use this constraint as a violated constraint for the Ellipsoids algorithm. 

If the above algorithm never terminates with an approximately feasible solution  $(z^i,\set{y^i_v}_{v\in V(G)})$, then we are guaranteed that after at most $\poly(N)$ iterations, the algorithm correctly certifies that ($F(\rset,C)$) does not have a feasible solution.
We then return the current collection $\rset(C)$ of vertex subsets.
For convenience, we denote by $A^1,A^1,\ldots,A^r$ the sequence of violated constraints that were fed to the Ellipsoids algorithm. Each of the constraints $A^j$ either corresponds to a set $S\in \rset(C)$, or it is  one of the constraints $0\leq z^i\leq m$; $0\leq y^i_v\leq m$ for all $v\in V(G)$; and $\frac{N}{k}\cdot z^i+\sum_{v\in V(G)} y^i_v\leq 2^C$. In other words, each such constraint $A^j$ is also a constraint of the LP  ($F(\rset(C),C)$).

It now remains to prove that in the latter case, ($F(\rset(C),C)$) does not have a feasible solution. In order to do so, consider applying the Ellipsoids algorithm to this linear program. We start with the same initial ellipsoid $E_1$ as before. Since the Ellipsoid algorithm is deterministic, its behavior is entirely determined by the initial ellipsoid $E_1$ and the sequence of 
the violated constraints that it receives. We will use exactly the same sequence $A^1,A^1,\ldots,A^r$  of violated constraints in this execution of Ellipsoids algorithm. This ensures that for all $i$, the ellipsoid $E_i$ that is used as the input to the $i$th iteration is identical to the ellipsoid that was used as input to iteration $i$ when solving ($F(\rset,C)$), which in turn ensures that constraint $A^i$ is a violating constraint for the center of ellipsoid $E_i$. Therefore, this execution of Ellipsoids algorithm is identical to the execution of the same algorithm when applied to LP ($F(\rset,C)$), and it will end up with a final ellipsoid $E_r$, whose volume is small enough to correctly establish that ($F(\rset(C),C)$) does not have a feasible solution. 
\end{proofof}
\end{proofof}

\section{Reductions from \DkS to \DkC and \WGP}
\label{sec: DkS to DkC and WGP}

In this section we prove the following theorem.

\begin{theorem}
\label{thm: alg_DkC gives alg_DkS}
Let $\alpha: \mathbb{Z^+}\to \mathbb{Z^+}$ be an increasing function with $\alpha(n)\leq o(n)$.  Then the following hold:

\begin{itemize}
\item  If there exists an efficient $\alpha(n)$-approximation algorithm $\aset$ for the \DkC problem, where $n$ is the number of vertices in the input graph, then there exists a randomized 
algorithm for the \DkS problem, whose running time is  $N^{O(\log N)}$, that with high probability computes an 
$O(\alpha(N^{O(\log N)})\cdot\log N)$-approximate solution to the input instance of the problem; here $N$ is the number of vertices in the input instance of \DkS.

\item 
If there exists an efficient $\alpha(n)$-approximation algorithm for the \WGP problem, where $n$ is the number of vertices in the input graph, then exists a randomized 
algorithm for the \DkS problem, whose running time is  $N^{O(\log N)}$, that with high probability computes an 
$O((\alpha(N^{O(\log N)}))^3\cdot\log^2 N)$-approximate solution to the input instance of the problem; here $N$ is the number of vertices in the input instance of \DkS.
\end{itemize}
\end{theorem}

We obtain the following immediate corollary of \Cref{thm: alg_DkC gives alg_DkS}.

\begin{corollary}\label{cor: hardness of WGP and DkC}
	Assume that \Cref{conj: main1} holds and that  $\NP\not \subseteq \BPTIME(n^{O(\log n)})$. Then for some constant $0<\eps'\leq 1/2$, there is no efficient $2^{(\log n)^{\eps'}}$-approximation algorithm for \WGP, and there is no efficient $2^{(\log n)^{\eps'}}$-approximation algorithm for \DkC.
\end{corollary}

\begin{proof}
	We prove the corollary for \WGP; the proof for \DkC is similar.
	Assume that \Cref{conj: main1} holds and that  $\NP\not \subseteq \DTIME(n^{O(\log n)})$. Then, from \Cref{thm: main_DkS2}, for some constant $0<\eps<1$, there is no randomized factor-$2^{(\log n)^{\eps}}$-approximation algorithm for \DkS with running time $n^{O(\log n)}$, where $n$ is the number of vertices in the input graph.
	
	We let $\eps'=\eps/c$, where $c$ is a sufficiently large constant. We now prove that there is no efficient $2^{(\log n)^{\eps'}}$-approximation algorithm for \WGP. Indeed, assume for contradiction that there is an efficient $2^{(\log n)^{\eps'}}$-approximation algorithm $\aset$ for \WGP.  From \Cref{thm: alg_DkC gives alg_DkS}, there is a randomized 
	algorithm for the \DkS problem, that, given an instance $\pdks(G,k)$ of the problem with $|V(G)|=N$, in time  $N^{O(\log N)}$, computes a
	$c'(\alpha(N^{O(\log N)}))^3\cdot\log^2 N$-approximate solution, where
	$\alpha(x)=2^{(\log x)^{\eps'}}$ and $c'$ is a constant independent of $N$. Note that:
	
	\[\alpha(N^{O(\log N)})=2^{(\log N)^{O(\eps')}}=2^{  (\log N)^{O(\eps/c)}}.\]
	
	Since we can let $c$ be a sufficiently large constant, we can ensure that $c'(\alpha(N^{O(\log N)}))^3\cdot\log^2 N<2^{(\log n)^{\eps}}$.
	
	Therefore, we obtain a randomized factor-$2^{(\log n)^{\eps}}$-approximation algorithm for \DkS, with running time 
$n^{O(\log n)}$, a contradiction.
	\end{proof}

The remainder of this section is dedicated to the proof of \Cref{thm: alg_DkC gives alg_DkS}. In order to obtain both reductions, we start with an instance $\pdks(G,k)$ of the \DkS problem, and construct another auxiliary graph $H$. This graph is then used in order to define the corresponding instances of \DkC and \WGP, respectively. We start by defining graph $H$ and analyzing its properties in \Cref{subsec: graph H}. We then complete the reduction from \DkS to \DkC in \Cref{subsec: finish DkS to DkC}, and the reduction from \DkS to \WGP in
\Cref{subsec: finish DkS to WGP}.
Throughout this section, for an integer $A$, we denote $[A]=\set{0,1,\ldots,A-1}$. We also assume that the  parameter $k$ in the input instance of the \DkS problem is greater than a large enough constant, since otherwise the problem can be solved in time $\poly(N)$ via exhaustive search.

\subsection{Auxiliary Graph $H$}
\label{subsec: graph H}

Let $\pdks(G,k)$ be an instance of the \DkS problem. Denote $V(G)=\set{v_0,v_1,\ldots,v_{N-1}}$.  We now provide a randomized algorithm to construct an auxiliary graph $H$ corresponding to this instance. The construction is somewhat similar to and inspired by the construction used in Section 2 of \cite{khanna2000hardness}.

Let $q=\ceil{\log N}$. We start by computing a prime number $M$, such that $N^{5q}\le M\le 2\cdot N^{5q}$. From the Bertrand-Chebyshev theorem \cite{bertrand1845memory,vcebyvsev1850memoire}, such a prime number must exist, and it can be computed in time $N^{O(\log N)}$ by checking every integer between  $N^{5q}$ and $2\cdot N^{5q}$.
We then construct a random mapping $f:[N]\to [M]$ as follows. For every integer $1\le i\le N$, we let $f(i)$ be an integer chosen independently and uniformly at random (with replacement) from $[M]$. 

We are now ready to define the graph $H$. The set of vertices of $H$ is $V(H)=\set{u_0,\ldots,u_{M-1}}$. For every edge $e=(v_i,v_j)\in E(G)$, we construct a collection $J(e)$ of $M$ edges in $H$: $J(e)=\set{\big(u_{f(i)+t},u_{f(j)+t}\big)\mid 0\le t\le M-1 }$, where the addition in the subscript is modulo $M$ (we use this convention throughout the remainder of this section). We say that edge $e$ is the \emph{origin} of every edge in set $J(e)$. We then set $E(H)=\bigcup_{e\in E(G)}J(e)$. We note that we do not allow parallel edges in $H$, so it is possible for an edge in $H$ to have several origin edges in $G$. This completes the definition of the graph $H$. We now analyze its properties.

\paragraph{Good Event $\event^g$.}
We say that a good event $\event^g$ happens if there is a collection $\set{H_1,\ldots,H_r}$ of $r=\floor{\frac M {k\log k}}$ disjoint subgraphs of $H$, such that the following hold:

\begin{itemize}
	\item for all $1\leq j\leq r$, $|V(H_j)|= k$; 
	\item for all $1\leq j\leq r$, $|E(H_j)|\leq \optdks(G,k)$; and
	\item $\sum_{1\le j\le r}|E(H_j)|\ge 0.1\cdot \floor{\frac M{k\log k}}\cdot\optdks(G,k)$.
\end{itemize}

 We start by showing that good event $\event^g$ happens with a sufficiently high probability.

\begin{claim}\label{claim: disjoint copies}
$	\prob{\event^g}\geq 0.8$.
\end{claim}

\begin{proof}
Let $S$ be the optimal solution to instance $\pdks(G,k)$ of the \DkS problem, and let $T$ be a subset of vertices of $H$, defined as $T=\set{u_{f(i)}\mid v_i\in S}$.
Let $I$ be a collection of $r=\floor{M/(k\log k)}$ integers from $[M]$, obtained by selecting each integer independently uniformly at random  (with replacement) from $[M]$. For every index $j\in I$, we define a set $T'_j$ of vertices of $H$ as follows: $T'_{j}=\set{u_{f(i)+j}\mid v_i\in S}$.
Finally, for every index $j\in I$, we define another set $T_j$ of vertices of $H$, by starting with the set $T'_j$ of vertices, and then removing from it every vertex that lies in set $\bigcup_{i\in I\setminus\set{j}}T'_i$.
Clearly, all resulting vertex sets in $\set{T_j}_{j\in I}$ are mutually disjoint, and each such set contains at most $k$ vertices.

For every index $j\in I$, we denote by $\hat E(T_j)$ the set of edges  $e\in E_H(T_j)$, such that an origin of $e$ in $G$ lies in $E_G(S)$. Equivalently: $\hat E(T_j)=\set{(u_{f(i)+j},u_{f(i')+j})\mid (v_i,v_{i'})\in E_G(S)}$. Clearly, $|\hat E(T_j)|\le \optdks(G,k)$ for all $j\in I$, and so $\sum_{j\in I}|\hat E(T_{j})|\leq r\cdot \optdks(G,k)=\floor{\frac{M}{k\log k}}\cdot \optdks(G,k)$.

We prove the following observation.

\begin{observation}\label{obs: expected cost}
\[\expect{\sum_{j\in I}|\hat E(T_{j})|}\ge 0.9\cdot\floor{\frac{M}{k\log k}}\cdot \optdks(G,k).\]
\end{observation}

Assume first that the observation holds. 
For simplicity of notation, denote $B=	\floor{\frac{M}{k\log k}}\cdot \optdks(G,k)$. 
For all $j\in I$, we define a subgraph $H_j$ of $H$, whose vertex set is $T_j$, and edge set is $\hat E(T_j)$. Clearly, the graphs in $\set{H_j\mid j\in I}$ are disjoint, and, for all $j\in I$, $|V(H_j)|=|T|\leq k$ holds. Additionally, from the above discussion, for all $j\in I$, $|E(H_j)|\leq  \optdks(G,k)$. If, additionally, $\sum_{j\in I}|\hat E(T_{j})|\geq 0.1B$ holds, then Event $\event^g$ happens. As observed above, $\sum_{j\in I}|\hat E(T_{j})|\leq B$ must hold. Let $p$ denote the probability that $\sum_{j\in I}|\hat E(T_{j})|\geq 0.1B$. Then:

\[\expect{\sum_{j\in I}|\hat E(T_{j})|}\leq (1-p)\cdot 0.1\cdot B+p\cdot B=0.1B+0.9Bp.\]

Since, from \Cref{obs: expected cost}, $\expect{\sum_{j\in I}|\hat E(T_{j})|}\geq 0.9B$, we conclude that $p\geq 0.8$, and so $\prob{\event^g}\geq 0.8$, as required. In order to complete the proof of \Cref{claim: disjoint copies}, it is now enough to prove \Cref{obs: expected cost}, which we do next.

\begin{proofof}{\Cref{obs: expected cost}}
We associate a collection $\set{X_1,\ldots,X_r}$ of random variables  with the set $I$ of indices. In order to do so, we view set $I$ as being constructed as follows. For each $1\le i\le r$, sample a value $X_i$ uniformly at random from $[M]$, and then let $I=\set{X_1,\ldots,X_r}$ be the collection of these sampled values.
Consider now any index $1\leq i\leq r$, the corresponding set $T'_{X_i}$ of vertices of $H$, and any edge $e=(v_a,v_b)\in E_G(S)$. Edge $e'=(u_{f(a)+X_i},u_{f(b)+X_i})$ of $H$ corresponding to $e$ belongs to set $\hat E(T_{X_i})$ if and only if neither of the vertices $u_{f(a)+X_i}$, $u_{f(b)+X_i}$ lies in $\bigcup_{j\in I\setminus\set{i}}T'_{X_{j}}$. Consider now an index $j\in I\setminus\set{i}$. The probability that a fixed vertex $u\in V(H)$ lies in $T'_{X_j}$ is at most $k/M$ (since for every vertex $v_z\in S$, there is a single index $s$ with $u_{f(z)+s}=u$). From the union bound, the probability that a fixed vertex $u\in V(H)$ lies in $\bigcup_{j\in I\setminus\set{i}}T'_{X_{j}}$ is at most $\frac{r\cdot k}{M}\leq \frac{1}{\log k}$. In particular, the probability that any of the endpoints of edge  $e'=(u_{f(a)+X_i},u_{f(b)+X_i})$ lies in 
$\bigcup_{j\in I\setminus\set{i}}T'_{X_{j}}$  is at most $\frac{2}{\log k}$. Therefore, $\expect{|\hat E(T_j)|}\geq |E_G(S)|\cdot\left (1-\frac{2}{\log k}\right )\geq 0.9\cdot |E_G(S)|$. Altogether, from the linearity of expectation, 
$\expect{\sum_{1\le i\le r}|\hat E(T_{X_i})|}\ge 0.9\cdot |E_G(S)|\cdot r=0.9\cdot\floor{\frac{M}{k\log k}}\cdot \optdks(G,k)$.
\end{proofof}
\end{proof}

\paragraph{Ensemble and Bad Event $\event^b$.}
Next, we define the notion of an ensemble. Recall that $q=\ceil{\log N}$. An ensemble $\bset$ consists of a collection $I\subseteq [N]$ of at most $2q$ indices, and, for every index $i\in I$, an integer $-q\leq x_i\leq q$ with $x_i\neq 0$. We denote the ensemble by $\bset=\set{I,\set{x_i}_{i\in I}}$. We say that ensemble  $\bset=\set{I,\set{x_i}_{i\in I}}$ is \emph{bad} if $\sum_{i\in I}x_i\cdot f(i)\equiv 0 \text{ }(\textnormal{mod } M)$. We let $\event^b$ be the bad event that there exists a bad ensemble. We start by showing that the probability of Event $\event^b$ happening is low. Later, we show that, if Event $\event^b$ does not happen, then graph $H$ has some useful properties.

\begin{observation}\label{obs: bound bad ensemble}
	$\prob{\event^b}\leq \frac{1}{N^q}$.
\end{observation}

\begin{proof}
	Consider any fixed ensemble $\bset=\set{I,\set{x_i}_{i\in I}}$. Let $i^*\in I$ be any fixed index. Consider now the following two-step process: in the first step, we select the values $f(i)$ for all indices $i\in I\setminus \set{i^*}$ independently uniformly at random from $[M]$. We then denote $S=\sum_{i\in I\setminus \set{i^*}}x_i\cdot f(i)$. In the second step, we select the value $a=f(i^*)$ at random from $[M]$. 
	Ensemble $\bset$ is bad if and only if $a\cdot x_{i^*}+S=0\mod M$. Since $M$ is a prime number, there is exactly one value $a'\in [M]$ with $a'\cdot x_{i^*}+S=0\mod M$ (indeed, if two such values $a',a''$ exist, then $a'\cdot x_{i^*}=a''\cdot x_{i^*}\mod M$, implying that $a'=a''$ must hold). The probability to choose $f(i^*)=a'$ is then $1/M$, and so the probability that a fixed ensemble $\bset$ is bad is $1/M$.
	
	Notice that the total number of ensembles is bounded by $\Big(\sum_{t=1}^{2q}\binom{N}{t}\Big)\cdot (2q)^{2q}\le (2q)\cdot\binom{N}{2q}\cdot (2q)^{2q}\le N^{4q}$ . Using the Union Bound, $\prob{\event^b}\le \frac{N^{4q}}{M}\leq \frac{1}{N^q}$, since $M\geq N^{5q}$.	
\end{proof}

Next, we show that, if Event $\event^b$ does not happen, then every edge $e\in E(H)$ has a unique origin edge in $G$.

\begin{observation}\label{obs: unique origin of an edge}
	Assume that  Event $\event^b$ did not happen. Let $e$ be any edge of $H$. Then there is a unique edge $e'\in E(G)$, such that $e'$ is an origin edge of $e$.
\end{observation}
\begin{proof}
Denote $e=(u_j,u_{j'})$, and assume for contradiction that there are two distinct edges $e_1,e_2\in E(G)$ that both serve as origin edges of $e$. Denote $e_1=(v_{i_1},v_{i'_1})$ and $e_2=(v_{i_2},v_{i'_2})$. From the construction of $H$, since edge $e_1$ is an origin edge of $e$, there exists an integer $t_1$, such that:

\begin{equation}\label{eq: first edge}
j\equiv f(i_1)+t_1 \text{ }(\textnormal{mod } M),\quad\mbox{and}\quad j'\equiv f(i'_1)+t_1 \text{ }(\textnormal{mod } M).
\end{equation}

Similarly, since edge $e_2$ is an origin edge of $e$, there exists an integer $t_2$, such that:

\begin{equation}\label{eq: second edge}
j\equiv f(i_2)+t_2 \text{ }(\textnormal{mod } M),\quad\mbox{and}\quad\quad j'\equiv f(i'_2)+t_2 \text{ }(\textnormal{mod } M).
\end{equation}

By adding the first equation of (\ref{eq: first edge}) to the second equation of (\ref{eq: second edge}), we get that $f(i_1)+f(i_2')\equiv j+j'-t_1-t_2\text{ }(\textnormal{mod } M)$. Similarly, by adding the second equation of (\ref{eq: first edge}) to the first equation of (\ref{eq: second edge}), we get that $f(i'_1)+f(i_2)\equiv j+j'-t_1-t_2\text{ }(\textnormal{mod } M)$. In other words, we get that:

\begin{equation}\label{eq: all four}
f(i_1)+f(i'_2)\equiv f(i'_1)+f(i_2) \text{ }(\textnormal{mod } M).
\end{equation}

Since $e_1$ is an edge of $G$, $i_1\neq i_1'$ must hold, and similarly, since $e_2$ is an edge of $G$, $i_2\neq i_2'$ must hold. Moreover, since $e_1\neq e_2$, either $i_1\neq i_2$, or $i_1'\neq i_2'$ must hold. Combining this with Equation (\ref{eq: all four}), we conclude that both $i_1\neq i_2$ and $i_1'\neq i_2'$ must hold.

We now consider four cases. The first case is when all indices in $\set{i_1,i_1',i_2,i_2'}$ are distinct. In this case, we consider the ensemble $\bset=\set{I,\set{x_i}_{i\in I}}$, where $I=\set{i_1,i'_1,i_2,i'_2}$, with $x_{i_1}=x_{i'_2}=1$ and $x_{i_2}=x_{i'_1}=-1$. From Equation \ref{eq: all four}, we get that $\sum_{i\in I}x_i\cdot f(i)\equiv 0 \text{ }(\textnormal{mod } M)$, so ensemble $\bset$ is bad, contradicting the fact that Event $\event^b$ did not happen.

The second case is when $i_1=i_2'$ but $i_1'\neq i_2$. Then we construct an ensemble $\bset=\set{I,\set{x_i}_{i\in I}}$, where $I=\set{i_1,i'_1,i_2}$, with $x_{i_1}=2$ and $x_{i_2}=x_{i'_1}=-1$. As before, from Equation \ref{eq: all four}, we get that $\sum_{i\in I}x_i\cdot f(i)\equiv 0 \text{ }(\textnormal{mod } M)$, so ensemble $\bset$ is bad, contradicting the fact that Event $\event^b$ did not happen.
	
The third case is when $i_1'=i_2$ but $i_1\neq i_2'$. We consider the ensemble $\bset=\set{I,\set{x_i}_{i\in I}}$, where $I=\set{i_1,i'_1,i'_2}$, with $x_{i_1}=x_{i'_2}=1$ and $x_{i_2}=-2$. From Equation \ref{eq: all four}, we get that $\sum_{i\in I}x_i\cdot f(i)\equiv 0 \text{ }(\textnormal{mod } M)$, so ensemble $\bset$ is bad, contradicting the fact that Event $\event^b$ did not happen.

From the above discussion, the only remaining case is when both $i_1'=i_2$ and $i_1=i_2'$ hold. But in this case, $e_1=e_2$, contradicting our assumption that these two edges are distinct.
\end{proof}

Assume that bad event $\event^b$ did not happen. For every edge $e\in E(H)$, we denote by $R(e)$ the unique edge of $G$ that serves as the origin edge of $e$. For a subgraph $H'\subseteq H$, we let $R(H')$ be the subgraph of $G$ induced by the set $\set{R(e)\mid e\in E(H')}$ of edges; we refer to $R(H')$ as the \emph{origin graph of $H'$}. In other words, the set of edges of graph $R(H')$ is $\set{R(e)\mid e\in E(H')}$, and the set of its vertices contains every vertex of $G$ that serves as an endpoint to any of these edges.

 In the next observation we show that, if bad Event $\event^b$ did not happen, then for every cycle $C\subseteq H$ containing at most $q$ edges, every vertex in the corresponding origin-graph $R(C)$ has an even degree.

\begin{observation}\label{obs: cycle origin}
	Assume that Event  $\event^b$ did not happen. Let $C\subseteq H$ be any simple cycle containing at most $q$ edges, and let $G'=R(C)$ be the origin graph of $C$. Then every vertex of $G'$ has an even degree in $G'$.
\end{observation}

\begin{proof}
Throughout the proof, we assume that Event $\event^b$ did not happen, and we fix a simple cycle  $C=(u_{j_1},\ldots,u_{j_z})$ in $H$, with $z\le q$. For all $1\le i\le z$, we denote $e_i=(u_{j_i},u_{j_{i+1}})$, and we denote the origin-edge of $e_i$ by $e'_i=R(e_i)=(v_{a_i},v_{b_i})$. From the definition of graph $H$, there must be an integer $t_i\in [M]$, with $j_{i}\equiv f(a_i)+t_i \text{ }(\textnormal{mod } M)$ and $j_{i+1}\equiv f(b_i)+t_i \text{ }(\textnormal{mod } M)$. Therefore, for all $1\leq i\leq z$:
\[
f(a_i)-f(b_i)\equiv j_{i}-j_{i+1}\text{ }(\textnormal{mod } M).
\]
Summing up the above equality over all $i=1,\ldots z$, we get that
\begin{equation}\label{eq: ensemble is bad}
\sum_{1\le i\le z}f(a_i)-\sum_{1\le i\le z}f(b_i)\equiv 0\text{ }(\textnormal{mod } M).
\end{equation}

Let $A$ be the set of indices lying in $\set{a_1,\ldots,a_z}$ (if an index appears several times in $\set{a_1,\ldots,a_z}$, we only include it once in $A$). For every index $a^*\in A$, let $x'_{a^*}$ be the number of integers $i\in \set{1,\ldots,z}$ with $a_i=a^*$. Similarly, we let $B$ be the set of indices lying in $\set{b_1,\ldots,b_z}$, and for every index $b^*\in B$, we let $x''_{b^*}$ be the number of integers $i\in \set{1,\ldots,z}$ with $b_i=a$. We claim that $A=B$ must hold, and, for every index $a\in A$, $x'_a=x''_a$ must hold. Indeed, assume otherwise. We then construct an ensemble $\bset=\set{I,\set{x_i}_{i\in I}}$ as follows. Set $I$ includes every index $i\in (A\setminus B)\cup (B\setminus A)$; for each such index $i$, we set $x_i=x'_i$ if $i\in A\setminus B$ and $x_i=-x''_i$ otherwise. Additionally, for every index $i\in A\cap B$ with $x'_i\neq x''_i$, we include index $i$ in $I$, with $x_i=x'_i-x''_i$. From our assumptions, $I\neq\emptyset$, $|I|\leq 2q$, and for all $i\in I$, $-q\leq x_i\leq q$, with $x_i\neq 0$. Therefore, 
$\bset=\set{I,\set{x_i}_{i\in I}}$ is a valid ensemble. But then, from Equation \ref{eq: ensemble is bad}, $\sum_{i\in I}x_i\cdot f(i)\equiv 0 \text{ }(\textnormal{mod } M)$. In other words, ensemble $\bset$ is bad, contradicting the assumption that bad event $\event^b$ did not happen.

We conclude that $A=B$ must hold, and, for every index $a\in A$, $x'_a=x''_a$ must hold.
Therefore, for every vertex $v\in V(G)$, the number of times that $v$ lies in $\set{v_{a_1},\ldots,v_{a_z}}$ is equal to the number of times that $v$ lies in $\set{v_{b_1},\ldots,v_{b_z}}$. Therefore, the number of edges of $\set{R(e_i)\mid 1\leq i\leq z}$ that are incident to $v$ is even.
\end{proof}

Lastly, we need the following claim.

\begin{claim}\label{claim: small origin}
	Let $H'$ be any subgraph of $H$ with $|V(H')|\leq N$. If Event $\event^b$ did not happen, then the origin graph $R(H')$ of $H'$ contains at most $c^*\cdot |V(H')|$ vertices, where $c^*$ is a constant independent of $N$.
\end{claim}

\begin{proof}
	Recall that the \emph{girth} of an unweighted graph $G^*$ is the length of the shortest cycle in $G^*$. For an integer $t\geq 1$, we say that a subgraph $G'$ of $G^*$ is a \emph{$t$-spanner} of $G^*$ if $V(G')= V(G^*)$, and, for every pair $v,v'$ of vertices of $G^*$, if we denote by $\dist_{G^*}(v,v')$ the length of the shortest $v$-$v'$ path in $G^*$, and we define $\dist_{G'}(v,v')$ similarly for $G'$, then  $\dist_{G'}(v,v')\le t\cdot \dist_{G^*}(v,v')$.
	
	Consider now any subgraph $H'$ of $H$. 
	We use the following algorithm of \cite{althofer1993sparse}, whose goal is to construct a $q$-spanner $H''$ of $H'$ that contains few edges. The algorithm starts with graph $H''$, whose vertex set is $V(H'')=V(H')$, and edge set is empty. It then processes every edge $e\in E(H')$ one by one. If graph $H''\cup\set{e}$ contains a cycle of length at most $q$, then we continue to the next iteration; otherwise, we add $e$ to $H''$, and continue to the next iteration. Consider the final graph $H''$ that is obtained at the end of the algorithm, once very edge of $H'$ is processed. It is immediate to see that the girth of $H''$ is greater than $q$. One can also show that the resulting graph  $H''$ is a $q$-spanner of $H'$, but we do not need to use this fact.
	We use the following theorem from \cite{bollobas2004extremal}.
	
	\begin{theorem}[Theorem 3.7 from \cite{bollobas2004extremal}]
		\label{thm: spanner}
		Let $G$ be an $n$-vertex graph with girth greater than $q$, for any integer $q>1$. Then $|E(G)|\le n\cdot \ceil{n^{2/(q-2)}}$.
	\end{theorem} 

From the above theorem, $|E(H'')|\leq |V(H')|^{1+O(1/q)}=O(|V(H')|)$, as $|V(H')|\le N$ and $q=\ceil{\log N}$. 

We denote by $W'\subseteq G$ the origin graph of $H'$, and we denote by $W''\subseteq G$ the origin graph of $H''$. 
Note that $|E(W'')|\leq |E(H'')|\leq O(|V(H')|)$, and so $|V(W'')|\leq O(|E(W'')|)\leq O(|V(H')|)$. We next show the following observation.

\begin{observation}\label{obs: vertex set the same}
	If Event $\event^b$ did not happen, then $V(W')=V(W'')$.
	\end{observation} 

Notice that the observation implies that $|V(W')|\leq O(|V(H')|)$,  completing the proof of \Cref{claim: small origin}. It now remains to prove \Cref{obs: vertex set the same}.

\begin{proofof}{\Cref{obs: vertex set the same}}
 Consider any edge $e\in E(H')\setminus E(H'')$. From the construction of graph $H''$, there must be a simple cycle $C$ in graph $H''\cup \set{e}$, whose length is at most $q$. Consider now the subgraph $W_C$ of $G$, induced by the edges of $\set{R(e')\mid e'\in E(C)}$; in other words, $W_C=R(C)$. Since the event $\event^b$ did not happen, from \Cref{obs: cycle origin}, graph $W_C$ is an even-degree graph. Therefore, if we denote by $\hat e=R(e)$ the origin-edge of $e$, then graph $W_C\setminus\set{\hat e}$ contains exactly two odd-degree vertices, that serve as endpoints of edge $\hat e$ in $G$. Notice however that $W_C\setminus\set{\hat e}\subseteq W''$. Therefore, for every edge $e\in E(H')\setminus E(H'')$, the endpoints of the origin edge $R(e)$ lie in $W''$. It then follows that $W'=W''$. 
\end{proofof}
\end{proof}

\paragraph{Bad Event $\event$.}
We say that the bad event $\event$ happens if either bad event $\event^b$ happens, or bad event $\event^g$ does not happen. By using the Union bound, together with \Cref{claim: disjoint copies} and \Cref{obs: bound bad ensemble}, we get that $\prob{\event}\leq 0.1$.

We are now ready to complete the proof of \Cref{thm: alg_DkC gives alg_DkS}.

\subsection{Completing the Reduction from \DkS to \DkC}
\label{subsec: finish DkS to DkC}

Let $\pdks(G,k)$ be an input instance of the \DkS problem. Denote $V(G)=\set{v_1,\ldots,v_N}$. We start by constructing the auxiliary graph $H$, from instance $\pdks(G,k)$. We add isolated vertices to graph $H$, until $|V(H)|$ becomes an integral multiple of $k$, and we denote $|V(H)|=n$. Clearly, $n\leq N^{O(\log N)}$. We then consider instance $\pcl(H,k)$ of the \DkC problem, where the parameter $k$ remains unchanged. Note that, if Event $\event$ did not happen, then, from the definition of Events $\event$ and $\event^g$, there is a collection $\set{H_1,\ldots,H_r}$ of $r=\floor{\frac M {k\log k}}$ disjoint subsets of vertices of $H$, such that for all $1\le j\leq r$, $|V(H_j)|\leq k$ holds, and additionally, $\sum_{1\le j\le r}|E(H_j)|\ge 0.1\cdot \floor{\frac M{k\log k}}\cdot\optdks(G,k)$. We can then define a solution $(S_1,\ldots,S_{n/k})$ to instance $\pcl(H,k)$ of the \DkC problem, as follows. For all $1\leq i\leq r$, we initially set $S_i=V(H_i)$, and for all $r<i\leq n/k$, we set $S_i=\emptyset$. Let $U=V(H)\setminus\left(\bigcup_{i=1}^rV(H_i)\right )$. Next, we partition the vertices of $U$ by adding them to sets $S_1,\ldots,S_{n/r}$ arbitrarily, to ensure that the cardinality of each set is exactly $k$. From the above discussion, if Event $\event$ did not happen, then:
 $$\sum_{i=1}^{n/k}|E_H(S_i)|\geq \sum_{i=1}^r|E(H_i)|\geq 0.1\cdot \floor{\frac M{k\log k}}\cdot\optdks(G,k)\geq \Omega\left (\frac n {k\log k}\right )\cdot\optdks(G,k).$$ 
 
 We conclude that, if Event $\event$ did not happen, then $\optcl(H,k)\ge \Omega\left (\frac n{k\log k}\right )\cdot\optdks(G,k)$.

 We apply the $\alpha(n)$-approximation algorithm for \DkC to instance $\pcl(H,k)$, and we denote the resulting solution by $(U_1,\ldots,U_{n/k})$. Note that: 
 
 $$\sum_{i = 1}^{n/k}|E_H(U_i)|\ge \frac{\optcl(H,k)}{\alpha(n)}\geq \Omega\left (\frac n{k\cdot \alpha(n)\cdot \log k}\right )\cdot\optdks(G,k) $$

must hold.
We let $U\in \set{U_1,\ldots,U_{n/k}}$ be a subset maximizing $|E_H(U_i)|$, so that 

$$|E_H(U)|\ge \frac{\sum_{i = 1}^{n/k}|E_H(U_i)|}{n/k}\geq \Omega\left (\frac {\optdks(G,k)}{\alpha(n)\cdot \log k}\right ).$$

Let $H'=H[U]$, and let $W=R(H')$ be the origin graph of $H'$. 
From \Cref{claim: small origin}, if Event $\event$ did not happen, then $|V(W)|\leq c^*\cdot |V(H')|\leq c^*k$, for some universal constant $c^*$.

Lastly, we apply the algorithm from \Cref{lem: size reducing} to graph $W$, to obtain a subgraph $W'$ of $W$ with $|V(W')|=k$ and $|E(W')|\ge \Omega(|E(W)|)$. We then return $S=|V(W')|$ as the solution to the input instance  $\pdks(G,k)$  of the \DkS problem. From the above discussion, $|E_G(S)|\geq \Omega(|E(W)|)\geq \Omega\left (\frac {\optdks(G,k)}{\alpha(n)\cdot \log k}\right )\geq \Omega\left (\frac {\optdks(G,k)}{\alpha\left (N^{O(\log N)}\right ) \cdot \log N}\right )$.
Therefore, if the event $\event$ does not happen, we obtain an 	$O(\alpha(N^{O(\log N)})\cdot\log N)$-approximate solution to the input instance of the \DkS problem. Recall that the probability of Event $\event$ happening is at most $0.1$. Lastly, since $|V(H)|\leq N^{O(\log N)}$, it is easy to verify that the running time of the algorithm is at most $N^{O(\log N)}$.

\subsection{Completing the Reduction from \DkS to \WGP}
\label{subsec: finish DkS to WGP}

Let $\pdks(G,k)$ be an input instance of the \DkS problem with $|V(G)|=N$. 
Our algorithm requires the knowledge of an estimate $h$ on the value of $\optdks(G,k)$, with $h/2\leq \optdks(G,k)\leq h$. In order to overcome this difficulty, we run the algorithm for every value of $h$ that is an integral power of $2$ between $1$ and $|E(G)|$, and output the best resulting solution. Therefore, it is now enough to provide a randomized algorithm that, given an estimate $h$ with $h/2\leq \optdks(G,k)\leq h$, with a constant probability produces a solution to instance $\pdks(G,k)$ of \DkS whose value is at least
$\Omega\left (\frac{\optdks(G,k)}{(\alpha(N^{O(\log N)}))^3\cdot\log^2 N}\right )$, such that the running time of the algorithm is  $N^{O(\log N)}$. From now on we assume that we are given an integer $h$ with $h/2\leq \optdks(G,k)\leq h$.

As before, we denote $V(G)=\set{v_0,\ldots,v_{N-1}}$, and we construct the auxiliary graph $H$ from instance $\pdks(G,k)$ of \DkS. We denote $|V(H)|=n$, so $n\leq N^{O(\log N)}$ holds. We then consider instance
$\pwgp(H, r, h)$  of the \WGP problem, where $r=\floor{\frac n {k\log k}}$. 

Note that, if Event $\event$ did not happen, then, from the definition of Events $\event$ and $\event^g$, there is a collection $\set{H_1,\ldots,H_r}$ of $r=\floor{\frac n {k\log k}}$ disjoint subgraphs of $H$, such that for all $1\le j\leq r$, $|E(H_j)|\leq \optdks(G,k)\leq h$ holds, and
$\sum_{1\le j\le r}|E(H_j)|\ge 0.1\cdot \floor{\frac n{k\log k}}\cdot\optdks(G,k)$.
Therefore, we obtain a solution $(H_1,\ldots,H_r)$ to instance $\pwgp(H,r,h)$ of \WGP, whose value is at least $0.1\cdot \floor{\frac n{k\log k}}\cdot\optdks(G,k)$.
We conclude that, if Event $\event$ did not happen, then $\optwgp(H,r,h)\ge \Omega\left (\frac n{k\log k}\right )\cdot\optdks(G,k)$.

We apply the $\alpha(n)$-approximation algorithm to instance $\pwgp(H,r,h)$ of \WGP, obtaining a solution $(H'_1,\ldots,H'_r)$, whose value is at least $\frac{\optwgp(H,r,h)}{\alpha(n)}$.
Denote $\hset=\set{H'_1,\ldots,H'_r}$. We partition set $\hset$ into two subsets: set $\hset'$ containing all graphs $H'_j\in \hset$ with $|V(H')|\leq 
100\cdot\alpha(n) k\log k$, and set $\hset''$ containing all remaining graphs. 

Assume for now that Event $\event$ did not happen. Then, as observed above:

\[\sum_{H'_j\in \hset}|E(H'_j)|\geq \frac{\optwgp(H,r,h)}{\alpha(n)}\geq  \frac{n\cdot \optdks(G,k)}{20k\cdot \alpha(n)\cdot \log k }. \]

Clearly, $|\hset''|\leq \frac{n}{100\cdot \alpha(n)\cdot k\log k}$, and so: 

$$\sum_{H'_j\in \hset''}|E(H'_j)|\leq \frac{n\cdot h}{100\cdot \alpha(n)\cdot k\log k}\leq \frac{n\cdot\optdks(G,k)}{100\cdot \alpha(n)\cdot k\log k}.$$

Altogether, we get that, if Event $\event$ did not happen, then:

\[\sum_{H'_j\in \hset'}|E(H)|\geq  \frac{n\cdot \optdks(G,k)}{50k\cdot \alpha(n)\cdot \log k }. \]

We let $H^*\in \hset'$ be the graph maximizing the number of edges. Since $|\hset'|\leq |\hset|=r=\floor{\frac n {k\log k}}$, from the above discussion, if Event $\event$ did not happen, then:

\[|E(H^*)|\geq \frac{n\cdot \optdks(G,k)}{r\cdot 50k\cdot \alpha(n)\cdot \log k }\geq  \frac{\optdks(G,k)}{50 \alpha(n)}.\]

From the definition of the collection $\hset'$ of graphs, $|V(H^*)|\leq 100\cdot\alpha(n) k\log k$.

Let $W=R(H^*)$ be the origin graph of $H^*$. 
From \Cref{claim: small origin}, if Event $\event$ did not happen, then $|V(W)|\leq c^*\cdot |V(H^*)|\leq O(\alpha(n) k\log k)$.

Lastly, we apply the algorithm from \Cref{lem: size reducing} to graph $W$, to obtain a subgraph $W'$ of $W$ with $|V(W')|\leq k$ and $|E(W')|\ge \Omega\left(\frac{|E(W)|}{(\alpha(n))^2\cdot \log ^2k}\right )$. We then return $S=|V(W')|$ as the solution to the input instance  $\pdks(G,k)$  of the \DkS problem. From the above discussion, 
$|V(S)|\leq k$, and, if Event $\event$ did not happen, then:

\[\begin{split}
|E_G(S)|&\geq \Omega\left(\frac{|E(W)|}{(\alpha(n))^2\cdot \log ^2k}\right ) \\ &\geq \Omega\left(\frac{|E(H^*)|}{(\alpha(n))^2\cdot \log ^2k}\right )\\
&\geq \Omega\left(\frac{\optdks(G,k)}{(\alpha(n))^3\cdot \log ^2k}\right )\\
&\geq \Omega\left(\frac{\optdks(G,k)}{(\alpha(N^{O(\log N)}))^3\cdot \log^2N}\right ).
\end{split}
\]

Therefore, if the event $\event$ does not happen, we obtain an 	$O((\alpha(N^{O(\log N)}))^3\cdot\log^2 N)$-approximate solution to the input instance of the \DkS problem. Recall that the probability of Event $\event$ happening is at most $0.1$. Lastly, since $|V(H)|\leq N^{O(\log N)}$, it is easy to verify that the running time of the algorithm is at most $N^{O(\log N)}$.

\section{Reductions between \WGP and \BCS}
\label{sec: WGP and BCS}

In this section we establish a connection between the \WGP and \BCS problems, by proving the following two theorems.

\begin{theorem}
\label{thm: alg_WGP gives alg_CN}
Let $\alpha: \mathbb{Z^+}\to \mathbb{Z^+}$ be an increasing function with $\alpha(n)=o(n)$. Assume that there exists an efficient $\alpha(n)$-approximation algorithm for the \WGP problem, where $n$ is the number of vertices in the input graph. Then there exists an efficient $O(\alpha(N)\cdot\poly\log N)$-approximation algorithm for \BCS, where $N$ is the number of vertices in the input instance of \BCS.
\end{theorem}

\begin{theorem}\label{thm: alg for CN gives alg for WGP}
	Let $\alpha: \mathbb{Z^+}\to \mathbb{Z^+}$ be an increasing function with $\alpha(n)=o(n)$. Assume that there exists an efficient $\alpha(N)$-approximation algorithm for the \BCS problem, where $N$ is the number of vertices in the input graph. Then there exists an efficient $O((\alpha(n))^2\cdot\poly\log n)$-approximation algorithm for \WGP, where $n$ is the number of vertices in the input instance of \WGP.
\end{theorem}

By combining \Cref{thm: alg for CN gives alg for WGP} with \Cref{cor: hardness of WGP and DkC}, we obtain the following corollary.

\begin{corollary}\label{cor: hardness of BDS}
	Assume that \Cref{conj: main1} holds and that  $\NP\not \subseteq \BPTIME(n^{O(\log n)})$. Then for some constant $0<\eps'\leq 1/2$, there is no efficient $2^{(\log n)^{\eps'}}$-approximation algorithm for \BCS.
\end{corollary}

\begin{proof}
	Assume that \Cref{conj: main1} holds and that  $\NP\not \subseteq \DTIME(n^{O(\log n)})$. Then, from \Cref{cor: hardness of WGP and DkC}, for some constant $0<\eps\leq 1/2$, there is no efficient factor-$2^{(\log n)^{\eps}}$-approximation algorithm for \WGP, where $n$ is the number of vertices in the input graph.
	
	We let $\eps'=\eps/c$, where $c$ is a sufficiently large constant. We now prove that there is no efficient $2^{(\log n)^{\eps'}}$-approximation algorithm for \BCS. Indeed, assume for contradiction that there is an efficient $2^{(\log n)^{\eps'}}$-approximation algorithm $\aset$ for \BCS.  From
	\Cref{thm: alg for CN gives alg for WGP}, there is an efficient  $c'\cdot (\alpha(N))^2\cdot\poly\log N$-approximation algorithm for \WGP, where $N$ is the number of vertices in the input instance of \WGP, $c'$ is some constant, and $\alpha(N)=2^{(\log N)^{\eps'}}$. Notice however that $c'\cdot (\alpha(N))^2\cdot\poly\log N\leq 2^{(\log N)^{\eps}}$ holds, if the constant $c$ is large enough. 
	
	Therefore, we obtain an efficient factor-$2^{(\log n)^{\eps}}$-approximation algorithm for \WGP, a contradiction.
\end{proof}

In the remainder of this section, we prove Theorems \ref{thm: alg for CN gives alg for WGP} and \ref{cor: hardness of WGP and DkC}.
We start by proving two auxiliary lemmas that will be used in the proofs of both theorems. We then complete the proofs of \Cref{thm: alg_WGP gives alg_CN} and  \Cref{thm: alg for CN gives alg for WGP} in sections \Cref{subsec: from BCS to WGP} and \Cref{subsec: from WGP to BCS}, respectively.

\subsection{Auxiliary Lemmas}
\label{subsec: aux lemma}

We start with the following definition, that will be used in the proofs of both theorems.

\begin{definition}
Let $\pwgp(G,r,h)$ be an instance of \WGP, and let $\set{H_1,\ldots,H_r}$  be a solution to this instance. We say that this solution is \emph{good}, if for all $1\leq i\leq r$, $h/2\le |E(H_i)|\le h$.
\end{definition}


We are now ready to state the first auxiliary lemma.

\begin{lemma}
	\label{obs: equal_size}
	There is an efficient algorithm, that, given a graph $G$ with $|V(G)|=n$, integers $r,h>0$ and any solution $\hset$ to instance $\pwgp(G,r,h)$ of \WGP, computes positive integers $r^*\le r, h^*\le h$ and a subset $\hset^*\subseteq \hset$ of subgraphs of $G$, such that $\hset^*$ is a good solution to instance $\pwgp(G,r^*,h^*)$, of \WGP, and $\sum_{H\in \hset^*}|E(H)|\geq \frac{\sum_{H'\in \hset}|E(H')|}{4\log n}$. 
\end{lemma}

\begin{proof}
	The idea of the proof is to partition the graphs in $\hset$ geometrically into groups by the cardinalities of their edge sets, and then select a group maximizing the total number of edges in its subgraphs.
	
	Specifically, let $q=\ceil{2\log n}$. For all $1\leq i\leq q$, we let $\hset_i\subseteq \hset$ contain all graphs $H$ with  $2^{i-1}\le |E(H)|< 2^{i}$. It is easy to verify that $\hset_1\ldots,\hset_q$ partition $\hset$. Clearly, there must be an index $1\leq i^*\leq q$, with $\sum_{H\in \hset_{i^*}}|E(H)|\geq \frac{\sum_{H'\in \hset}|E(H')|}{q}\geq \frac{\sum_{H'\in \hset}|E(H')|}{4\log n}$. We set $r^*=|\hset_{i^*}|$, $h^*=2^{i^*}$, and we let $\hset^*=\hset_{i^*}$. It is immediate to verify that $\hset^*$ is a good solution to instance $\pwgp(G,r^*,h^*)$, of \WGP, and, from the above discussion,  $\sum_{H\in \hset^*}|E(H)|\geq \frac{\sum_{H'\in \hset}|E(H')|}{4\log n}$. 
%
\end{proof}

We are now ready to prove our second auxiliary lemma.

\begin{lemma}
\label{lem: WGP_BCS}
There is an efficient algorithm, whose input consists of an instance $\pcn(G,L)$ of the \BCS problem with $|V(G)|=N$,  where $N$ is greater than a sufficiently large constant, together with a solution $H$ to this instance, such that  $|E(H)|\ge 4N\log^6N$ holds. The algorithm computes integers $r,h>0$, such that $r\cdot h^2\le L\cdot \log^6N$ and $r\cdot h\ge \Omega(|E(H)|/ \log N)$ hold, together with a good solution $\set{H_1,\ldots,H_r}$ to instance $\pwgp(G,r,h)$ of \WGP.
\end{lemma}

\begin{proof}
Let $\hat G$ be any graph. A \emph{cut} in $\hat G$ is a partition $(A,B)$ of vertices of $\hat G$ into two non-empty subsets. The \emph{value} of the cut $(A,B)$ is $|E(A,B)|$. For a parameter $1/2<\beta<1$, we say that cut $(A,B)$ is \emph{$\beta$-balanced}, if $|E(A)|,|E(B)|\le \beta\cdot |E(\hat G)|$. We say that cut $(A,B)$ is a \emph{minimum $\beta$-balanced cut} if it is a $\beta$-balanced cut whose value is the smallest among all such cuts.  We use the following theorem that follows from the results of \cite{ARV}, and was formally proved in \cite{chuzhoy2022subpolynomial}. 
	
	\begin{theorem}[Theorem 4.11 in the full version of \cite{chuzhoy2022subpolynomial}]
		\label{thm: bal_cut alg}
		There is an efficient algorithm, that, given a graph $\hat G$ with $|V(\hat G)|=\hat N$, computes a $\gamma$-balanced cut in $\hat G$, whose value is at most $O(\sqrt{\log \hat N})$ times the value of the minimum $(3/4)$-balanced cut in $\hat G$, for some universal constant $3/4<\gamma<1$ that does not depend on $\hat N$. 
	\end{theorem}
	
	We use the following theorem, that is a simple corollary of the Planar Separator Theorem by Lipton and Tarjan \cite{lipton1979separator}, and was formally proved in \cite{chuzhoy2022subpolynomial}. 
	A variation of this theorem for vertex-balanced cuts was proved in \cite{pach1996applications}. 
	
	\begin{theorem}[Lemma 4.12 in the full version of \cite{chuzhoy2022subpolynomial}]
		\label{thm: bal_cut size}
		Let $\hat G$ be a connected graph with $m$ edges and maximum vertex degree $\Delta<\frac{m}{2^{40}}$. If $\CrN(\hat G)\le \frac{m^2}{2^{40}}$, then the value of the minimum $(3/4)$-balanced cut in $\hat G$ is at most $O\left (\sqrt{\CrN(\hat G)+\Delta\cdot m}\right )$.
	\end{theorem}
	
	We are now ready to complete the proof of \Cref{lem: WGP_BCS}. Recall that we are given an instance $\pcn(G,L)$  of \BCS, where $|V(G)|=N$, together with a solution $H$ to this instance, such that  $|E(H)|\ge 4N\log^6 N$. 
	
	The algorithm starts by iteratively decomposing graph $H$ into smaller subgraphs. Throughout the decomposition procedure, we maintain a collection $\hset$ of connected subgraphs of $H$, that are all mutually disjoint. Each graph $H'\in \hset$ is marked as either \emph{active} or \emph{inactive}. At the beginning of the algorithm, we let $\hset$ contain all connected components of $H$, which are all marked as active. The algorithm performs iterations, as long as at least one graph in $\hset$ is inactive. 
	
	In order to execute an iteration, we select an arbitrary active graph $H'\in \hset$. We apply the algorithm from \Cref{thm: bal_cut alg} to compute a $\gamma$-balanced cut $(A,B)$ of $H'$. If $|E_{H'}(A,B)|\ge \frac{|E(H')|}{\log^2 N}$, then we mark $H'$ as inactive and continue to the next iteration. Otherwise, we remove graph $H'$ from $\hset$, and we add all connected components of graphs $H'[A]$ and $H'[B]$ to $\hset$, that are all marked as active graphs. We then continue to the next iteration. This completes the description of the decomposition procedure. Let $\hset'$ be the collection $\hset$ of subgraphs of $H$ that we obtain at the end of the procedure. We prove the following simple observation.
	
	\begin{observation}\label{obs: many edges remain}
		$\sum_{H'\in \hset'}|E(H')|\ge |E(H)|/2$.
	\end{observation}
	\begin{proof}
		We use a charging scheme. We observe the set $\hset$ of graphs over the course of the partitioning procedure. Throughout the execution of the partitioning procedure, we denote by $E'=E(H)\setminus \left(\bigcup_{H'\in \hset}E(H')\right )$, and we call the edges of $E'$ \emph{deleted edges}. Over the course of the partitioning procedure we maintain, for every edge $e\in E(H)$, a non-negative value $c(e)$, that we refer to as the \emph{charge} of $e$. We will ensure that, at every point of the algorithm's execution, $\sum_{e\in E(H)}c(e)\geq |E'|$, and that, for every edge $e\in E(H)$, $c(e)\leq 1/2$ always holds. We note that, even when an edge $e$ is added to the set $E'$ of deleted edges, its charge $c(e)$ may remain strictly positive. It is then easy to verify that, at the end of the algorithm, $|E'|\leq \sum_{e\in E(H)}c(e)\leq |E(H)|/2$ holds, and so $\sum_{H'\in \hset'}|E(H')|\geq |E(H)|/2$.
		
		It now remains to describe the assignment of the charge values $c(e)$ to the edges $e\in E(H)$, for which the above properties hold. Initially, $E'=\emptyset$, and we set $c(e)=0$ for every edge $e\in E(H)$.
		
		Consider now some iteration of the algorithm, and assume that, at the beginning of the iteration, $\sum_{e\in E(H)}c(e)\geq |E'|$ holds. Let $H'\in \hset$ be the graph that was processed in the current iteration, and let $(A,B)$ be the cut in $H'$ that the algorithm computed. If $|E_{H'}(A,B)|\ge \frac{|E(H')|}{\log^2 N}$, then no new edges were added to $E'$ in the current iteration, and the charge values $c(e)$ remain unchanged for all edges $e\in E(H)$. Assume now that $|E_{H'}(A,B)|< \frac{|E(H')|}{\log^2 N}$ holds, and denote $E''=E_{H'}(A,B)$. Then in the current iteration, the edges of $E''$ were added to set $E'$. We increase the charge $c(e)$ of every edge $e\in E(H')$ by $\frac{|E''|}{|E(H')|}$, and leave all other edge charges unchanged. This ensures that $\sum_{e\in E(H)}c(e)\geq |E'|$ holds at the end of the iteration. Since  $|E''|< \frac{|E(H')|}{\log^2 N}$, for every edge $e\in E(H')$, the charge $c(e)$ increases by at most $\frac 1{\log^2N}$ in the current iteration.
		
		From the above discussion, at the end of the algorithm, $|E'|\leq \sum_{e\in E(H)}c(e)$ holds. It now remains to show that for every edge $e\in E(H)$, $c(e)\leq 1/2$ holds at the end of the algorithm.
		
		Consider any edge $e\in E(H)$, and denote by $H_1,H_2,\ldots,H_r$ the sequence of subgraphs of $H$ that belonged to $\hset$ over the course of the algorithm, and contained $e$. In other words, $H_1=H$, and, for all $1<i\leq r$, graph $H_i$ was obtained via a balanced cut from graph $H_{i-1}$. Then the charge of $e$ has increased in at most $r+1$ iterations, and in each such iteration, the increase in the charge was bounded by   $\frac 1{\log^2N}$. Furthermore, for all $1<i\leq r$, $|E(H_i)|\leq \gamma\cdot |E(H_{i+1})|$ holds, and so $r\leq O(\log N)$ as $\gamma$ is a constant. Therefore, at the end of the algorithm, $c(e)\leq \frac{r+1}{\log^2N}\leq \frac 1 2$, since we have assumed that $N$ is sufficiently large.
	\end{proof}
	
	Consider now the final collection $\hset'$ of graphs. We say that a graph $H'\in \hset'$ is \emph{dense} iff $|E(H')|\ge |V(H')|\cdot \log^6 N$, and otherwise we say it is \emph{sparse}. We partition the set $\hset'$ of graphs into a collection $\hset^d$ containing all dense graphs and a collection $\hset^s$ containing all sparse graphs. We need the following simple observation.

	\begin{observation}\label{obs: few edges in dense graphs}
		$\sum_{H'\in \hset^d}|E(H')|\geq \frac{|E(H)|}{4}\geq N\log^6N$.
	\end{observation}
\begin{proof}
	Assume for contradiction that $\sum_{H'\in \hset^d}|E(H')|< \frac{|E(H)|}{4}$. Since, from \Cref{obs: many edges remain}
	$\sum_{H'\in \hset'}|E(H')|\ge |E(H)|/2$, we get that $\sum_{H'\in \hset^s}|E(H')|> \frac{|E(H)|}{4}$. However:
	
	\[ \sum_{H'\in \hset^s}|E(H')|\leq \sum_{H'\in \hset^s}|V(H')|\cdot \log^6N\leq N\log^6N.  \]
	
	We then conclude that $|E(H)|< 4\sum_{H'\in \hset^s}|E(H')|\leq 4N\log^6N$, contradicting the statement of \Cref{lem: WGP_BCS}.
\end{proof}

We also need the following obsevation.

\begin{observation}\label{obs: dense graph high CN}
	Let $H'\in \hset^d$ be a dense graph. Then $\CrN(H')\geq\Omega\left(\frac{|E(H')|^2}{\log^{5} N}\right )$.
\end{observation}

We provide the proof of \Cref{obs: dense graph high CN} below, after we complete the proof of \Cref{lem: WGP_BCS} using it. Let $r'=|\hset^d|$ and $h'=\max_{H'\in \hset^d}\set{|E(H')|}$. Clearly, collection $\hset^d$ of graphs is a valid solution to instance $\pwgp(G,r',h')$ of \WGP. We apply the algorithm from \Cref{obs: equal_size} to the soluton $\hset^d$ to instance $\pwgp(G,r',h')$ of \WGP. Recall that the algorithm computes positive integers $r\leq r',h\leq h'$, and a subset $\hset^*\subseteq \hset^d$ of subgraphs of $G$, such that $\hset^*$ is a good solution to instance $\pwgp(G,r,h)$, of \WGP, and $\sum_{H'\in \hset^*}|E(H')|\geq \frac{\sum_{H'\in \hset^d}|E(H')|}{4\log N}$. 
It now remains to verify that $r\cdot h^2\le L\cdot \log^6N$ and $r\cdot h\ge \Omega(|E(H)|/ \log N)$ hold.


Observe first that:

\[r\cdot h\geq \sum_{H'\in \hset^*}|E(H')|\ge \frac{\sum_{H'\in \hset^d}|E(H')|}{4\log N} \geq \frac{|E(H)|}{16\log N}, \]

from \Cref{obs: few edges in dense graphs}.

Finally, since for each graph $H'\in \hset^d$, $\CrN(H')\geq\Omega\left(\frac{|E(H')|^2}{\log^{5} N}\right )$ holds, and since, for every graph $H'\in \hset^*$, $|E(H')|\geq h/2$ holds, we get that:

\[r\cdot h^2\le \sum_{H'\in \hset^*}4\cdot |E(H')|^2
\le \sum_{H'\in \hset^*}O(\log^{5} N)\cdot \CrN(H')
\le O(\log^5 N)\cdot \CrN(H)\le L\cdot \log^6 N,
\]

since $N$ is large enough.

In order to complete the proof of \Cref{lem: WGP_BCS}, it is now enough to prove \Cref{obs: dense graph high CN}, which we do next.

\begin{proofof}{\Cref{obs: dense graph high CN}}
	 Since graph $H'$ is marked inactive by the algorithm, the $\gamma$-balanced cut of $H'$ computed by the algorithm from \Cref{thm: bal_cut alg} had value at least $\frac{|E(H')|}{\log^2 N}$. Therefore the minimum $(3/4)$-balanced cut of $H'$ has value at least $\Omega\left (\frac{|E(H')|}{\log^{2.5} N}\right )$. 

	Let $\Delta$ denote the maximum vertex degree in $H'$. Clearly, $\Delta\leq V(H')$ must hold. On the other hand, from the definition of a dense graph, $|E(H')|\geq |V(H')|\cdot \log^6N$, and so $\Delta\leq |V(H')|\leq \frac{|E(H')|}{\log^6N}<\frac{|E(H')|}{2^{40}}$, if $N$ is sufficiently large.

	Recall that, from \Cref{thm: bal_cut size}, either $\CrN(H')\ge \frac{|E(H')|^2}{2^{40}}$, or the value of minimum $(3/4)$-balanced cut in $H'$ is at most  $\sqrt{\CrN(H')+\Delta\cdot |E(H')|}$. In the former case, we immediately get that $\CrN(H')\geq\Omega\left(\frac{|E(H')|^2}{\log^{5} N}\right )$. In the latter case, since the value of the minimum $(3/4)$-balanced cut in $H'$ is $\Omega\left (\frac{|E(H')|}{\log^{2.5} N}\right )$, we get that
	 $\sqrt{\CrN(H')+\Delta\cdot |E(H')|}\geq \Omega\left (\frac{|E(H')|}{\log^{2.5} N}\right )$. Moreover, since
	$\Delta\leq \frac{|E(H')|}{\log^6N}$:
	
	\[ \CrN(H')\geq \Omega\left (\frac{|E(H')|^2}{\log^{5} N}\right )-\Delta\cdot |E(H')|\geq \Omega\left (\frac{|E(H')|^2}{\log^{5} N}\right )-\frac{|E(H')|^2}{\log^6N}\geq  \Omega\left (\frac{|E(H')|^2}{\log^{5} N}\right ),    \]
	
	since $N$ is sufficiently large.
\end{proofof}	
\end{proof}

\subsection{Reduction from \BCS to \WGP: Proof of \Cref{thm: alg_WGP gives alg_CN}}
\label{subsec: from BCS to WGP}

In this subsection we prove \Cref{thm: alg_WGP gives alg_CN}. 
Let $\pcn(G,L)$ be a given instance of \BCS, with $|V(G)|=N$.
Note that we can assume without loss of generality that $G$ contains no isolated vertices, since all such vertices can be deleted without changing the problem.

As our first step, we compute an arbitrary spanning forest $F$ of graph $G$. Since $G$ contains no isolated vertices, $|E(F)|\geq N/2$. Clearly, $\CrN(F)=0$. Consider now the optimal solution $H^*$ to instance $\pcn(G,L)$ of \BCS. If $|E(H^*)|<4N\cdot \log^6N$, then $F$ is a factor-$O(\log^6 N)$ approximate solution to instance $\pcn(G,L)$.

Assume now that $|E(H^*)|\ge 4N\log^6N$. Then, from \Cref{lem: WGP_BCS}, there exist integers 
$r,h>0$ with $r\cdot h^2\le L\cdot \log^6N$ and $r\cdot h\ge \Omega(|E(H^*)|/ \log N)$, such that there exists a good solution to instance $\pwgp(G,r,h)$ of \WGP. We will now attempt to guess such integers $r,h$, and then use the approximation algorithm for the \WGP problem, in order to compute a solution to the corresponding instance $\pwgp(G,r,h)$ of \WGP, whose value is sufficiently high.

We say that a pair $(r,h)$ of positive integers is \emph{eligible}, if $r\cdot h^2\le L\cdot\log^6 N$. 
For each eligible pair $(r,h)$ of integers, we consider the instance $\pwgp(G,r,h)$ of \WGP and we use the $\alpha(n)$-approximation algorithm for the \WGP problem to compute a solution $\hset_{r,h}$ to the instance $\pwgp(G,r,h)$, such that 
$\sum_{H'\in \hset_{r,h}}|E(H')|\ge \frac{\optwgp(G,r,h)}{\alpha(N)}$ (since $|V(G)|=N$). We use the following observation.

\begin{observation}
\label{obs: eligible}
If $\optcn(G,L)\ge 4N\log^6 N$, then there exists an eligible pair $(r,h)$ of integers, with: 
\[\sum_{H'\in \hset_{r,h}}|E(H')|\ge \Omega\bigg(\frac{\optcn(G,L)}{\alpha(N)\log N}\bigg).\]
\end{observation}
\begin{proof}
Let $H^*$ be an optimal solution to the instance $\pcn(G,L)$. From \Cref{lem: WGP_BCS}, there is an eligible pair $(r^*,h^*)$ of integers, with $r^*\cdot h^*\geq \Omega(E(H^*)|/\log N)$, so that there exists a good solution $\hset^*$ to instance $\pwgp(G,r^*,h^*)$ of \WGP. From the definition of a good solution, 
\[\sum_{H'\in \hset^*}|E(H')|\ge \frac{r^*h^*}{2}\ge \Omega\bigg(\frac{|E(H^*)|}{\log N}\bigg)
\ge \Omega\bigg(\frac{\optcn(G,L)}{\log N}\bigg).\]
Therefore, if $\hset_{r^*,h^*}$ is the approximate solution that we obtained for instance $\pwgp(G,r^*,h^*)$ of \WGP, then:

$$\sum_{H'\in \hset_{r^*,h^*}}|E(H')|\ge \Omega\bigg(\frac{\optwgp(G,r^*,h^*)}{\alpha(N)}\bigg)
\ge
\Omega\bigg(\frac{\sum_{H'\in \hset^*}|E(H')| }{\alpha(N)}\bigg)
\ge \Omega\bigg(\frac{\optcn(G,L)}{\alpha(N)\log N}\bigg),$$
and the observation follows.
\end{proof}

Let $(r',h')$ be the eligible pair of integers that maximizes $\sum_{H\in \hset_{r',h'}}|E(H)|$. 
For every graph $H\in \hset_{r',h'}$, let $\tilde H$ be a graph that is obtained from $H$ as follows. We set $V(\tilde H)=V(H)$, and we let $E(\tilde H)$ contain an arbitrary subset of $\floor{\frac{|E(H)|}{\log^3N}}$ edges of $E(H)$. Note that:

\[ \frac{|E(H)|}{2\log^3N}\leq  |E(\tilde H)|\leq \frac{h'}{\log^3N}.\]

Finally, we define a graph $H'=\bigcup_{H\in \hset_{r',h'}}\tilde H$. Since the crossing number of any $m$-edge graph is bounded by $m^2$, it is easy to verify that:

\[\CrN(H')\le \sum_{H\in \hset_{r',h'}}\CrN(\tilde H)\le r'\cdot\frac{(h')^2}{\log^6N}\le L.\]

Moreover:

\[|E(H')|=\sum_{H\in \hset_{r',h'}}|E(\tilde H)|\geq \sum_{H\in \hset_{r,h}}\frac{|E(H)|}{\log^3N}\geq \Omega\bigg(\frac{\optcn(G,L)}{\alpha(N)\log^4 N}\bigg). \]

Recall that we have computed a spanning forest $F$ of $G$. We return the graph in $\set{F,H'}$ that contains more edges as the outcome of the algorithm. From the above discussion, we obtain an $O(\alpha(N)\cdot\poly\log N)$-approximate solution.

\subsection{Reduction from \WGP to \BCS -- Proof of \Cref{thm: alg for CN gives alg for WGP}}
\label{subsec: from WGP to BCS}

In this subsection we prove \Cref{thm: alg for CN gives alg for WGP}.
Let $\pwgp(G,r,h)$ be the input instance of \WGP, and  denote $n=|V(G)|$. For convenience, we will assume that the value $C^*=\optwgp(G,r,h)$ is known to the algorithm: since $0\leq \optwgp(G,r,h)\leq |E(G)|$ holds, and $\optwgp(G,r,h)$ is an integer, we can try all possible guesses for the value $C^*$, and then output the best of the resulting solutions. It is sufficient to ensure that the algorithm correctly computes an $O((\alpha(n))^2\cdot\poly\log n)$-approximate solution to instance $\pwgp(G,r,h)$ if the guess $C^*$ is correct, that is, $C^*=\optwgp(G,r,h)$. From now on we assume that we are given a value $C^*=\optwgp(G,r,h)$.

We distinguish between two cases. The first case happens if $C^*\geq 16n\alpha(n)\log^7n$. In this case, we proceed as follows. By applying  \Cref{obs: equal_size} to the optimal solution to instance $\pwgp(G,r,h)$, we conclude that there are positive integers $r^*\le r, h^*\le h$ and a collectoin $\hset^*$ of subgraphs of $G$, such that $\hset^*$ is a good solution to instance $\pwgp(G,r^*,h^*)$, of \WGP, and $r^*h^*\geq \sum_{H\in \hset^*}|E(H)|\geq  \frac{C^*}{4\log n}$. Since the values of integers $r^*,h^*$ are not known to our algorithm, we will try all possible candidate values $1\leq r'\leq r$ and $1\leq h'\leq h$ with $r'h'\geq \frac{C^*}{4\log n}$.
 For each such pair $(r',h')$ of integers, we will compute a solution $\hset_{r',h'}$ to instance $\pwgp(G,r',h')$. We will then output the best solution from among $\set{\hset_{r',h'} \mid r'h'\geq \frac{C^*}{4\log n}}$. It is sufficient to ensure that, for integers $(r',h')=(r^*,h^*)$, the value of the resulting solution $\hset_{r',h'}$ is close to $C^*=\optwgp(G,r,h)$.

Consider now a pair of integers $1\leq r'\leq r$ and $1\leq h'\leq h$ with $r'h'\geq \frac{C^*}{4\log n}$, and assume that values $r',h'$ were guessed correctly, that is, $(r',h')=(r^*,h^*)$. In other words, there is a good solution $\hset^*$ to instance  $\pwgp(G,r^*,h^*)$, of \WGP, whose value is at least  $\frac{C^* }{4\log n}$. Consider now the graph $H'=\bigcup_{H\in \hset^*}H$. Since the crossing number of a graph $H$ may not be higher than $|E(H)|^2$, we get that $\CrN(H)\leq r'\cdot (h')^2$. Let $L=r'\cdot (h')^2$, and consider instance $\pcn(G,L)$ of the \BCS problem. From the above discussion, the value of the optimal solution to this problem is at least $|E(H)|\geq \frac{ \optwgp(G,r,h) }{4\log n}$. Therefore, by applying the $\alpha(N)$-approximation algorithm for \BCS to this instance, we obtain a solution $G'$ to instance  $\pcn(G,L)$ of \MBCS, whose value is at least  $\frac{ C^* }{4\alpha(n)\log n}$. Since we have assumed that $C^*\geq 16n\alpha(n)\log^7n$, we get that $|E(G')|\geq 4n\log^6n$. We can now use the algorithm from \Cref{lem: WGP_BCS} to compute integers $r'',h''>0$, such that $r''\cdot (h'')^2\leq L\cdot \log ^6n$, together with a good solution $\hset'$ to instance $\pwgp(G,r'',h'')$, whose value is at least $\Omega\left(\frac{|E(G')|}{\log n} \right )\geq \Omega \left( \frac{ C^* }{\alpha(n)\log^2 n}\right )$. Notice however that it is possible that $r''>r'$ or $h''>h'$ hold, so the solution that we obtain may not be a valid solution to instance $\pwgp(G,r',h')$ of \WGP. We show that, if $(r',h')=(r^*,h^*)$, then $h''$ cannot be much larger than $h'$. We then slightly modify solution $\hset'$, to transform it into a valid solution $\hset_{r',h'}$ to instance $\pwgp(G,r',h')$, while only decreasing the solution cost slightly. This completes the computation of the solution $\hset_{r',h'}$ associated with parameters $(r',h')$, and the algorithm for the first case.

Consider now the second case, where $C^*< 16n\alpha(n)\log^7n$. In this case, we start by computing a maximal subgraph $F$ of $G$, such that $F$ is a forest, with maximum vertex degree at most $h$. 
Let $S$ be the set of all vertices of $G$ that are adjacent to at least one edge of $F$, and denote $|S|=n'$.

We consider two subcases of Case 2. The first subcase happens if $|E(F)|\geq \frac{C^*}{32\alpha(n')\cdot \log^7n'}$. In this case, we show an algorithm that decomposes $F$ into $r$ subgraphs containing at most $h$ edges each, so that the total number of edges in all such subgraphs is close to $|E(F)|$. Therefore, we obtain a solution $\hset$ to instance $\pwgp(G,r,h)$, whose value is close to $C^*=\optwgp(G,r,h)$. Consider now the second subcase, where $|E(F)|<\frac{C^*}{32 \alpha(n')\cdot \log^7n'}$.  We show that in this case, there is a solution to instance $\pwgp(G[S],r,h)$ of \WGP, whose value is at least  $\frac{C^*}{2}$. From now on we only consider instane $\pwgp(G[S],r,h)$. We assume again that we are given the value $C^{**}$ of the optimal solution to this instance, where $\frac{C^*}2\leq C^{**}\leq C^*$. As before, this can be assumed since we can try all guesses for the value $C^{**}$, and it is sufficient to ensure that the algorithm works correctly if the value $C^{**}$ is guessed correctly. Recall that we have denoted $n'=|S|$. Since $(h-1)\cdot |S|\leq |E(F)|\leq \frac{C^*}{32 \alpha(n')\cdot \log^7n'}$, while $C^{**}\geq \frac{C^*}{2}$, we get that $C^{**}\geq 16n'\cdot \alpha(n')\cdot \log^7n'$.

We have now obtained a new instance $\pwgp(G[S],r,h)$ of $\WGP$, in which the value of the optimal solution $C^{**}\geq 16n\alpha(n)\log^7n'$, where $n'=|S|$. We can now repeat our algorithm for Case 1, to obtain the desired approximate solution to instance $\pwgp(G[S],r,h)$, which, in turn will provide an approximate solution to the original instance $\pwgp(G,r,h)$.

We now turn to the formal proof of \Cref{thm: alg for CN gives alg for WGP}. We assume that we are given an instance $\pwgp(G,r,h)$ of the \WGP problem, where $|V(G)|=n$, together with a guess $C^*$ on the value $\optwgp(G,r,h)$ of the optimal solution to this instance. Our goal is to compute a solution $\hset$ to instance $\pwgp(G,r,h)$ of \WGP, such that, if $C^*=\optwgp(G,r,h)$, then the value of the solution $\hset$ is at least $\Omega\left(\frac{C^*}{(\alpha(n))^2\poly\log n}\right )$. Note that we can assume that $n$ is greater than a sufficiently large constant, since otherwise we can solve the problem efficiently via exhaustive search.
We distinguish between two cases, depending on whether 
$C^*\geq 16n\alpha(n)\log^7n$ holds.

\subsubsection{Case 1: $C^*\geq 16n\alpha(n)\log^7n$}

Applying \Cref{obs: equal_size} to the optimal solution tp instance $\pwgp(G,r,h)$, we conclude that there are positive integers $r^*\le r, h^*\le h$ and a collection $\hset^*$ of subgraphs of $G$, such that $\hset^*$ is a good solution to instance $\pwgp(G,r^*,h^*)$, of \WGP. Moreover, if $C^*=\optwgp(G,r,h)$, then $r^*\cdot h^*\geq \sum_{H\in \hset^*}|E(H)|\geq  \frac{C^*}{4\log n}$. Our algorithm tries all possible values of integers $1\leq r'\leq r$ and $1\leq h'\leq h$ with $r'h'\geq \frac{C^*}{4\log n}$.
For each such pair $(r',h')$ of integers, we will compute a solution $\hset_{r',h'}$ to instance $\pwgp(G,r',h')$. At the end, our algorithm will output the best solution from among $\set{\hset_{r',h'} \mid r'h'\geq \frac{C^*}{4\log n}}$. It is sufficient to ensure that, for integers $(r',h')=(r^*,h^*)$, the value of the resulting solution $\hset_{r',h'}$ is  at least $\Omega\left(\frac{C^*}{(\alpha(n))^2\poly\log n}\right )$.

From now on we fix a pair $1\leq r'\leq r$, $1\leq h'\leq h$ of integers, with $r'h'\geq \frac{C^*}{4\log n}$.
Let $L=r'\cdot (h')^2$.
We apply the $\alpha(N)$-approximation algorithm for the \BCS problem to instance $\pcn(G,L)$, and obtain a solution that we denote by $H$. We use the following observation.

\begin{observation}\label{obs: small solution incorrect values}
	If $C^*=\optwgp(G,r,h)$, $r'=r^*$ and $h'=h^*$, then $|E(H)|\geq \frac{C^*}{4\alpha(n)\log n}$ must hold.
\end{observation}
\begin{proof}
	Assume that $C^*=\optwgp(G,r,h)$, $r'=r^*$ and $h'=h^*$. Recall that there exists a good solution $\hset^*$ to instance  $\pwgp(G,r^*,h^*)$, with $\sum_{H'\in \hset^*}|E(H')|\geq  \frac{C^*}{4\log n}$. Consider the graph $\hat H=\bigcup_{H'\in \hset^*}H'$. Since the crossing number of a graph $G'$ may not be higher than $|E(G')|^2$, we get that $\CrN(\hat H)\leq \sum_{H'\in \hset^*}\CrN(H')\leq r'\cdot (h')^2=L$. Therefore, $\hat H$ is a valid solution to instance $\pcn(G,L)$, whose value is at least  $\frac{C^*}{4\log n}$. Since we use an $\alpha(N)$-approximation algorithm for \BCS, and $|V(G)|=n$, we get that $|E(H)|\geq \frac{C^*}{4\alpha(n)\log n}$.
\end{proof}

If $|E(H)|<\frac{C^*}{4\alpha(n)\log n}$, then we terminate the algorithm and return an empty solution: in this case, we are guaranteed that either $C^*$, or $r',h'$ are guessed incorrectly. Therefore, we assume from now on that 
$|E(H)|\geq \frac{C^*}{4\alpha(n)\log n}$ holds. Note that, since in Case 1, $C^*\geq 16n\alpha(n)\log^7n$ holds, we are guaranteed that $|E(H)|\geq 4n\log^6n$.

Next, we apply the algorithm from \Cref{lem: WGP_BCS} to instance $\pcn(G,L)$ of \BCS, to compute integers $r'',h''>0$, such that $r''\cdot (h'')^2\leq L\cdot \log^6n=r'\cdot (h')^2\cdot \log^6n$, and $r''\cdot h''\geq \Omega\left(\frac{|E(H)|}{\log n}\right )\ge  \Omega\left(\frac{C^*}{\alpha(n)\log^2 n}\right )$. The algorithm also computes a good solution $\hset'$ to instance $\pwgp(G,r'',h'')$ of \WGP.
Note that, while the value of the solution $\hset'$ to instance $\pwgp(G,r'',h'')$  is guaranteed to be close to $C^*$, we are only guaranted that $\hset'$ is a valid solution to instance  $\pwgp(G,r'',h'')$ of the problem, and it may not be a valid solution to instance  $\pwgp(G,r',h')$. In our next steps, we will either correctly established that at least one of $C^*,r',h'$ was not guessed correctly; or we will slightly modify $\hset'$ to obtain a valid solution to intance  $\pwgp(G,r',h')$ of \WGP, whose value remains close to that of $\hset'$.
We start with the following observation.

\begin{observation}\label{obs: h not too large}
	There is a large enough constant $c'$, such that,	if $C^*=\optwgp(G,r,h)$, $r'=r^*$ and $h'=h^*$ hold, then $h''\leq c'h'\cdot \alpha(n)\cdot \log^8n$.
\end{observation}
\begin{proof}
Recall that we have established that:

\[ r''\cdot h''\geq \Omega\left(\frac{C^*}{\alpha(n)\log^2 n}\right ).  \]	

Assume for contradicton that $h''> c'h'\cdot \alpha(n)\cdot \log^8n$, where $c'$ is a large enough constant. Then:

\[r''\cdot (h'')^2\geq 4h'\cdot C^*\log^6n.\]

Notice that $C^*\geq r'h'/2$ must hold. Indeed, since $r'\leq r$ and $h'\leq h$, any solution to  instance $\pwgp(G,r',h')$ of \WGP is also a feasible solution to instance $\pwgp(G,r,h)$. Since we have assumed that $r'=r^*$ and $h'=h^*$, there is a good solution to instance $\pwgp(G,r',h')$ of \WGP, and the value of any such good solution is at least $r'h'/2$. Since we have assumed that $C^*=\optwgp(G,r,h)$, we get that $r'h'/2\leq C^*$ must hold.
We conclude that, if $h''> h'\cdot \alpha(n)\cdot \log^8n$, then:

\[r''\cdot (h'')^2\geq  2r'(h')^2\log^6n. \]

But we have already established above that $r''\cdot (h'')^2\leq r'\cdot (h')^2\cdot \log^6n$, a contradiction.
\end{proof}

We will now slightly modify the collection $\hset'$ of subgraphs of $G$ to obtain a feasible solution to instance $\pwgp(G,r,h)$, whose value is close to the value of $\hset'$. First, for every cluster $H\in \hset'$, if $|E(H)|>h'/2$, then we discard arbitrary edges from graph $H$, until $|E(H)|=h'/2$ holds. From \Cref{obs: h not too large}, if $C^*=\optwgp(G,r,h)$, $r'=r^*$ and $h'=h^*$ hold, then the total number of edges in the graphs of $\hset'$ decreases by at most factor $O(\alpha(n)\cdot \log^8n)$ as the result of this transformation, and so $\sum_{H\in \hset'}|E(H)|\geq \Omega\left(\frac{r''\cdot h''}{\alpha(n)\cdot \log^8n}\right )\geq  \Omega\left(\frac{C^*}{(\alpha(n))^2\log^{10} n}\right )$ holds. Also, if, at the end of this transformation, $\sum_{H\in \hset'}|E(H)|>\frac{h'\cdot r'}{2}$ holds, then we discard arbitrary edges from the graphs in $\hset'$ until $\sum_{H\in \hset'}|E(H)|\leq \frac{h'\cdot r'}{2}$ holds.  Since $h'r'\geq\frac{C^*}{4\log n}$, $\sum_{H\in \hset'}|E(H)|\geq \Omega\left(\frac{C^*}{(\alpha(n))^2\log^{10} n}\right )$ continues to hold.

If $|\hset'|\leq r'$, then we have obtained a valid solution to instance $\pwgp(G,r',h')$ of value $\Omega\left(\frac{C^*}{(\alpha(n))^2\log^{10} n}\right )$. Otherwise, we perform further transformations to the set $\hset'$ of graphs as follows.

While $|\hset'|>r'$, we let $H',H''\in \hset'$ be a pair of graphs with smallest number of edges, breaking ties arbitrarily. We remove $H'$ and $H''$ from $\hset'$, and we add a new graph $H=H'\cup H''$ to $\hset'$ instead. The procedure is terminated once $|\hset'|=r'$ holds. We claim that at the end of this procedure, for every graph $H\in \hset'$, $|E(H)|\leq h'$ holds. Indeed, assume otherwise. Consider the first time when a graph $H$ with $|E(H)|>h'$ was added to $\hset'$. Then $H=H'\cup H''$ must hold, where $H',H''$ are two graphs that belonged to $\hset'$ prior to this iterations. Then at least one of these two graphs must contain more than $h'/2$ edges. From the choice of the graphs $H',H''$, and from the fact that $|\hset'|>r'$ held at the beginning of the iteration, we get that, at the beginning of the iteration, there were at least $r'$ graphs $\tilde H\in \hset'$ with $|E(\tilde H)|>h'/2$. But then $\sum_{\tilde H\in \hset'}|E(\tilde H)|>\frac{r'h'}{2}$ held at the beginning of the iteration. Since the total number of edges contained in the graphs of $\hset'$ does not change over the course of the algorithm, we reach a contradiction, since we have ensured that, at the beginning of the algorithm, $\sum_{\tilde H\in \hset'}|E(\tilde H)|\leq \frac{h'\cdot r'}{2}$ held. We return the resulting collection $\hset'$ of subgraphs of $G$, which is guaranteed to be a feasible solution to instance $\pwgp(G,r',h')$, of value at least $\Omega\left(\frac{C^*}{(\alpha(n))^2\log^{10} n}\right )$.

\subsubsection{Case 2: $C^*< 16n\alpha(n)\log^7n$}

In this case, we start by computing a maximal subgraph $F$ of $G$, such that $F$ is a forest, and maximum vertex degree in $F$ is at most $h$. Such a graph $F$ can be computed via a simple greedy algorithm. We start with graph $F$ containing the set $V(F)=V(G)$ of vertices and no edges. We then consider the edges of $G$ one by one. For each such edge $e\in E(G)$, if graph $F\cup \set{e}$ remains a forest with maximum vertex degree at most $h$, then we add $e$ to $F$. Once every edge of $G$ is processed, we obtain the final graph $F$. Let $S$ be the set of all vertices of $G$ that are adjacent to at least one edge of $F$, and denote $|S|=n'$.

We consider two subcases of Case 2. The first subcase, Case 2a happens if $|E(F)|\geq \frac{C^*}{64\alpha(n')\cdot \log^7n'}$. In this case, we use the following simple observation, that will allow us to decompose the forest $F$ to obtain a solution to instance $\pwgp(G,r',h')$ of $\WGP$, whose value is close to $|E(F)|$.

\begin{observation}
	\label{obs: cutting a tree}
	There is an efficient algorithm, that, given a tree $T$ with maximum vertex degree at most $\Delta$ and $|E(T)|\ge \Delta/2$, computes a collection $\tset$ of vertex-disjoint subgraphs of $T$, such that for each subgraph $T'\in \tset$, $\frac{\Delta}2\le |E(T')|\le \Delta$, and  $\sum_{T'\in \tset}|E(T')|\ge \frac{|E(T)|} 2$.
\end{observation}
\begin{proof}
We root the tree $T$ in an arbitrary vertex $v$. Initially, we let $\tset=\emptyset$. As long as $|E(T)|>\Delta$, we perform iterations. In every iteration, we consider an arbitrary vertex $u$ in the current tree $T$, that is a non-leaf vertex, but all children of $u$ are leaf vertices. Let $T'$ be the subtree of $T$ rooted at vertex $u$. Note that $1\le |E(T')|\leq \Delta-1$. We add graph $T'$ to $\tset$, and we delete all vertices of $T'$ from $T$. Note that, as the result of this iteration, the unique edge $e$ connecting $u$ to its parent-vertex in the tree $T$ is deleted from $T$, and it does not belong to any graph in $\tset$. We let $e'\in E(T)$ be an arbitrary edge (which must exist since $|E(T)|\geq 1$), and we say that $e'$ is \emph{responsible} for the deletion of the edge $e$. We then continue to the next iteration. The algorithm terminates once $|E(T)|\leq \Delta$ holds. We then add graph $T$ to the collection $\tset$ and terminate the algorithm. It is immediate to verify that, for every graph $T'\in \tset$, $|E(T')|\leq \Delta$ holds. Moreover, if an edge $e$ belonged to the original graph $T$, and it does not belong to $\bigcup_{T'\in \tset}E(T')$, then some edge of $\bigcup_{T'\in \tset}E(T')$ is designated as being responsible for deleting $e$. It is easy to verify that every edge $e'\in \bigcup_{T'\in \tset}E(T')$ may be responsible for the deletion of at most one edge. Therefore, at the end of the algorithm, $\sum_{T'\in \tset}|E(T')|\geq |E(T)|/2$ holds.
\end{proof}

Initially, we construct a collection $\hset'$ of subgraphs of $F$ as follows. For every tree $T$ of the forest $F$, if $|V(T)|\leq h$, then we add $T$ to $\hset'$. Otherwise, we apply  
\Cref{obs: cutting a tree} to tree $T$ with parameter $\Delta=h$, and add the graphs in the resulting collection $\tset$ to $\hset'$. At the end of this algorithm, for every graph $H'\in \hset'$, $|E(H')|\leq h$ holds, and $\sum_{H'\in \hset'}|E(H')|\geq\frac{|E(F)|}{2}\geq \frac{C^*}{128\alpha(n')\cdot \log^7n'}$. For every graph $H'\in \hset'$, if $|E(H')|>h/2$, then we delete edges from $H'$ until $|E(H')|\leq h/2$ holds. Clearly, after this transformation, $\sum_{H'\in \hset'}|E(H')|\geq \frac{C^*}{256\alpha(n')\cdot \log^7n'}\geq \Omega\left(\frac{C^*}{\alpha(n)\cdot \log^7n}\right )$ holds. If $\sum_{H\in \hset'}|E(H)|>hr/2$, then we discard arbitrary edges from the graphs in $\hset'$, until $\sum_{H\in \hset'}|E(H)|=hr/2$ holds. Since, if $C^*=\optwgp(G,r,h)$, $C^*\leq rh$, we are still guaranteed that $\sum_{H'\in \hset'}|E(H')|\geq\Omega\left(\frac{C^*}{\alpha(n)\cdot \log^7n}\right )$ holds.

If $|\hset'|\leq r$, then we have obtained a valid solution to instance $\pwgp(G,r,h)$ of value $\Omega\left(\frac{C^*}{\alpha(n)\cdot \log^7n}\right )$. Otherwise, 
we proceed exactly like in Case 1 in order to transform $\hset'$ into a valid solution to instance $\pwgp(G,r,h)$ of $\WGP$, without changing the total number of edges that lie in the graphs of $\hset'$.
While $|\hset'|>r$, we let $H',H''\in \hset'$ be a pair of graphs with smallest number of edges, breaking ties arbitrarily. We remove $H'$ and $H''$ from $\hset'$, and we add a new graph $H=H'\cup H''$ to $\hset'$ instead. The procedure is terminated once $|\hset'|=r$ holds. We claim that at the end of this procedure, for every graph $H\in \hset'$, $|E(H)|\leq h$ holds. Indeed, assume otherwise. Consider the first time when a graph $H$ with $|E(H)|>h$ was added to $\hset'$. Then $H=H'\cup H''$ must hold, where $H',H''$ are two graphs that belonged to $\hset'$ prior to this iterations. Then at least one of these two graphs must contain more than $h/2$ edges. From the choice of the graphs $H',H''$, and from the fact that $|\hset'|>r$ held at the beginning of the iteration, we get that, at the beginning of the iteration, there were at least $r$ graphs $\tilde H\in \hset'$ with $|E(\tilde H)|>h/2$. But then $\sum_{\tilde H\in \hset'}|E(\tilde H)|>\frac{r h}{2}$ held at the beginning of the iteration. Since the total number of edges contained in the graphs of $\hset'$ does not change over the course of the algorithm, we reach a contradiction, since we have ensured that, at the beginning of the algorithm, $\sum_{\tilde H\in \hset'}|E(\tilde H)|\leq \frac{h\cdot r}{2}$ held. We return the resulting set $\hset'$ of subgraphs of $G$, which is guaranteed to be a feasible solution to instance $\pwgp(G,r,h)$ of \WGP, of value at least $\Omega\left(\frac{C^*}{\alpha(n)\log^{7} n}\right )$.

It now remains to consider Case (2b), where $C^*< 16n\alpha(n)\log^7n$ and $|E(F)|< \frac{C^*}{64\alpha(n')\cdot \log^7n'}$. Recall that $S$ is the set of all vertices of $G$ that are adjacent to at least one edge of $F$, and recall that we have denoted $|S|=n'$.

In this case, we let $G'=G[S]$, and we consider instance $\pwgp(G',r,h)$ of \WGP. Notice that, if $\hset'$ is a valid solution to instance $\pwgp(G',r,h)$ of \WGP, then it is also a valid solution to instance $\pwgp(G,r,h)$ of \WGP. We start by showing that $\optwgp(G',r,h)$ is close to $C^*$.

\begin{observation}\label{obs: case 2b}
If $C^*=\optwgp(G,r,h)$, and Case (2b) happens, then $\optwgp(G',r,h)\geq \frac{C^*}{2}$.
\end{observation}
\begin{proof}
	Let $S'\subseteq S$ be the set of all vertices whose degree in $F$ is $h$, and let $S^*$ the set of all vertices of $F$ that are isolated.
	
	Let $\hset$ be the optimal solution to instance $\pwgp(G,r,h)$, and let $E'=\bigcup_{H\in \hset}E(H)$. We partition the set $E'$ of edges into two subsets: set $E'_1$ containing all edges that lie in $G'=G[S]$, and set $E'_2$ containing all remaining edges. Clearly, for every edge $e=(x,y)\in E'_2$, at least one endpoint of $e$ must lie in $S^*$. Assume w.l.o.g. that $x\in S^*$. We claim that $y\in S'$ must hold. Indeed, otherwise $F\cup\set{e}$ remains a forest, with maximum vertex degree at most $h$, contradicting the fact that $F$ is a maximal subgraph of $G$ with these properties.
	
	Therefore, very edge of $E'_2$ connects a vertex of $S^*$ to a vertex of $S'$. We claim that $|E'_2|\leq |S'|\cdot h$. Indeed, from the definition of the \WGP problem, for every graph $H\in \hset$, $|E(H)|\leq h$, and all graphs in $\hset$ are disjoint in their vertices. Therefore, every vertex of $G$ may be incident to at most $h$ edges of $E'$. Since every edge of $E'_2$ has a vertex of $S'$ as its endpoint, we get that $|E'_2|\leq |S'|\cdot h$.
	
	Recall that, from our definition, every vertex $v\in S'$ has degree $h$ in $F$. Therefore, $|E(F)|\geq (h-1)\cdot |S'|$. We conclude that $|E(F)|\geq |E'_2|/2$, and so $|E'_2|\leq 2|E(F)|<\frac{C^*}{4}$. Since $|E'|=C^*$, we get that $|E'_1|\geq C^*/2$. 
	
	We now define a solution $\hset'$ to instance $\pwgp(G',h,r)$ of \WGP. For every graph $H\in \hset$, we let $H'$ be a graph that is obtained from $H$ by deleting all vertices of $S^*$ from it, and we let $\hset'=\set{H'\mid H\in \hset}$. It is easy to verify that $\hset'$ is a valid solution to instance $\pwgp(G',r,h)$, and that its value is at least $|E'\setminus E'_2|\geq \frac{C^*}{2}$. We conclude that $\optwgp(G',r,h)\geq \frac{C^*}{2}$.
\end{proof}

Denote $C'=\optwgp(G',r,h)$. From the above discussion $\frac{C^*}2\leq C'\leq C^*$. Notice that for every tree $T$ of $F$, if $T$ is not a singleton vertex, then $|E(T)|\geq |V(T)|-1\geq \frac{|V(T)|}{2}$. Therefore, $|E(F)|\geq \frac{|S|}{2}=\frac{n'}2$. On the other hand, from our assumpution, $|E(F)|< \frac{C^*}{64\alpha(n')\cdot \log^7n'}$. We then conclude that $C^*>32n'\alpha(n')\cdot \log^7n$.

We will now focus on solving instance $\pwgp(G',r,h)$ of the \WGP problem. As before, we will try all guesses $C^{**}$ on the value $C'$ of the optimal solution for this problem. Note that we only need to consider values $C^{**}$ that are integers, with $\frac{C^*}2\leq C^{**}\leq C^*$. Furthermore, from the above discussion, for each such guess, $C^{**}\ge \frac{C^*}{2}\geq 16n'\alpha(n')\cdot \log^7n$ holds, so Case 1 will occur. We execute the algorithm from Case 1 for each such guessed value $C^{**}$ and output the best among the resulting solutions. We are then guaranteed to obtain a solution $\hset$ to instance $\pwgp(G',r,h)$ of value at least $\Omega\left(\frac{C'}{(\alpha(n))^2\log^{10} n}\right )\geq \Omega\left(\frac{C^*}{(\alpha(n))^2\log^{10} n}\right )$. Clearly, $\hset$ is also a valid solution to instance $\pwgp(G,r,h)$. Overall, we obtain an efficient $O((\alpha(n))^2\cdot\poly\log n)$-approximation algorithm for \WGP.

\section{Acknowledgement}
The authors thank Irit Dinur and Uri Feige for insightful and helpful discussions.

\newpage

\appendix

\section{Proof of \Cref{lem: DkS and Dk1k2S}}
\label{apd: Proof of DkS and Dk1k2S}

We prove each of the directions of the reductions separately, in the following two subsections.

\subsection{Reduction from \DBS to \DkS}

Assume that exists an $\alpha(n)$-approximation algorithm $\aset$ for the \DkS problem with running time at most $T(n)$, where $n$ is the number of vertices in the input graph. We show an  $O(\alpha(N^2))$-approximation algorithm for the \DBS problem, whose running time is at most $O(T(N^2)\cdot \poly(N))$, where $N$ is the number of vertices in the input graph.

	Let $\pdks(G,k_1,k_2)$ be the input instance to the \DBS problem. Denote $G=(A,B,E)$, so $|A\cup B|=N$. We construct another bipartite graph $H=(A',B',E')$, that will serve as input to the \DkS problem, as follows. We define, for every vertex $u\in A$, a collection $T_u=\set{u^{1},\ldots, u^{k_2}}$ of vertices that we call \emph{copies of $u$}, and we let $A'=\bigcup_{u\in A}T_u$. Similarly, we define, for every vertex $v\in B$, a set $T_v=\set{v^{1},\ldots, v^{k_1}}$ of $k_1$ vertices, that we call \emph{copies of $v$}, and we let $B'=\bigcup_{v\in B}T_v$. 
	 The set $E'$ of edges of $H$ contains, for every edge $e=(u,v)\in E(G)$, all edges in $T_u\times T_v$. Note that $|V(H)|\leq \max\set{k_1,k_2}\cdot |A\cup B|\le N^2$.
	 
	 Let $k=2k_1k_2$, and 
consider the instance $\pdks(H,k)$ of the \DkS problem. We use the following observation to lower-bound its optimal solution cost.
	
	\begin{observation}\label{obs: opt for dks}
	$\optdks(H,k)\ge k_1k_2\cdot \optbdks(G,k_1,k_2)$.
\end{observation}
\begin{proof}
	Let $S^*$ be the optimal solution to instance $\pbdks(G,k_1,k_2)$ of \DBS.
	Denote $S^*_A=S^*\cap A$ and $S^*_B=S^*\cap B$, so $|S^*_A|=k_1$ and $|S^*_B|=k_2$ hold. We define $T^*_A=\bigcup_{u\in S^*_A}T_u$ and $T^*_B=\bigcup_{v\in S^*_B}T_v$. From the construction of $H$, it is clear that $|T^*_A|=|T^*_B|=k_1\cdot k_2$, and $|E_H(T^*_A,T^*_B)|=k_1k_2\cdot|E_G(S^*_A,S^*_B)|$.
	Therefore, $T^*_A\cup T^*_B$ is a feasible solution to instance $\pdks(H,k)$ of \DkS, and so $\optdks(H,2k_1k_2)\ge k_1k_2\cdot \optbdks(G,k_1,k_2)$.
\end{proof}

In order to complete the reduction, we need the following claim.

\begin{claim}
\label{clm: DKS to DBS}
There is an efficient algorithm, that, given any solution $W$ to the instance $\pdks(H,k)$ of the \DkS problem, computes a solution to instance $\pbdks(G,k_1,k_2)$ of \BDkS, whose value is at least $|E_H(W)|/(4k_1k_2)$.
\end{claim}

\begin{proof}
Denote $W_A=W\cap A$ and $W_B=W\cap B$.
We start by computing a partition $(W^0_A,\ldots,W^{2k_2-1}_A)$ of the vertices of $W_A$ into $2k_2$ subsets, containing at most $k_1$ vertices each, so that, for every vertex $u$ of $A$, no two copies of $u$ appear in the same subset.

In order to do so, we let $\sigma$ be an arbitrary ordering of the vertices of $W_A$, in which, for every vertex $u\in A$, all copies of $u$ that belong to $W_A$ appear consecutively. For all $0\le i\le 2k_2-1$, we let $W^i_A\subseteq W_A$ to be the set of all vertices $x\in W_A$, whose index is $i \text{ }(\text{mod }2k_2)$ in this ordering.
Since, for every vertex $u\in A$, $|T_u|=k_2$, it is immediate to verify that all copies of $u$ in $W_A$ lie in distinct sets. It is also immediate to verify that, for all $0\leq i<2k_2$, $|W^i_A|\leq k_1$.

We similarly compute a partition $(W^0_B,\ldots,W^{2k_1-1}_B)$ of the  vertices of $W_B$ into $2k_1$ subsets, containing at most $k_2$ vertices each, so that, for every vertex $v$ of $B$, no two copies of $v$ appear in the same subset. 

%
%

Let $0\leq i^*<2k_2,0\leq j^*<2k_1$ be a pair of indices, for which $|E_H(W^{i^*}_A,W^{j^*}_B)|$ is maximized. Clearly, $|E_H(W^{i^*}_A,W^{j^*}_B)|\ge |E_H(W)|/(4k_1k_2)$.
Finally, let $X\subseteq V(G)$ be the set of vertices containing every vertex $u\in V(G)$, whose copy lies in  $W^{i^*}_A\cup W^{j^*}_B$. Note that $|X\cap A|= |W^{i^*}_A|\le k_1$ and $|X\cap B|= |W^{j^*}_B|\le k_2$ must hold, so $X$ is a valid solution to instance $\pbdks(G,k_1,k_2)$ of \BDkS. Since $W^{i^*}_A\cup W^{j^*}_B$ contains at most one copy of every vertex of $V(G)$, from the construction of graph $H$, it is easy to verify that $|E_G(X)|=|E_H(W^{i^*}_A,W^{j^*}_B)|\ge |E_H(W)|/(4k_1k_2)$.
\end{proof}	

We are now ready to complete our reduction. 
We apply the approximation algorithm $\aset$ for the \DkS problem to instance $\pdks(H,k)$, to obtain a solution $W$. 
Since $|V(H)|\leq N^2$, and since $\aset$ is a factor-$\alpha(n)$ approximation algorithm, from \Cref{obs: opt for dks}, we get that:

\[|E_H(W)|\geq \frac{\optdks(H,k)}{\alpha(N^2)}\geq \frac{k_1\cdot k_2\cdot \optbdks(G,k_1,k_2)}{\alpha(N^2)}.\]

Additionally, the running time of the algorithm is $O(T(N^2))$.

We then apply the algorithm from \Cref{clm: DKS to DBS}, whose running time is bounded by $O(\poly(N))$ to solution $W$ to instance  $\pdks(H,k)$, to obtain a solution $X$ to instance $\pdks(G,k_1, k_2)$ of \BDkS. We are guaranteed that:

$$
|E_G(X)|\ge \frac{|E_H(W)|}{4k_1k_2}\ge \Omega\bigg(\frac{\optdks(G,k_1,k_2)}{\alpha(N^2)}\bigg).$$

It is easy to verify that the running time of the algorithm is bounded by $O(T(N^2)\cdot \poly(N))$.

\subsection{Reduction from \DkS to \DBS}

We now assume that there exists an efficient $\alpha(N)$-approximation algorithm $\aset'$ for the \DBS problem, where $N$ is the number of vertices in the input graph. We show that there exists an efficient $O(\alpha(2n))$-approximation algorithm for the \DkS problem, where $n$ is the number of vertices in the input graph.

Let $\pdks(G,k)$ be the input instance for the \DkS problem, so $|V(G)|=n$. We construct a bipartite graph $H=(V_1,V_2,E)$, where the vertex sets are $V_1=\set{u^1\mid u\in V(G)}$, $V_2=\set{u^2\mid u\in V(G)}$, and the edge set is $E=\set{(u^1,v^2),(u^2,v^1)\mid (u,v)\in E(G)}$. We denote $N=|V(H)|=2n$.
Consider the instance $\pdks(H,k_1,k_2)$ of \DBS, where $k_1=k_2=k$.
We use the following observation to lower-bound the optimal solution cost of this instance.

\begin{observation}\label{obs: lower bound opt 2}
	$\optbdks(H,k_1,k_2)\ge 2\cdot \optdks(G,k)$.
\end{observation}
\begin{proof}
Let $V^*$ be the optimal solution to instance $\pdks(G,k)$ of $\DkS$. We define $U^*=\set{v^1,v^2\mid v\in V^*}$, so $|U^*\cap V_1|=k$ and $|U^*\cap V_2|=k$. Clearly, $U^*$ is a feasible solution to instance $\pbdks(H,k,k)$. Moreover, it is easy to verify that $|E_H(U^*)|=2\cdot |E_G(V^*)|$, so $\optbdks(H,k_1,k_2)\ge 2\cdot \optdks(G,k)$.
\end{proof}

 We apply Algorithm $\aset'$ to instance, $\pdks(H,k_1,k_2)$ of \DBS, obtaining a solution $U'$. We denote $U_1=U'\cap V_1$ and $U_2=U'\cap V_2$, so $|U_1|=|U_2|=k$, and $|E_H(U_1,U_2)|\ge \optbdks(H,k,k)/\alpha(N)\geq 2\optdks(G,k)/\alpha(2n)$ from \Cref{obs: lower bound opt 2}.

Let $U=\set{v\mid v^1\in U_1 \text{ or }v^2\in U_2}$ be a subset of vertices of $G$. Clearly, $|U|\le 2k$, and $|E_G(U)|\geq \frac{|E_H(U^*)|}2\geq \frac{\optdks(G,k)}{\alpha(2n)}$.

We then apply the algorithm from \Cref{lem: size reducing} to graph $G$ and set $U$ of vertices, with parameter $\beta=1/2$, to obtain a set $\tilde U\subseteq U$ of vertices, with $|\tilde U|\leq k$, and $|E_G(\tilde U)|\geq \Omega(|E_G(U)|)\geq \Omega(\optdks(G,k)/\alpha(2n))$. Therefore, we obtained an $O(\alpha(2n))$-approximate solution to instance $\pdks(G,k)$ of the \DkS problem.

\section{Reduction from \WGP to \DkS}
\label{apd: WGP to DkS}

In this section we complete the  proof \Cref{thm: alg_DkS gives alg_DkC}, by showing a reduction from \WGP to \DkS. The reduction is very similar to the reduction from \DkC to \DkS described in \Cref{sec: DkC and WGP to DkS}. 

 We start by formulating an LP-relaxation of the problem, whose number of constraints is bounded by $O(N)$, but the number of variables may be large. We then show an LP-rounding algorithm for this LP-relaxation, whose running time is $O(\poly(N))$ if it is given a solution to the LP-relaxation whose support size is bounded by $O(\poly(N))$. In order to compute an approximate LP-solution whose support size is sufficiently small, we design an approximate separation oracle for the dual of the LP-relaxation. We start with describing the LP-relaxation and providing an LP-rounding algorithm for it.

\subsection{Linear Programming Relaxation and an LP-Rounding Algortihm}

Let $\pwgp(G,r,h)$ be the input instance of \WGP, and denote $|V(G)|=N$.
We let $\hset$ be the collection of all subgraphs $H\subseteq G$ with $|E(H)|\leq h$. For each such subgraph $H$, we denote $m(H)=|E(H)|$.
We consider the following LP-relaxation of the \WGP problem, that has a variable $x_H$ for every graph $H\in \hset$.

\begin{eqnarray*}
	\mbox{(LPW-P)}&&\\
		\max &\sum_{H\in \hset} m(H)\cdot x_H&\\
	\mbox{s.t.}&&\\
	&\sum_{\stackrel{H\in \hset:}{v\in V(H)}} x_H\leq 1 &\forall v\in V(G)\\
	&\sum_{H\in \hset}x_H \leq r\\
	&x_H\geq 0&\forall H\in \hset
\end{eqnarray*}

It is easy to verify that (LPW-P) is an LP-relaxation of the \WGP problem. Indeed, consider a solution $(H_1,\ldots,H_r)$ to the input instance $\pwgp(G,r,h)$. For all $1\leq i\leq r$, we set $x_{H_i}=1$, and for every other graph $H\in \hset$, we set $x_H=0$. This provides a feasible solution to (LPW-P), whose value is precisely $\sum_{i=1}^r|E(H)|$. We denote the value of the optimal solution to (LPW-P) by $\opt_{\textnormal{LP-P}}$. From the above discussion, $\opt_{\textnormal{LP-P}}\ge \optwgp(G,r,h)$. 

In the following claim we provide an LP-rounding algorithm for (LPW-P). The claim is an analogue of \Cref{claim: LP-rounding}. Its proof is almost identical and is provided here for completeness.

\begin{claim}\label{claim: LP-rounding WGP}
	There is an efficient randomized algorithm, whose input consists of an instance $\pwgp(G,r,h)$ of the \WGP problem with $N=|V(G)|$, such that $N$ is greater than a large enough constant, and a solution $\set{x_H\mid H\in \hset}$ to (LPW-P), in which the number of variables $x_H$ with $x_H>0$ is bounded by $O(\poly(N))$, and $\sum_{H\in \hset}m(H)\cdot x_H\geq \opt_{\textnormal{LP-P}}/\beta$, for some parameter $1\leq \beta\leq N^3$; the solution is given by only specifying values of variables $x_H$ that are non-zero. The algorithm with high probability returns an integral solution $(H_1,\ldots,H_r)$ to instance $\pwgp(G,r,h)$, such that $\sum_{i=1}^{r}|E(H_i)|\geq \frac{\optwgp(G,r,h)}{2000\beta \log^3N}$.
\end{claim}

\begin{proof}
We assume that we are given a solution  $\set{x_H\mid H\in \hset}$ to (LPW-P), in which the number of variables $x_H$ with $x_H>0$ is bounded by $O(\poly(N))$. Denote $C=\sum_{H\in \hset}m(H)\cdot x_H$, and recall that $C\geq \opt_{\textnormal{LP-P}}/\beta$ holds. We denote by $\hset'\subseteq \hset$ the collection of all graphs $H\in \hset$ with $x_H>0$. 

We construct another collection $\hset''\subseteq \hset'$ of subgraphs of $G$ as follows. For every subgraph $H\in \hset'$, we add $H$ to $\hset''$ independently, with probability $x_H$. Clearly, $\expect{\sum_{H\in \hset''}m(H)}=\sum_{H\in \hset}m(H)\cdot x_H=C$.
	
We say that a bad event $\event_1$ happens if some vertex $v\in V(G)$ lies in more than $5\log N$ graphs in $\hset''$. We say that a bad event $\event_2$ happens if $|\hset''|> 5r\log N$. We say that a bad event $\event_3$ happens if $\sum_{H\in \hset''}m(H)<\frac{C}{8}$.
Lastly, we say that a bad event $\event$ happens if either of the events $\event_1, \event_2$, or $\event_3$ happen.
The following observation is an analogue of \Cref{obs: first bad event}. Its proof is identical and is omitted here.

\begin{observation}\label{obs: first bad event WGP}
	$\prob{\event}\leq 2/N^3$.
\end{observation}

Observe that we can efficiently check whether Event $\event$ happened. If Event $\event$ happens, then we terminate the algorithm with a FAIL. We assume from now on that Event $\event$ did not happen. Then $\sum_{H\in \hset''}m(H)\geq \frac{C}{8}\geq \frac{\optwgp(G,r,h)}{8\beta}$ must hold. We denote $\hset''=\set{H_1,H_2,\ldots,H_z}$, where the graphs are indexed according to their value $m(H)$, so that $m(H_1)\geq m(H_2)\geq \cdots\geq m(H_z)$. We then let $\hset^*=\set{H_1,\ldots,H_{r}}$ (if $z<r$, then  we set $H_{z+1}=\cdots=H_r=\emptyset$). For all $1\leq i\leq r$, we denote $E_i=E(H_i)$, so $|E_i|=m(H_i)$. Recall that, since Event $\event$ did not happen, $|\hset''|\leq 5r\log N$ holds. Therefore:

\[ \sum_{i=1}^r|E_i|\geq \frac{\sum_{H\in \hset''}m(H)}{5\log N}\geq \frac{\optwgp(G,r,h)}{40\beta\log N}. \]

As before, the graphs in set $\hset$ may not be mutually disjoint. However, since Event $\event$ did not happen, every vertex of $V(G)$ may lie in at most $5\log N$ such graphs. We now construct a new collection $\hset^{**}=\set{H_1',\ldots,H_r'}$ of graphs, as follows. For all $1\leq i\leq r$, we will define a subset $V_i\subseteq V(H_i)$ of vertices, and we will then set $H_i'=H_i[V_i]$. 
In order to define vertex sets $V_1,\ldots,V_r$, we start by setting $V_1=V_2=\cdots=V_r=\emptyset$, and then process vertices $v\in V(H)$ one by one.
Consider any vertex $v\in V(G)$, and let $H_{i_1},H_{i_2},\ldots,H_{i_a}\in \hset^*$ be the graphs of $\hset^*$ containing $v$. Vertex $v$ chooses an index $i^*\in \set{i_1,\ldots,i_a}$ at random, and is then added to $V_{i^*}$.
Once all vertices of $V(G)$ are processed, we obtain a final collection $V_1,\ldots,V_r$ of sets of vertices, where for all $1\leq i\leq r$, $V_i\subseteq V[H_i]$. For all $1\leq i\leq r$, we then set $H'_i=H_i[V_i]$. Since $H'_i\subseteq H_i$, we are then guaranteed that $|E(H'_i)|\leq h$ holds.

Note that for all $1\leq j\leq r$, for every vertex $v\in V(H_j)$, the probability that $v\in V_j$ is at least $1/(5\log N)$. We say that an edge $e=(u,v)\in E_j$ \emph{survives} if both $u,v\in V_j$. We denote by $E''\subseteq \bigcup_{i=1}^rE_i$ the set of all edges that survive. Since $\prob{u\in V_j}\geq 1/(5\log N)$, $\prob{v\in V_j}\geq 1/(5\log N)$, and the two events are independent, we get that the probability that edge $e$ survives is at least $1/(25\log^2N)$. Overall, we get that:

\[\expect{|E''|}\geq \frac{\sum_{i=1}^r|E_i|}{25\log^2N}\geq \frac{\optwgp(G,r,h)}{1000\beta\log^3 N} .\]

The final solution to instance $\pwgp(G,r,h)$ is $\hset^{**}=\set{H_1',\ldots,H_r'}$. Clearly, the value of this solution is $|E''|$.

So far we have obtained a randomized algorithm that either returns FAIL (with probability at most $2/N^3$), or it returns a solution to instance instance $\pwgp(G,r,h)$ of the \WGP problem, whose expected value is at least $\frac{\optwgp(G,r,h)}{1000\beta\log^3 N}$.

Let $p'$ be the probability that the algorithm returned a solution of value at least   $\frac{\optwgp(G,r,h)}{2000\beta\log^3 N}$, given that it did not return FAIL. Note that the expected solution value, assuming the algorithm did not return FAIL, is at most $\frac{\optwgp(G,r,h)}{2000\beta\log^3 N}+p'\cdot \optwgp(G,r,h)$. Since this expectation is also at least $\frac{\optwgp(G,r,h)}{1000\beta\log^3 N}$, we get that $p'\geq \frac{1}{1000\beta\log^3N}$. Overall, the probability that our algorithm successfully returns a solution of value at least $\frac{\optwgp(G,r,h)}{2000\beta\log^3 N}$ is $p'\cdot \prob{\neg\event}\geq \Omega\left(\frac{1}{\beta\log^3N}\right )$. By repeating the algorithm $\poly(N)$ times we can ensure that it successfully computes a solution of value at least $\frac{\optwgp(G,r,h)}{2000\beta\log^3 N}$ with high probability.
\end{proof}

\subsection{Approximately Solving the LP-Relaxation}

In this subsection we provide an approximate separation oracle for the dual linear program of (LPW-P). This is sufficient in order to obtain an algorithm with running time $O(\poly(N))$ that approximately solves (LPW-P) using the methods described in \Cref{subsec: solve the LP}. The following Linear Program is a Dual of (LPW-P). It has a variable $y_v$ for every vertex $v\in V(G)$, and an additional variable $z$.

\begin{eqnarray*}
	\mbox{(LPW-D)}&&\\
	\min& r\cdot z+\sum_{v\in V(G)} y_v\\
	\mbox{s.t.}&&\\
	& z+\sum_{v\in V(H)} y_v\geq m(H) &\forall H\in \hset\\
	&z\ge 0\\
	&y_v\geq 0&\forall v\in V(G)
\end{eqnarray*}

We denote the value of the optimal solution to (LPW-D) by $\opt_{\textnormal{LPW-D}}$.

The following lemma provides a randomized separation oracle for (LPW-D). It is an analogue of \Cref{lemma: separation oracle}, and its proof is essentially identical. We provide it here for completeness.

\begin{lemma}\label{lemma: separation oracle WGP}
Assume that there is an efficient $\alpha(n)$-approximation algorithm for the \DkS problem, where $\alpha$ is an increasing function, and $n$ is the number of vertices in the input graph. Then	there is a randomized $\beta(N)$-approximate separation oracle for (LPW-D), where $N$ is the number of variables in the input graph $G$, and $\beta(N)=O(\alpha(N^2)\cdot \log^2N)$.
\end{lemma}

\begin{proof}
Recall that we are given as input real values $z$ and $\set{y_v\mid v\in V(G)}$. As before, we can efficiently check whether $z\geq 0$, and whether $y_v\geq 0$ for all $v\in V(G)$. If this is not the case, we can return the corresponding violated constraint.

We say that a subgraph $H\in \hset$ is \emph{bad} if $z+\sum_{v\in V(H)}y_v<m(H)/\beta$ holds, where $\beta=c\cdot \alpha(N^2)\cdot \log^2N$, and $c$ is a large enough constant whose value we set later.
Our goal is to design an efficient algorithm that either returns a violated constraint of the LP (that is, a graph $H\in \hset$ for which 
$z+\sum_{v\in V(H)}y_v<m(H)$ holds); or it returns ``accept''. We require that, if there exists a bad subgraph $H\in \hset$, then the probability that the algorithm returns ``accept'' is at most $2/3$.

We slightly modify the input values in $\set{y_v\mid v\in V(G)}$, almost exactly like in the proof of \Cref{lemma: separation oracle WGP}.  First, for every vertex $v\in V(G)$ with $y_v>h$, we let $y'_v$ be the smallest integral power of $2$ that is greater than $h$, and for every vertex $v\in V(G)$ with $y_v<1/4$, we set $y'_v=0$. For each remaining vertex $v$, we let $y'_v$ be the smallest integral power of $2$ that is greater than $4y_v$. Notice that, for every vertex $v$ with $y'_v\neq 0$, $1\leq y'_v\leq 4h$ holds, and $y'_v$ is an integral power of $2$. We also set $z'=2z$. We say that a subgraph $H\in \hset$ is \emph{problematic} if 
 $z'+\sum_{v\in V(H)}y'_v<8m(H)/\beta$ holds. We
use the following two observations, that are analogues of \Cref{obs: relating old and new LP values} and \Cref{obs: relating old to new 2}; their proofs are also almost identical.

\begin{observation}\label{obs: relating old and new LP values WGP}
	If $H\in \hset$ is a bad subgraph of $G$, then it is a problematic subgraph of $G$. 
\end{observation}
\begin{proof}
	Recall that, if $H$ is a bad subgraph, then $z+\sum_{v\in V(H)}y_v<m(H)/\beta$ must hold. Since, for every vertex $v\in V(G)$, $y'_v\leq 8y_v$, and $z'=2z$, we get that:
	
\[z'+\sum_{v\in V(H)}y'_v\leq 2z+8\sum_{v\in V(H)}y_v\leq 8\left(z+\sum_{v\in V(H)}y_v  \right )<8m(H)/\beta.\]
	
Therefore, subgraph $H$ is problematic.	
\end{proof}

\begin{observation}\label{obs: relating old to new 2 WGP}
	Assume that there exists a subgraph $H\in \hset$, for which $z'+\sum_{v\in V(H)}y'_v<m(H)$ holds. Let $H'\subseteq H$ be the graph obtained from $H$ after removing all isolated vertices from it. Then $z+\sum_{v\in V(H')}y_v<m(H')$ holds. 
\end{observation}

\begin{proof}
Since every vertex $v\in V(H)\setminus V(H')$ is isolated in $H$, we get that $m(H')=|E(H')|=|E(H)|=m(H)$. We partition the vertices of $H'$ into two subsets: set $X$ containing all vertices $v\in V(H')$ with $y_v<1/4$, and set $Y$ containing all remaining vertices. Clearly, $\sum_{v\in X}y_v<\frac{|X|}{4}\leq \frac{m(H')} 2$ (since $m(H')\geq |V(H')|/2\geq |X|/2$, as graph $H'$ contains no isolated vertices).

Assume for contradiction that $z+\sum_{v\in V(H')}y_v\geq m(H')$. Then:

\[ z+\sum_{v\in Y}y_v\geq m(H')-\sum_{v\in X}y_v\geq m(H')/2.\]

We now consider two cases. The first case is when there is some vertex $v\in Y$ with $y_v\geq h$. In this case, $y'_v\geq h$ holds, and $z'+\sum_{v\in S}y'_v\geq h>m(H')$ holds, a contradiction.

Otherwise,  for every vertex $v\in Y$, $y'_v\geq 4y_v$ holds. Since $z'=2z$ also holds, we get that:

\[z'+\sum_{v\in V(H)}y'_v\geq z'+\sum_{v\in Y}y'_v\geq 2z+4\sum_{v\in Y}y_v\geq m(H')=m(H),\]

a contradiction. 
\end{proof}

From now on we focus on values $z',\set{y'_v\mid v\in V(G)}$. It is now enough to design an efficient randomized algorithm, that either computes a subgraph $H\in \hset$, for which $z'+\sum_{v\in V(H)}y'_v<m(H)$ holds, or returns ``accept''. It is enough to ensure that, if there is a problematic subraph $H\in \hset$, then the algorithm returns ``accept'' with probability at most $2/3$. Indeed, if there is a bad subgraph $H\in \hset$, then, from  \Cref{obs: relating old and new LP values WGP}, there is a problematic subgraph, and the algorithm will return ``accept'' with probability at most $2/3$. On the other hand, if the algorithm computes a subgraph $H\in \hset$ of vertices, for which $z'+\sum_{v\in S}y'_v<m(S)$ holds, then we can return the subgraph $H'\subseteq H$ from the statement of \Cref{obs: relating old to new 2 WGP}, that defines a violated constraint with respect to the original LP-values.

Our algorithm is essentially the same as before: it computes a random partition $(A,B)$ of the vertices of $G$, where every vertex $v\in V(G)$ is independently added to $A$ or to $B$ with probability $1/2$ each. Let $q=\ceil{\log(8h)}$. For all $1\leq i\leq q$, we define a set $A_i\subseteq A$ of vertices: $A_i=\set{v\in A\mid y'_v=2^{i-1}}$, and we let $A_0=\set{v\in A\mid y'_v=0}$. Clearly, $(A_0,\ldots,A_q)$ is a partition of the set $A$ of vertices.

We compute  a partition $(B_0,\ldots,B_q)$ of the vertices of $B$ similarly. For all $0\leq i,j\leq  q$, we denote by $E_{i,j}$ the set of all edges $e=(u,v)$ with $u\in A_i$ and $v\in B_j$, and we define a bipartite graph $G_{i,j}$, whose vertex set is $A_i\cup B_j$, and edge set is $E_{i,j}$.

Recall that we have assumed that there is an efficient $\alpha(n)$-approximation algorithm for the \DkS problem, where $n$ is the number of vertices in the input graph. From \Cref{lem: DkS and Dk1k2S}, there exists an efficient $O(\alpha(\hat n^2))$-approximation algorithm for the \BDkS problem, where $\hat n$ is the number of vertices in the input graph.
We denote this algorithm by $\aset'$.

For every pair $0\leq i,j\leq q$ of integers, and every pair $0\leq k_1,k_2\leq N$ of integers, we apply Algorithm $\aset'$ for the \BDkS problem to graph $G_{i,j}$, with parameters $k_1$ and $k_2$. Let $S_{i,j}^{k_1,k_2}$ be the output of this algorithm, and let $m_{i,j}^{k_1,k_2}$ be the number of edges in the subgraph of $G_{i,j}$ that is induced by the set $S_{i,j}^{k_1,k_2}$ of vertices. We say that the application of algorithm $\aset'$ is \emph{successful} if $z'+\sum_{v\in S_{i,j}^{k_1,k_2}}y'_v<\min\set{h,m_{i,j}^{k_1,k_2}}$, and otherwise it is \emph{unsuccessful}. If, for any quadruple $(i,j,k_1,k_2)$ of indices, the application of algorithm $\aset'$ was successful, then we return a graph $H$, that is defined as the subgraph of $G$ induced by the set  $S=S_{i,j}^{k_1,k_2}$ of vertices; if this graph contains more than $h$ edges, then we delete arbitrary edges from it, until $|E(H)|=h$ holds. Clearly, $H\in \hset$ must hold. Moreover, we are guaranteed that $z'+\sum_{v\in V(H)}y'_v<\min\set{h,m_{i,j}^{k_1,k_2}}\leq m(H)$, as required. If every application of algorithm $\aset'$ is unsuccessful, then we return ``accept''. The following observation will finish the proof of \Cref{lemma: separation oracle WGP}. The observation is an analogue of \Cref{obs: if problematic then fail} and its proof is essentially identical.

\begin{observation}\label{obs: if problematic then fail WGP}
	Suppose there is a problematic subgraph $H\in \hset$. Then the probability that the algorithm returns ``accept'' is at most $2/3$.
\end{observation}

\begin{proof}
	Let $H\in \hset$ be a problematic subgraph, and denote $S=V(H)$, so $z'+\sum_{v\in S}y'_v<8m(S)/\beta$ holds. Let $E'=E(H)$, so $|E'|=m(H)\leq h$. 
	
	Denote $A_S=A\cap S,B_S=B\cap S$, and let $E''\subseteq E'$ be the set of edges $e$, such that exactly one endpoint of $e$ lies in $A$. Clearly, for every edge $e\in  E'$, $\prob{e\in E''}=1/2$. Therefore, $\expect{|E''|}=|E'|/2$. Let $\event'$ be the bad event that $|E''|<|E'|/8$. Using the same arguments as in the proof of \Cref{obs: if problematic then fail}, $\prob{\event'}\leq 2/3$. Next, we show that, if Event $\event'$ does not happen, then the algorithm does not return ``accept''. 
	
	From now on we assume that Event $\event'$ did not happen, so $|E''|\geq |E'|/8$. 
	Therefore: 
	
	$$z'+\sum_{v\in S}y'_v <\frac{8m(H)}{\beta}\leq \frac{64|E''|}{\beta}$$
	
	 holds.

	Clearly, there must be a pair $0\leq i,j\leq q$ of indices, such that $|E''\cap E_{i,j}|\geq \frac{|E''|}{4q^2}\geq\frac{|E''|}{128\log^2 m}$. We now fix this pair $i,j$ of indices, and denote $A'_i=A_i\cap S$ and $B'_j=B_j\cap S$. We also denote $k_1=|A'_i|$ and let $k_2=|B'_j|$.
	Denote $M_{i,j}=|E''\cap E_{i,j}|$.
	Fom our choice of indices $i,j$, we get that:

	\[ z'+\sum_{v\in A'_i\cup B'_j}y'_v\leq z'+\sum_{v\in S}y'_v\leq  \frac{64|E''|}{\beta}\leq \frac{M_{i,j}}{\beta}\cdot (2^{13}\cdot\log^2m).\]
	
	 Notice that the set $S'=A'_i \cup B'_j$ of vertices provides a solution to the instance of the \BDkS problem on graph $G_{i,j}$ with parameters $k_1,k_2$, whose value is at least $M_{i,j}$. Let $S''=S_{i,j}^{k_1,k_2}$ be the set of vertices obtained by applying Algorithm $\aset'$ to graph $G_{i,j}$ with parameters $k_1,k_2$. Since $|V(G_{i,j})|\leq N$, and since $\aset'$ is an $O(\alpha(N^2))$-approximation algorithm for \BDkS, we are guaranteed that $|E_G(S'')|\geq \Omega\left(\frac{M_{i,j}}{\alpha(N^2)}\right )$. Recall that $|A\cap S''|\leq k_1$; $A\cap S''\subseteq A_i$, and all vertices  $v\in A_i$ have an identical value $y'_v$. Therefore, $\sum_{v\in A\cap S''}y'_v\leq \sum_{v\in A'_i}y'_v$. Using a similar reasoning, 
	 $\sum_{v\in B\cap S''}y'_v\leq \sum_{v\in B'_j}y'_v$. Overall, we then get that:

	 \[
	 \begin{split}
	 z'+\sum_{v\in S''}y'_v& \leq z'+\sum_{v\in S'}y'_v\\
	 &\leq   \frac{M_{i,j}}{\beta}\cdot (2^{13}\cdot\log^2m)\\
	 &\leq O\left (\frac{\alpha(N^2)\cdot 2^{13}\cdot \log^2m}{\beta}\right )\cdot \min\set{|E_G(S'')|,h}.
	 \end{split}\]

Recall that $\beta=c\cdot \alpha(N^2)\cdot \log^2N$. By letting the value of the constant $c$ be large enough, we can ensure that $z'+\sum_{v\in S''}y'_v<\min\set{|E_G(S'')|,h}$, and so the application of algorithm $\aset'$ to graph $G_{i,j}$ with parameters $k_1$ and $k_2$ is guaranteed to be successful. Therefore, if Event $\event'$ does not happen, and we set $c$ to be a large enough constant, then our algorithm does not return ''accept''. Since $\prob{\event'}\leq 2/3$, the observation follows.
\end{proof}
\end{proof}

We can use the separation oracle described in \Cref{lemma: separation oracle WGP} in order to obtain a $\beta(N)$-approximate solution to (LPW-P), whose support size is bounded by $O(\poly(N))$ using the standard techniques that were described in \Cref{subsec: solve the LP}; we do not repeat them here. By applying the LP-rounding algorithm from \Cref{claim: LP-rounding WGP} to the resulting LP-solution, with high probability we obtain, in time $O(\poly(N))$,  an integral solution $(H_1,\ldots,H_r)$ to instance $\pwgp(G,r,h)$, such that $\sum_{i=1}^{r}|E(H_i)|\geq \Omega\left(\frac{\optwgp(G,r,h)}{300\beta(N) \log^3N}\right )$.  Since $\beta(N)=O(\alpha(N^2)\cdot \log^2N)$, with high probability we obtain an $O(\alpha(N^2)\cdot \poly\log N)$-approximate solution to the input instance of \WGP.

\bibliographystyle{alpha}

\bibliography{REF}

\end{document}